\documentclass[12pt,a4paper,openright,titlepage,twoside,final]{report}

\usepackage{amssymb,amsbsy,times,fancyhdr,color}
\usepackage{amsmath}
\usepackage[dvips]{epsfig,graphics}
\usepackage{graphicx}
\usepackage{texdraw}
\usepackage{color}
\usepackage{latexsym}
\usepackage{mathdots}
\usepackage{makeidx}
\usepackage{nomencl}
\makenomenclature
\renewcommand{\nomname}{List of abbreviations}

\usepackage{etoolbox}

\makeindex

\newcommand{\field}[1]{\mathbb{#1}}
\newcommand{\tr}{\mbox{tr}}
\newcommand{\diag}{\mbox{diag}}
\DeclareMathOperator{\rank}{rank}
\DeclareMathOperator{\End}{End}
\newcommand{\str}{\mbox{str}}
\newcommand{\sdet}{\mbox{sdet}}

\fancypagestyle{plain}{
 \fancyhf{}
 \fancyhead[R]{\thepage}
 
}
\renewcommand{\theequation}{\thesection.\arabic{equation}}
\newcounter{alphabetical}
\newcounter{alphappendix}

\newenvironment{proof}{\noindent\textbf{Proof}\\}{\noindent$\Box$\\}

\newtheorem{lemma}{Lemma}[section]
\newtheorem{theorem}[lemma]{Theorem}
\newtheorem{corollary}[lemma]{Corollary}
\newtheorem{definition}[lemma]{Definition}
\newtheorem{note}[lemma]{Note}
\newtheorem{remark}[lemma]{Remark}

\newtheorem{example}[lemma]{Example}
\newtheorem{proposition}[lemma]{Proposition}

\begin{document}
\cleardoublepage


\begin{titlepage}
\centering\Huge
Darboux Transformations, Discrete Integrable Systems   \\and Related Yang-Baxter Maps\\
\Large
\vspace{1 cm}
\textcolor{blue}{Sotiris Konstantinou-Rizos}\\
\vspace{1.5cm}

\begin{figure}[h]
	\centering
	\includegraphics[width=0.4\textwidth,height=0.3\textheight,keepaspectratio=true,scale=0.8]{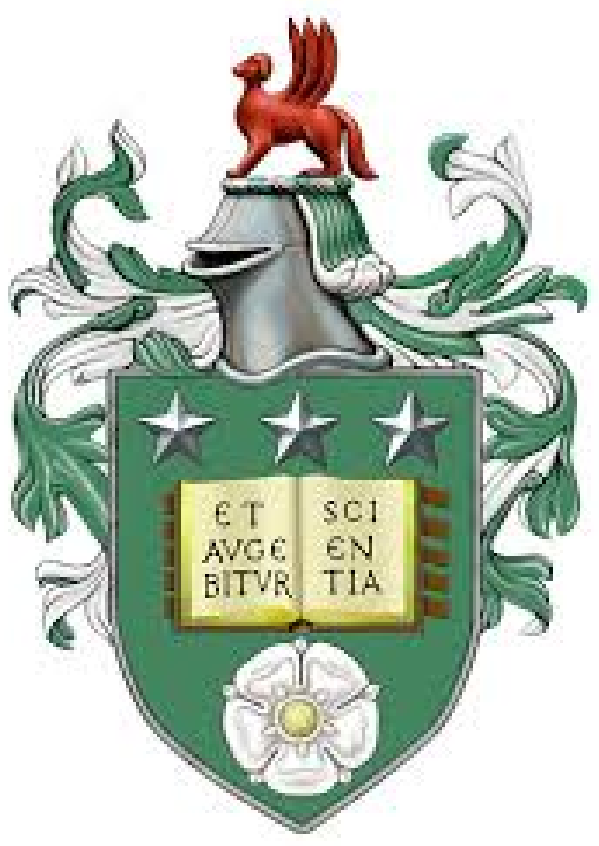}
\end{figure}

\Large
\vspace{1.3cm}

\textcolor{blue}{Department of Applied Mathematics} \\
\textcolor{blue}{University of Leeds} \\
\large
\vspace{1.4cm}
Submitted in accordance with the requirements for the degree of Doctor of Philosophy\\
\vspace{0.6cm}
April 2014\\

\pagenumbering{gobble}
\cleardoublepage
\large
\vspace{0.8cm}
The candidate confirms that the work submitted is his own, except where work which has formed part of jointly-authored publications has been included. The contribution of the candidate and the other authors to this work has been explicitly indicated below. The candidate confirms that appropriate credit has been given within the thesis where reference has been made to the work of others.\\
\vspace{0.5cm}
\normalsize
\begin{flushleft}
\begin{itemize}
	\item Chapter 5 is based on S. Konstantinou-Rizos and A.V. Mikhailov, ``Darboux trans-\\
	formations, finite reduction groups and related Yang-Baxter maps'', Journal of Physics A \textbf{46} (2013), 425201.
\end{itemize}

The construction of Yang-Baxter maps by matrix refactorisation problems of Darboux matrices was my original idea which we accomplished with the co-author, A.V. Mikha-\\ilov, for all the cases of finite reduction groups with two-dimensional representation. The contribution was equal. Moreover, I wrote the first draft of the paper which, after the co-author's amendments, was brought to a publishable form.
\end{flushleft}

\vspace{3.0cm}
\large
This copy has been supplied on the understanding that it is copyright material and that no quotation from the thesis may be published without proper acknowledgement.

\vspace{3.0cm} \copyright \, 2014 The University of Leeds and Sotiris Konstantinou-Rizos
\normalsize
\end{titlepage}
\vspace{10cm}

\cleardoublepage 
\vspace*{\stretch{1}}
\begin{center}
\Large{To my parents, Anastasia and Giorgos}
\end{center}
\vspace*{\stretch{2}}


\cleardoublepage
\centerline{\Huge{Acknowledgements}}

Doing a PhD is a hard job, but at the same time it can become a pleasant experience if you have the right people around you. Therefore, from this position, I would like to thank all the people who contributed --one way or another-- to making it happen.

Above all, I would like to thank my supervisor, Alexander Mikhailov, for his guidance and for being patient throughout my PhD studies. He has the experience and ability to tell you straightaway if your ideas can work, which is very useful for a research student. Moreover, he encouraged me to participate in many conferences from which I benefited a lot. His constructive criticism on my research made me mature as a mathematician and for that I am most grateful. At this point, I would also like to thank my adviser, Frank Nijhoff, for his interest and for setting the guidelines in our yearly meetings.

I am greatly indebted to the University of Leeds for granting me the William Wright Smith scholarship and to J.E. Crowther for the scholarship-contribution to fees. Moreover, I extend grateful thanks to the school of Mathematics for the financial support during my writing-up period and for financing all the conferences in which I participated; I feel that it was the ideal enviroment for postgraduate studies, in a very well-organised University which I would highly recommend to any graduate student. It was also very beneficial to be part of the integrable systems group, with weekly seminars and many other research activities. Many thanks are also due to Alastair Rucklidge and Jeanne Shuttleworth for welcoming me here in the school and for being very helpful over these years.

However, my postgraduate studies started in Greece, at the Univeristy of Patras, and therefore there are several people who contributed to building my maths background.

First of all, I was very lucky to do my Master's thesis under the supervision of Dimitris Tsoubelis. He is the person who introduced me to the theory of PDEs and, in general, to Mathematical Physics, and he has been an enormous source of knowledge and experience for me ever since. But more importantly, he taught me how to study mathematics myself, how to write mathematical texts and research proposals and several other things which I use daily in my research. Working with him is a pleasure, and he is a person who I always seek advice from at any stage of my career. Therefore, I would like to thank him for supporting my studies in several ways. He and his wife, Eudokia Papafotopoulou, have been very kind to me and for that I am most grateful.

I would also like to thank Vassilis Papageorgiou for his interest throughout these years and for contributing into building my theoretical background during both my undergraduate and postgraduate studies. He is always happy to discuss mathematics with me and he is always willing to help me. Moreover, he and his wife, Eug\'enie Foustoucos, had helped me to come here in Leeds in the beginning of my studies for which I am very grateful as well as for their continuous interest in my progress.

Moreover, I would not be able to start my PhD studies without the help of Tolis Fasois who prepared me for the IELTS exam. His job was difficult, as there were only a few days before the exams; yet, he managed to improve my English in no time and I would like to thank him from this position.

The person who introduced me into the theory of Yang-Baxter maps --which covers a big part of this thesis-- is Thodoris Kouloukas. We have had endless mathematical conversations and I would like to thank him for that, as well as for being a very good friend over these years. Many thanks are also due to my ex-neighbour and friend Georgi Grahovski for explaining to me what a Grassmann algebra is and for being very helpful during his postdoctoral stay in the University of Leeds. Moreover, I would like to thank my academic brother and friend, George Berkeley, for reading my thesis and making comments on it. It has been very nice working with both of them. It was also very beneficial to have several maths discussions with Pavlos Kassotakis, thus I would like to thank him for that.

Pavlos Xenitidis and Anastasia Kazaltzi have been my family here in Leeds. I met them in Patras when Pavlos was my academic tutor and they have been more than great friends ever since. I came to Leeds one year after they moved here when I decided to ``follow'' them. Living in Leeds without them would be completely different, as they made me feel as if I were home when I was away from it. I would like to thank them for everything from this position.

I was also very lucky, since Ilia Roustemoglou decided to join us one year after, and also Giorgos Papamikos a bit after her. They are both long-time friends of mine and it was a pleasure to have them around. Therefore, I would like to thank them for the nice time we had in the UK. Additionally, I would like to thank Giorgos for proof-reading my thesis.

My postgraduate studies would never be the same without the support of my friends Eleni Christodoulidi, Spiros Dafnis, Stelios ``Spawn'' Dimas, Nikos Kallinikos, Chrysavgi Kostopoulou, Akis Matzaris, Mitsos Nomikos, Nikitas Nikandros, Grigoris Protsonis, Panagiotis Protsonis and Stavros Anastasiou. At this point, I should extend many thanks to Panagis Karazeris for being patient and friendly to all of us. Additionally, I would also like to thank Giota Adamopoulou and Maura Capuzzo who are the ``new entries'' in our friends list. I would like to thank all the former for having an amazing time throughout these years (and they are quite a few!) and for making the daily routine a pleasant journey. Moreover, I had a really nice time with my PhD-fellows and friends Laurence Hawke, Vijay Teeluck, Julia Sauter, Lamia Alqahtani, Neslihan Delice, Abeer Al-Nahdi, Huda Alshanabari, Bartosz Szczesny, Gareth Hurst and all the rest of the people in the satellite postgraduate office. And, of course, special thanks to Ms Charikleia for cooking for all of us. Finally, living in a shared flat would never be the same without the pleasant company of my flatmates and very good friends Sophia-Marie Honny, Priya Sinha and Moa N\"asstr\"om.

Special thanks are due to my long-time friends Christos Fokas, Giorgos Panagiotidis in Athens and Costantinos Balamoshev and Ewa Golic in Warsaw, who I always think of no matter how far away they are.

Finally, I am more than grateful to my parents Anastasia and Giorgos for supporting everything I choose to do in my life. Moreover, special thanks to my uncle and aunt, Charilaos and Aleka, for their interest and support and to all the rest of the family but especially my twin brother Michael and my cousins and friends Costas, Stelios, Zoe and Marianna who I grew up with.

\cleardoublepage

\centerline{\Huge{Abstract}}
\addcontentsline{toc}{section}{Abstract}
Darboux transformations constitute a very important tool in the theory of integrable systems. They map trivial solutions of integrable partial differential equations to non-trivial ones and they link the former to discrete integrable systems. On the other hand, they can be used to construct Yang-Baxter maps which can be restricted to completely integrable maps (in the Liouville sense) on invariant leaves.

In this thesis we study the Darboux transformations related to particular Lax operators of NLS type which are invariant under the action of the so-called reduction group. Specifically, we study the cases of: 1) the nonlinear Schr\"odinger equation (with no reduction), 2) the derivative nonlinear Schr\"odinger equation, where the corresponding Lax operator is invariant under the action of the $\field{Z}_2$-reduction group and 3) a deformation of the derivative nonlinear Schr\"odinger equation, associated to a Lax operator invariant under the action of the dihedral reduction group. These reduction groups correspond to recent classification results of automorphic Lie algebras.

We derive Darboux matrices for all the above cases and we use them to construct novel discrete integrable systems together with their Lax representations. For these systems of difference equations, we discuss the initial value problem and, moreover, we consider their integrable reductions. Furthermore, the derivation of the Darboux matrices gives rise to many interesting objects, such as B\"acklund transformations for the corresponding partial differential equations as well as symmetries and conservation laws of their associated systems of difference equations.

Moreover, we employ these Darboux matrices to construct six-dimensional Yang-Baxter maps for all the afore-mentioned cases. These maps can be restricted to four-dimensional Yang-Baxter maps on invariant leaves, which are completely integrable; we also consider their vector generalisations.

Finally, we consider the Grassmann extensions of the Yang-Baxter maps corresponding to the nonlinear Schr\"odinger equation and the derivative nonlinear Schr\"odinger equation. These constitute the first examples of Yang-Baxter maps with noncommutative variables in the literature.

\pagenumbering{roman}
\cleardoublepage
\addcontentsline{toc}{section}{Contents}
\tableofcontents

\clearpage
\addcontentsline{toc}{section}{List of abbreviations}
 \printnomenclature
\markboth{List of abbreviations}{}
\markboth{\nomname}{\nomname}

\cleardoublepage
\addcontentsline{toc}{section}{List of figures}
\markboth{List of figures}{}
\listoffigures
\cleardoublepage

\pagenumbering{arabic}

\chapter{Introduction}
\label{chap1}
\setcounter{equation}{0}
\renewcommand{\theequation}{\thechapter.\arabic{equation}}

The aim of this chapter is to give an introduction to the subject of integrable systems, which forms the context of this thesis. Integrable systems arise in nonlinear processes and, both in their classical and quantum version, have many applications in various fields of mathematics and physics.

However, the definition of integrable systems is itself highly nontrivial; many scientists have different opinions on what ``integrable'' should mean, which makes the definition of integrability elusive, rather than tangible. In fact, a comprehensive definition of integrability is not yet available. As working definitions we often use the existence of a Lax pair, the solvability of the system by the IST, the existence of infinitely many symmetries or conservation laws, or the existence of a sufficient number of first integrals which are in involution (Liouville integrability); there is even a book entirely devoted to \textit{what is integrability} \cite{Zakharov1991}.

In this thesis we are interested in the derivation of discrete integrable systems and Yang-Baxter maps, from (integrable) PDEs which admit Lax representation, via Darboux transformations. Specifically, we shall be focusing on particular PDEs of NLS type whose corresponding Lax operators possess certain symmetries, due to the action of the so-called reduction group.

Since these AKNS-type Lax operators\index{Lax operator(s)!AKNS type} we are dealing with constitute a key role in the integrability of their associated equations under the IST, the inverse scattering method and the AKNS scheme deserve a few pages in the first part of this introduction. However, we will skip the technical parts of their methods, as it is not the aim of this thesis. For detailed information on the methods and the historical review of the results, we indicatively refer to \cite{Ablo-Segur, Ablo-Clarkson, DJ} (and the references therein).

The second part of this chapter is devoted to a brief introduction to the integrability of discrete systems; in the main, their multidimensional consistency\index{multidimensional consistency} and some recent classification results.

\section{Lax representations and the IST}
The inverse scattering transformation (or just transform)\nomenclature{IST}{Inverse scattering transform}\index{inverse scattering transform} is a method for solving nonlinear PDEs. Its name is due to the main idea of the method, namely the recovery of the time evolution of the potential solution of the nonlinear equation, from the time evolution of its scattering data.\index{scatering data} As a matter of fact, the method of the inverse scattering transform\index{inverse scattering transform} is of the same philosophy as the Fourier transform\index{Fourier transform} technique for solving linear PDEs; actually, the IST is also found in the literature as the nonlinear Fourier transform.\index{nonlinear Fourier transform} However, it does not apply to all nonlinear equations in a systematic way. 

The first example of nonlinear PDE solvable by the IST\index{inverse scattering transform} method, is the KdV equation,\index{KdV equation} namely 
\begin{equation}\label{KdV}
u_t=6uu_x-u_{xxx},\quad u=u(x,t),
\end{equation}
which is undoubtedly the most celebrated nonlinear PDE over the last few decades. It mostly owes its popularity to Gardner, Greene, Kruskal and Miura,\nomenclature{GGKM}{Gardner, Greene, Kruskal and Miura} who were the first to derive the exact solution of the Cauchy problem\index{Cauchy problem} for the KdV equation,\index{KdV equation} for rapidly decaying initial values, in late sixties \cite{GGKM}. However, equation (\ref{KdV}) was derived by Diederik Korteweg and Gustav de Vries in 1895, as a mathematical model of water-waves in shallow channels. In fact, they showed that the KdV equation represents Scott Russel's solitary wave, known as \textit{soliton}\index{soliton} (see \cite{DJ} for details). The name ``soliton''\index{soliton} was given by Zabusky and Kruskal\footnote{They initially called it ``solitron'', but at the same time a company was trading with the same name and therefore had to remove the ``r''.} in 1965, when they discovered numerically that these wave solutions behave like particles; they retain their amplitude and speed after collision.

The work of GGKM in 1967, namely the IST\index{inverse scattering transform} method, is probably one of the most significant results of the last century in the theory of nonlinear PDEs. It is not only a technique for solving the initial value problem for KdV,\index{KdV equation!initial value problem} but it also initiated a more general scheme applicable to other nonlinear PDEs. In fact, P. Lax was the one who contributed in this direction, formulating a more general framework a year later in \cite{Lax}. 

\subsection{Lax representations}
Lax's generalisation concerns nonlinear \textit{evolution equations},\index{evolution equation(s)} namely equations of the form
\begin{equation}\label{EE}
u_t=N(u),\quad u=u(x,t),
\end{equation}
where $N$ is a nonlinear differential operator, which does not depend on $\partial_t$.

In particular, Lax considered a pair of linear differential operators, $\mathcal{L}$ and $\mathcal{A}$. Operator $\mathcal{L}$ is associated to the following spectral problem of finding eigenvalues and eigenfunctions
\begin{subequations}\label{Lax_pair}
\begin{align}
\mathcal{L}\psi &=\lambda\psi,\quad \psi=\psi(x,t)\label{Lax_pair-x}
\intertext{while $\mathcal{A}$ is the operator related to the time evolution of the eigenfunctions}
\psi_t &=\mathcal{A}\psi.\label{Lax_pair-t}
\end{align}
\end{subequations}

\begin{proposition}(Lax's equation)
If the spectral parameter does not evolve in time, namely $\lambda_t=0$, then relations (\ref{Lax_pair}) imply
\begin{equation}\label{Lax's_equation}
\mathcal{L}_t+\left[\mathcal{L},\mathcal{A}\right]=0,
\end{equation}
where $\left[\mathcal{L},\mathcal{A}\right]:=\mathcal{L}\mathcal{A}-\mathcal{A}\mathcal{L}$.
\end{proposition}
\begin{proof}
Differentiation of (\ref{Lax_pair-x}) with respect to $t$ implies
\begin{equation}
\mathcal{L}_t\psi+\mathcal{L}\psi_t=\lambda\psi_t.
\end{equation}
Using both relations (\ref{Lax_pair}), the above equation can be rewritten as
\begin{equation}
\left(\mathcal{L}_t+\mathcal{L}\mathcal{A}-\mathcal{A}\mathcal{L}\right)\psi=0.
\end{equation}
Now, since the above holds for the arbitrary eigenfunction $\psi(x,t)$, it implies equation (\ref{Lax's_equation}).
\end{proof}

If a nonlinear evolution equation (or a system of equations) of the form (\ref{EE}) is equivalent to (\ref{Lax's_equation}), then we can associate to it a pair of linear operators as (\ref{Lax_pair}). In this case, equation (\ref{Lax's_equation}) is called the \textit{Lax equation}\index{Lax equation}, while equations (\ref{Lax_pair}) constitute a \textit{Lax representation}\index{Lax representation} or, simply, a \textit{Lax pair}\index{Lax pair} for (\ref{EE}). In particular, equation (\ref{Lax_pair-x}) is called the \textit{spatial part} of the Lax pair\index{Lax pair!spatial part ($x$-part)} (or $x$-part), while equation (\ref{Lax_pair-t}) is called its \textit{temporal part} (or $t$-part).\index{Lax pair!temporal part ($t$-part)}

The property of a nonlinear evolution equation\index{evolution equation(s)} to be written as a compatibility condition of a pair of linear equations (\ref{Lax's_equation}) plays a key role towards the solvability of the equation under the IST,\index{inverse scattering transform} and it is usually used as an integrability criterion.

\begin{remark}\normalfont
For a given nonlinear evolution equation\index{nonlinear evolution equation(s)} (\ref{EE}) there is no systematic method of writing it as a compatibility condition of a pair of linear equations, namely to determine operators $\mathcal{L}$ and $\mathcal{A}$. In fact, the usual procedure is to first study differential operators of certain form, and then to examine what kind of PDEs result from their compatibility condition.
\end{remark}

\begin{example}\normalfont
The KdV equation\index{KdV equation} (\ref{KdV}) can be written as a compatibility condition of the form (\ref{Lax's_equation}), of a system of linear equations (\ref{Lax_pair}), where $\mathcal{L}$ and $\mathcal{A}$ are given by
\begin{subequations}\label{Lax_KdV}
\begin{align}
\mathcal{L}&=-\partial_x^2+u,\quad u=u(x,t),\label{Lax_KdV-x}\\
\mathcal{A}&=-4\partial_x^3+3u\partial_x+3u_x.\label{Lax_KdV-t}
\end{align}
\end{subequations}
Operators $\mathcal{L}$ and $\mathcal{A}$ constitute a Lax pair for the KdV equation.\index{Lax pair!for KdV equation}\index{KdV equation!Lax pair}
\end{example}

Operator (\ref{Lax_KdV-x}) is the so-called Schr\"odinger operator\index{Schr\"odinger operator}\index{Schr\"odinger operator} and the corresponding equation $\mathcal{L}\psi=\lambda\psi$ is the time-independent Schr\"odinger\index{time-independent Schr\"odinger equation} equation, which constitutes a fundamental equation in mathematical physics since the first quarter of the $20^{th}$ century. However, the Lax pair (\ref{Lax_KdV}) for the KdV equation\index{KdV equation!Lax pair}\index{Lax pair!for KdV equation} was not derived from the equation itself. As a matter of fact, the ``guess'' of the operator (\ref{Lax_KdV-x}) was inspired by the desire to link the KdV equation\index{KdV equation} with Schr\"odinger's equation.

We will come back to this Schr\"odinger equation in the next chapter, where we shall study its covariance under the so-called Darboux transformation.\index{Darboux transformation(s)}

\subsection{The inverse scattering transform}
Although so far the method of the inverse scattering transform\index{inverse scattering transform} is not yet formulated to be uniformly applicable to all nonlinear evolution equations,\index{evolution equation(s)} it always consists of three basic steps. We briefly explain these steps, and we also present them schematically in Figure \ref{ISTscheme}.

Consider the following Cauchy problem
\begin{equation}\label{CauchyNEE}
u_t=N(u),\quad u(x,0)=f(x),\quad u:=u(x,t),
\end{equation}
for a nonlinear evolution equation\index{evolution equation(s)}. Let us also assume that the above PDE admits Lax representation\index{Lax representation} (\ref{Lax_pair}).

\subsubsection{Step $I$: The direct problem}
The direct problem consists of finding the scattering transformation at a fixed value of the temporal parameter, say $t=0$, by using the initial condition $u(x,0)=f(x)$. That is to find the spectral data of operator $\mathcal{L}$, which are called the \textit{scattering data}.\index{scattering data} The scattering transform at $t=0$ is nothing but a set of scattering data, which we denote $S(u)|_{t=0}$.

\subsubsection{Step $II$: Time evolution of the scattering data}
This is the part where one needs to determine the scattering data at an arbitrary time $t\in\field{R}$, i.e. given $S(u)|_{t=0}$, use the second equation of (\ref{Lax_pair}) to determine $S(u)|_{t\in\field{R}}$. The significance of this part lies in the fact that we are now dealing with a linear problem,  (\ref{Lax_pair-t}), rather than a nonlinear one as the original.

\subsubsection{Step $III$: The inverse problem}
Analogously to the Fourier transform\index{Fourier transform} method, the final step is to recover $u=u(x,t)$ from $S(u)|_{t\in\field{R}}$.

\begin{figure}[ht]
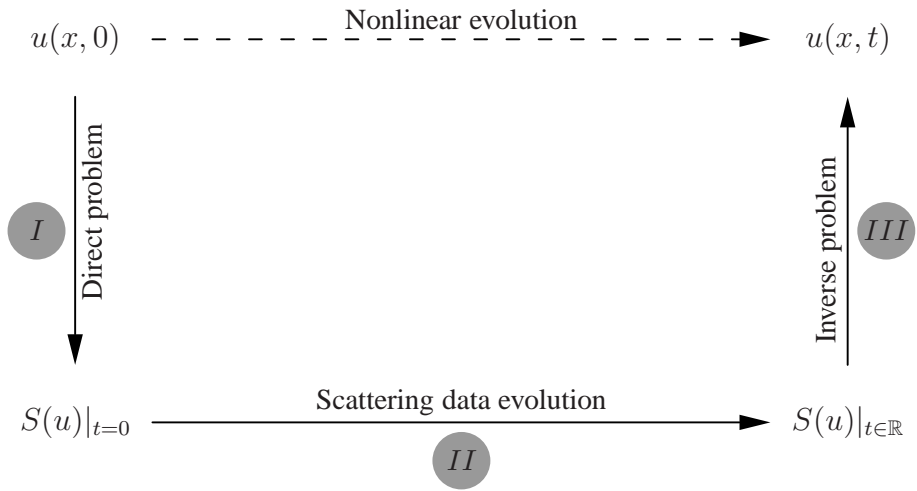

\centering
\centertexdraw{ 
\move (-1.6 0)  \arrowheadtype t:F \avec(1.6 0)
\move (-2 1.7)  \arrowheadtype t:F \avec(-2 0.3)
\move (2 0.3)  \arrowheadtype t:F \avec(2 1.7) 
\textref h:C v:C \htext(-2 2){$u(x,0)$}
\textref h:C v:C \htext(2 2){$u(x,t)$}
\textref h:C v:C \htext(-2 0){$S(u)|_{t=0}$}
\textref h:C v:C \htext(2 0){$S(u)|_{t\in \field{R}}$}\lpatt(0.067 0.1)
\move (-1.6 2)  \arrowheadtype t:F \avec(1.6 2)
\move (-2.2 1)\fcir f:.6 r:.15
\textref h:C v:C \small{\htext(-2.2 1){$I$}}
\move (2.2 1)\fcir f:.6 r:.15
\textref h:C v:C \small{\htext(2.2 1){$III$}}
\move (0 -.2)\fcir f:.6 r:.15
\textref h:C v:C \small{\htext(0 -.2){$II$}}
\textref h:C v:C \small{\htext(0 2.1){Nonlinear evolution}}
\textref h:C v:C \small{\htext(0 .1){Scattering data evolution}}
\textref h:C v:C \small{\vtext(-1.9 1.1){Direct problem}}
\textref h:C v:C \small{\vtext(1.9 .9){Inverse problem}}
}
\caption{IST scheme}\label{ISTscheme}
\end{figure}

\subsection{The AKNS scheme}
In 1971 Zakharov and Shabat \cite{ZS} applied the inverse scattering transform\index{inverse scattering transform} method to solve the NLS equation, introducing a more general formulation than Lax's. Specifically, they introduced a pair of linear equations, namely
\begin{subequations}\label{ZS_pair}
\begin{align}
\partial_x\psi &=\mathcal{L}\psi,\quad \psi:=\psi(x,t),\\
\partial_t\psi &=\mathcal{T}\psi,
\end{align}
\end{subequations}
where $\mathcal{L}=\mathcal{L}(x,t;\lambda)$ and $\mathcal{T}=\mathcal{T}(x,t;\lambda)$ are $2\times 2$ matrices. They showed that the NLS equation,\index{nonlinear Schr\"odinger (NLS)!equation}
\begin{equation}
p_t=p_{xx}+4p^2q,\qquad q_t=-q_{xx}-4pq^2,
\end{equation}
can be written as a compatibility condition, $\psi_{xt}=\psi_{tx}$, of the system of linear equations (\ref{ZS_pair}), where $\mathcal{L}$ and $\mathcal{T}$ are given by
\begin{equation}\label{AKNS-type}
\mathcal{L}=D_x+U,\quad \mathcal{T}=D_t+V,
\end{equation}
and $U$ and $V$ by
\begin{subequations}
\begin{align}
U&=\lambda \sigma_3+\left(\begin{array}{cc} 0 & 2p \\ 2q & 0\end{array}\right),\quad \sigma_3:=\diag(1, -1),\\
V&=\lambda^2\sigma_3+\lambda \left(\begin{array}{cc} 0 & 2p \\ 2q & 0\end{array}\right)+\left(\begin{array}{cc} -2pq & p_x \\ -q_x & 2pq\end{array}\right).
\end{align}
\end{subequations}

A year later, Ablowitz, Kaup, Newell and Segur in \cite{AKNS}, motivated by Zakharov and Shabat's result, solved the sine-Gordon equation and they generalised this method to cover a wider number of nonlinear PDEs (see \cite{AKNS2}). In the rest of this thesis, we shall refer to operators of the form (\ref{AKNS-type}) as \textit{Lax operators of AKNS-type}.\index{Lax operator(s)!AKNS type}

\section{Discrete integrable systems}
Discrete systems, namely systems with their independent variables taking discrete values, are of particular interest and have many applications in several sciences as physics, biology, financial mathematics, as well as several other branches of mathematics, since they are essential in numerical analysis. Initially, they were appearing as discretisations of continuous equations, but now discrete integrable systems,\index{integrable system(s)!discrete} and in particular those defined on a two-dimensional lattice, are appreciated in their own right from a theoretical perspective.


The study of discrete systems and their integrability earned its interest in late seventies; Hirota studied particular discrete systems in 1977, in a series of papers \cite{HirotaKdV, HirotaToda, HirotaSG, HirotaTBT} where he derived discrete analogues of many already famous PDEs. In the early eighties, semi-discrete and discrete systems started appearing in field-theoretical models in the work of Jimbo and Miwa; they also provided a method of generating discrete soliton equations \cite{Jimbo-Miwa,Jimbo-Miwa-II,Jimbo-Miwa-III,Jimbo-Miwa-IV,Jimbo-Miwa-V}. Shortly after, Ablowitz and Taha in a series of papers \cite{Ablowitz-I, Ablowitz-II, Ablowitz-III} are using numerical methods in order to find solutions for known integrable PDEs, using as basis of their method some partial difference equations, which are integrable in their own right. Moreover, Capel, Nijhoff, Quispel and collaborators provided some of the first systematic tools for studying discrete integrable systems and, in particular, for the direct construction of integrable lattice equations (we indicatively refer to \cite{Capel-Nijhoff,Capel-Nijhoff-II}); that was a starting point for new systems of discrete equations to appear in the literature.

In 1991 Grammaticos, Papageorgiou and Ramani proposed the first discrete integrability test, known as \textit{singularity confinement}\index{singularity confinement} \cite{GRP}, which is similar to that of the Painlev\'e property for continuous integrability. However, as mentioned in \cite{Gramm-Schw-Tam}, it is not sufficient criterion for predicting integrability, as it does not furnish any information about the rate of growth of the solutions of the discrete integrable system.

As in the continuous case, the usual integrability criterion being used for discrete systems is the existence of a Lax pair. Nevertheless, a very important integrability criterion is that of the 3D-\textit{consistency}\index{3D-consistency} and, by extension, the \textit{multidimensional consistency}.\index{multidimensional consistency} This was proposed independently by Nijhoff in 2001 \cite{Frank4} and Bobenko and Suris in 2002 \cite{Bobenko-Suris}.

In what follows, we briefly explain what is the 3D-consistency\index{3D-consistency} property and we review some recent classification results. For more information on the integrability of discrete systems we refer to \cite{Frank5} which is one of the few self-contained monographs, as well as \cite{Gramm-Schw-Tam} for a collection of results.

\subsection{Equations on Quad-Graphs: 3D-consistency}
Let us consider a discrete equation of the form
\begin{equation}\label{QGeq}
Q(u,u_{10},u_{01},u_{11};a,b)=0,
\end{equation}
where $u_{ij}$, $i,j=0,1$, $u\equiv u_{00}$, belong in a set $\mathcal{A}$ and the parameters $a,b\in\field{C}$. Moreover, we assume that (\ref{QGeq}) is uniquely solvable for any $u_i$ in terms of the rest. We can interpret the fields $u_i$ to be attached to the vertices of a square as in Figure \ref{3Dconsistency}-(a).

If equation (\ref{QGeq}) can be generalised in a consistent way on the faces of a cube, then it is said to be \textit{3D-consistent}. In particular, suppose we have the initial values $u$, $u_{100}$, $u_{010}$ and $u_{001}$ attached to the vertices of the cube as in Figure \ref{3Dconsistency}-(b). Now, since equation (\ref{QGeq}) is uniquely solvable, we can uniquely determine values $u_{110}$, $u_{101}$ and $u_{011}$, using the bottom, front and left face of the cube. Then, there are three ways to determine value $u_{111}$, and we have the following.

\begin{definition}
If for any choice of initial values $u$, $u_{100}$, $u_{010}$ and $u_{001}$, equation $Q=0$ produces the same value $u_{111}$ when solved using the left, back or top face of the cube, then it is called 3D-consistent.\index{3D consistent equation(s)}
\end{definition}

\begin{figure}[ht]
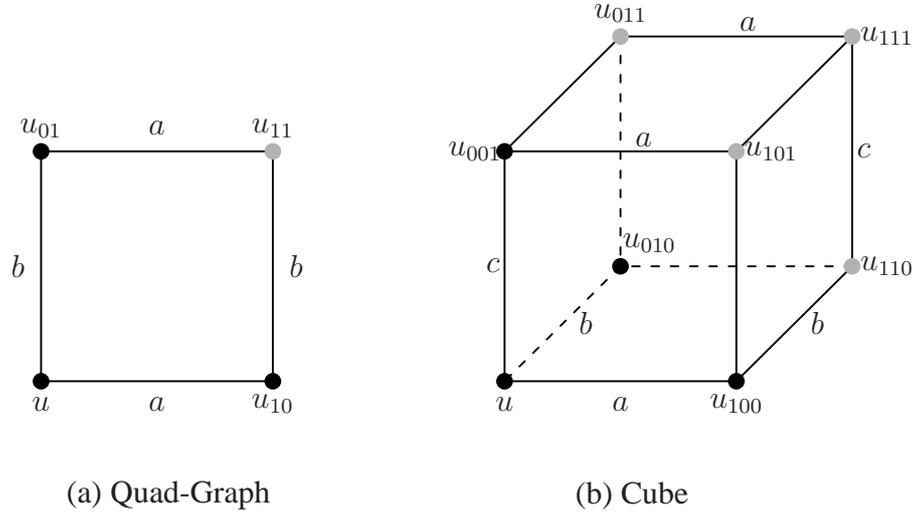

\centering
\centertexdraw{ 
\setunitscale 0.6
\move (-2 0)  \lvec(0 0) \lvec(0 2) \lvec(-2 2) \lvec(-2 0)
\move (-2 0) \fcir f:0.0 r:0.075 \move (0 0) \fcir f:0.0 r:0.075
\move (0 2) \fcir f:.7 r:0.075 \move (-2 2) \fcir f:0.0 r:0.075
\textref h:C v:C \htext(-2 -0.2){$u$}
\textref h:C v:C \htext(0 -0.2){$u_{10}$}
\textref h:C v:C \htext(-2 2.2){$u_{01}$}
\textref h:C v:C \htext(0 2.2){$u_{11}$}
\textref h:C v:C \htext(-1 -0.2){$a$}
\textref h:C v:C \htext(-1 2.2){$a$}
\textref h:C v:C \htext(-2.2 1){$b$}
\textref h:C v:C \htext(0.2 1){$b$}

\lpatt()
\setgray 0
\move (2 0)  \lvec(4 0) \lvec(5 1) 
\lpatt(0.067 0.1) \lvec(3 1) \lvec(2 0) 
\lpatt() \lvec(2 2) \lvec(3 3) 
\lpatt (0.067 0.1) \lvec(3 1) 
\lpatt() \move (3 3) \lvec(5 3) \lvec(4 2) \lvec(4 0) 
\move (2 2) \lvec(4 2)
\move (5 3) \lvec(5 1)
\move (2 0) \fcir f:0.0 r:0.075 \move (4 0) \fcir f:0.0 r:0.075
\move (2 2) \fcir f:0.0 r:0.075 \move (4 2) \fcir f:0.7 r:0.075
\move (3 1) \fcir f:0.0 r:0.075 \move (5 1) \fcir f:0.7 r:0.075
\move (3 3) \fcir f:0.7 r:0.075 \move (5 3) \fcir f:0.7 r:0.075
\textref h:C v:C \htext(2 -.2){$u$}
\textref h:C v:C \htext(4 -.2){$u_{100}$}
\textref h:C v:C \htext(5.3 1){$u_{110}$}
\textref h:C v:C \htext(3.25 1.2){$u_{010}$}
\textref h:C v:C \htext(1.745 2){$u_{001}$}
\textref h:C v:C \htext(4.3 2){$u_{101}$}
\textref h:C v:C \htext(5.3 3){$u_{111}$}
\textref h:C v:C \htext(3 3.2){$u_{011}$}
\textref h:C v:C \htext(3 -.2){$a$}
\textref h:C v:C \htext(3.2 2.1){$a$}
\textref h:C v:C \htext(4.1 3.1){$a$}
\textref h:C v:C \htext(5.1 2){$c$}
\textref h:C v:C \htext(2.7 .5){$b$}
\textref h:C v:C \htext(4.7 .5){$b$}
\textref h:C v:C \htext(1.9 1){$c$}

\textref h:C v:C \htext(3.1 -1){(b) Cube}
\textref h:C v:C \htext(-0.9 -1){(a) Quad-Graph}
}
\caption{3D-consistency.}\label{3Dconsistency}\index{3D consistency}
\end{figure}

\begin{note}\normalfont
In the above interpretation, we have adopted the following notation: We consider the square in Figure \ref{3Dconsistency}-(a) to be an elementary square in a two dimensional lattice. Then, we assume that field $u$ depends on two discrete variables $n$ and $m$, i.e. $u=u(n,m)$. Therefore, $u_{ij}$s on the vertices of \ref{3Dconsistency}-(a) are
\begin{equation}
u_{00}=u(n,m),\quad u_{10}=u(n+1,m),\quad u_{01}=u(n,m+1),\quad u_{11}=u(n+1,m+1).
\end{equation}
Moreover, for the interpretation on the cube we assume that $u$ depends on a third variable $k$, such that
\begin{equation}
u_{000}=u(n,m,k),\quad u_{100}=u(n+1,m,k),\ldots\quad u_{111}=u(n+1,m+1,k+1).
\end{equation}
\end{note}

Now, as an illustrative example we use the discrete potential KdV equation\index{discrete (potential) KdV equation} which first appeared in \cite{HirotaKdV}.

\begin{example}\label{dpKdvEx}\normalfont (Discrete potential KdV equation)\index{discrete (potential) KdV equation}
Consider equation (\ref{QGeq}), where $Q$ is given by
\begin{equation}
Q(u,u_{10},u_{01},u_{11};a,b)=(u-u_{11})(u_{10}-u_{01})+b-a.
\end{equation}
Now, using the bottom, front and left faces of the cube \ref{3Dconsistency}-(b), we can solve equations
\begin{subequations}
\begin{align}
Q(u,u_{100},u_{010},u_{110};a,b)&=0,\\
Q(u,u_{100},u_{001},u_{101};a,c)&=0,\\
Q(u,u_{010},u_{001},u_{011};b,c)&=0,
\end{align}
\end{subequations}
to obtain solutions for $u_{110}$, $u_{101}$ and $u_{011}$, namely
\begin{subequations}\label{u110u101u001}
\begin{align}
u_{110}&=u+\frac{a-b}{u_{010}-u_{100}},\label{u110}\\
u_{101}&=u+\frac{a-c}{u_{001}-u_{100}},\label{u101}\\
u_{011}&=u+\frac{b-c}{u_{001}-u_{010}},\label{u011}
\end{align}
\end{subequations}
respectively.

Now, if we shift (\ref{u110}) in the $k$-direction, and then substitute $u_{101}$ and $u_{011}$ (which appear in the resulting expression for $u_{11}$) by (\ref{u110u101u001}), we deduce
\begin{equation}
u_{111}=-\frac{(a-b)u_{100}u_{010}+(b-c)u_{010}u_{001}+(c-a)u_{100}u_{001}}{(a-b)u_{001}+(b-c)u_{100}+(c-a)u_{010}}.
\end{equation}
It is obvious that, because of the symmetry in the above expression, we would obtain exactly the same expression for $u_{111}$ if we had alternatively shifted $u_{101}$ in the $m$-direction and substituted $u_{110}$ and $u_{011}$ by (\ref{u110u101u001}), or if we had shifted $u_{011}$ in the $n$-direction and substituted $u_{110}$ and $u_{101}$. Thus, the dpKdV\nomenclature{dpKdV}{discrete potential Korteweg-de Vries} equation is 3D-consistent.\index{3D consistent equation(s)} 
\end{example}

\subsection{ABS classification of maps on quad-graphs}
In 2003 \cite{ABS-2004} Adler, Bobenko and Suris\nomenclature{ABS}{Adler, Bobenko and Suris} classified all the 3D-consistent equations\index{3D consistent equation(s)} in the case where $\mathcal{A}=\field{C}$. In particular, they considered all the equations of the form (\ref{QGeq}), where $u,u_{10},u_{01},u_{11},a,b$$\in\field{C}$, that satisfy the following properties:

\textbf{(I) Multilinearity.} Function $Q=Q(u,u_{10},u_{01},u_{11};a,b)$ is a first order polynomial in each of its arguments, namely linear in each of the fields $u,u_{10},u_{01},u_{11}$. That is,
\begin{equation}
Q(u,u_{10},u_{01},u_{11};a,b)=a_1uu_{10}u_{01}u_{11}+a_2uu_{10}u_{01}+a_3uu_{10}u_{11}+\ldots+a_{16},
\end{equation}
where $a_i=a_i(a,b)$, $i=1,\ldots,16$.

\textbf{(II) Symmetry.} Function $Q$ satisfies the following symmetry property
\begin{equation}
Q(u,u_{10},u_{01},u_{11};a,b)=\epsilon Q(u,u_{01},u_{10},u_{11};b,a)=\sigma Q(u_{10},u,u_{11},u_{01};a,b),
\end{equation}
with $\epsilon,\sigma=\pm1$.

\textbf{(III) Tetrahedron property.} That is, the final value $u_{111}$ is independent of $u$.

ABS proved that all the equations of the form (\ref{QGeq}) which satisfy the above conditions, can be reduced to seven basic equations, using M\"obius (fraction linear) transformations of the independent variables and point transformations of the parameters. These seven equations are distributed into two lists known as the $Q$-\textit{list}\index{$Q$-list} (list of 4 equations) and the $H$-\textit{list}\index{$H$-list} (list of 3 equations).

\begin{remark}\normalfont
The dpKdV equation\index{discrete (potential) KdV equation} in Example \ref{dpKdV} is the fisrt member of the $H$-list ($H1$ equation).
\end{remark}

\subsubsection{Lax representations}
Those equations of the form (\ref{QGeq}) which satisfy the multilinearity condition (I), admit Lax representation\index{Lax representation}. In fact, in this case, introducing an auxiliary spectral parameter, $\lambda$, there is an algorithmic way to find a matrix $L$ such that equation (\ref{QGeq}) can be written as the following \textit{zero-curvature} equation
\begin{equation}\label{zerocurv}
L(u_{11},u_{01};a,\lambda)L(u_{01},u;b,\lambda)=L(u_{11},u_{10};b,\lambda)L(u_{10},u;a,\lambda).
\end{equation}

We shall see later on that 1) equations of the form (\ref{QGeq}) with the fields on the edges of the square \ref{3Dconsistency}-(a) are related to Yang-Baxter maps\index{Yang-Baxter (YB) map(s)} and 2) Yang-Baxter maps may have Lax representation as (\ref{zerocurv}).

\subsection{Classification of quadrirational maps: The $F$-list}
A year after the classification of the 3D-consistent equations, ABS in \cite{ABS-2005} classified all the quadrirational maps\index{quadrirational map(s)} in the case where $\mathcal{A}=\field{CP}^1$\index{$\field{CP}^1$}; the associated list of maps is known as the $F$-list.\index{$F$-list} Recall that, a map $Y:(x,y)\mapsto (u(x,y),v(x,y))$ is called \textit{quadrirational}, if the maps
\begin{equation}
u(.,y):\mathcal{A}\rightarrow \mathcal{A},\quad v(x,.):\mathcal{A}\rightarrow \mathcal{A},
\end{equation}
are birational. In particular, we have the following.

\begin{theorem}(ABS, $F$-list)\index{$F$-list} Up to M\"obius transformations\index{M\"obius transformations}, any quadrirational map \index{quadrirational map(s)} on $\field{CP}^1\times\field{CP}^1$\index{$\field{CP}^1$} is equivalent to one of the following maps
\begin{align}
 u&=ayP, ~~~\quad v=bxP,~~~\qquad P=\frac{(1-b)x+b-a+(a-1)y}{b(1-a)x+(a-b)xy+a(b-1)y}; \tag{$F_I$}\label{FI}\\
 u&=\frac{y}{a}P, \quad\quad v=\frac{x}{b}P,\quad\qquad P=\frac{ax-by+b-a}{x-y}; \tag{$F_{II}$}\label{FII}\\
 u&=\frac{y}{a}P,\quad \quad v=\frac{x}{b}P, \quad\qquad P=\frac{ax-by}{x-y};\tag{$F_{III}$}\label{FIII}\\
 u&=yP ~~\quad\quad v=xP,~\quad\qquad P=1+\frac{b-a}{x-y}\tag{$F_{IV}$};\label{FIV}\\
 u&=y+P,\quad v=x+P,\qquad P=\frac{a-b}{x-y}\tag{$F_{V}$},\label{FV}
\end{align}
up to suitable choice of the parameters $a$ and $b$.
\end{theorem}

We shall come back to the $F$-list in chapter 4, where we shall see that all the equations of the $F$-list have the Yang-Baxter property\index{Yang-Baxter property}; yet, the other members of their equivalence classes may not satisfy the Yang-Baxter equation\index{Yang-Baxter equation}. However, we shall present a more precise list given in \cite{PSTV}. 

Finally, we devote the last part of this introduction to present the plan of this thesis.

\section{Organisation of the thesis}
The results of the thesis are distributed to chapters 3, 5 and 6 and appear in the articles \cite{SPS}, \cite{Sokor-Sasha} and \cite{GSS}, respectively. The character of chapter 2 is introductory, while chapter 4 is a review to recent developments in the area of Yang-Baxter maps.\index{Yang-Baxter (YB) map(s)} Specifically, this thesis is organised as follows.

\textbf{Chapter 2} deals with B\"acklund and Darboux transformations. In particular, starting with the original theorem of Darboux, that was presented in 1882 (\cite{Darboux}), we explain that a Darboux transformation\index{Darboux transformation(s)} is nothing else but a transformation which leaves covariant a Sturm-Liouville problem.\index{Sturm-Liouville problem} We show that this fact can be used to construct hierarchies of solutions of particular nonlinear equations and we present the very well-known Darboux transformation for the KdV equation. Moreover, we explain what are B\"acklund transformations,\index{B\"acklund transformation(s)} namely transformations which relate either solutions of a particular PDE (auto-BT),\index{auto-B\"acklund transformation(s)} or solutions of different PDEs (hetero-BT).\index{hetero-B\"acklund transformation(s)} We show how, using BTs, one can construct solutions of a nonlinear PDE  in an algebraic manner, and we present the well-known examples of the BTs for the sine-Gordon equation\index{sine-Gordon equation} and the KdV equation.\index{KdV equation}

In \textbf{chapter 3} we derive Darboux transformations for particular NLS type equations, namely the NLS equation,\index{nonlinear Schr\"odinger (NLS)!equation} the DNLS equation\index{derivative nonlinear Schr\"odinger (DNLS)!equation} and a deformation of the DNLS equation.\index{deformation of DNLS!equation} The spatial parts of the Lax pair\index{Lax pair!spatial part ($x$-part)} of these equations are represented by (Lax) operators\index{Lax operator(s)} which possess certain symmetries; in particular, these symmetries are due to the action of the reduction group.\index{reduction group} In all the afore-mentioned cases, we derive DTs which are understood as gauge-like transformations\index{gauge-like transformation(s)} which depend rationally on a spectral parameter and inherit the symmetries of their corresponding Lax operator. These DTs\index{Darboux transformation(s)} are employed in the construction of novel discrete integrable systems\index{integrable system(s)!discrete} which have first integrals\index{first integral(s)} and, in some cases, can be reduced to Toda type equations.\index{Toda (type) equation(s)} Moreover, the derivation of the DT\index{Darboux transformation(s)} implies other significant objects, such as B\"acklund transformations\index{B\"acklund transformation(s)} for the corresponding PDEs, as well as symmetries and conservation laws\index{conservation law(s)} for the associated discrete systems. All these cases of NLS type equations studied in this chapter correspond to recent classification results.

\textbf{Chapter 4} has introductory character and it is devoted to Yang-Baxter maps.\index{Yang-Baxter (YB) map(s)} In particular, we explain what Yang-Baxter maps are and what is their connection with matrix refactorisation problems.\index{refactorisation problem(s)} Moreover, we show the relation between the YB equation\index{Yang-Baxter (YB) equation} and 3D consistency\index{3D consistency} equations, plus we review some of the recent developments, such as the associated transfer dynamics\index{transfer dynamics} and some recent classification results.

In \textbf{chapter 5} we employ the Darboux transformations\index{Darboux transformation(s)} --derived in \textbf{chapter 3}-- in the construction of Yang-Baxter maps,\index{Yang-Baxter (YB) map(s)} and we study their integrability as finite discrete maps. Particularly, we construct six-dimensional YB maps which can be restricted to four-dimensional YB maps which are completely integrable\index{completely integrable} in the Liouville sense. These integrable restrictions are motivated by the existence of certain first integrals.\index{first integral(s)} In the case of NLS equation, the four-dimensional restriction is the Adler-Yamilov map.\index{Adler-Yamilov!map}

\textbf{Chapter 6} is devoted to the noncommutative extensions\index{noncommutative extension(s)} of both Darboux transformations\index{Darboux transformation(s)} and Yang-Baxter maps\index{Yang-Baxter (YB) map(s)} in the cases of NLS and DNLS equations. Specifically, we show that there are explicit Yang-Baxter maps with Darboux-Lax representation between Grassman algebraic varieties. We deduce novel endomorphisms of Grassmann varieties and, in particular, we present ten-dimensional maps which can be restricted to eight-dimensional Yang-Baxter maps on invariant leaves, related to the Grassmann-extended NLS and DNLS equations. We discuss their Liouville integrability and we consider their vector generalisations.

Finally, in \textbf{chapter 7} we provide the reader with a summary of the results of the thesis, as well as with some ideas for future work.

\chapter{B\"acklund and Darboux transformations}
\label{chap2} \setcounter{equation}{0}
\renewcommand{\theequation}{\thechapter.\arabic{equation}}

\section{Overview}
B\"acklund and Darboux (or Darboux type) transformations \index{Darboux transformation(s)}\index{B\"acklund transformation(s)} originate from differential geometry of surfaces in the nineteenth century, and they constitute an important and very well studied connection with the modern soliton theory and the theory of integrable systems.\index{integrable system(s)}

In the modern theory of integrable systems, these transformations are used to generate solutions of partial differential equations, starting from known solutions, even trivial ones. In fact, Darboux transformations\index{Darboux transformation(s)} apply to systems of linear equations, while B\"acklund transformations\index{B\"acklund transformation(s)} are generally related to systems of nonlinear equations.

This chapter is organised as follows: The next section deals with Darboux transformations and, in particular, the original theorem of Darboux and its application to the KdV equation,\index{KdV equation} as well as its generalisation, namely Crum's theorem. Then, section 3 is devoted to B\"acklund transformations and how they can be used to construct solutions in a algebraic way starting with known ones, using Bianchi's permutability;\index{Bianchi's!permutability} in particular, we present the examples of the B\"acklund transformation for the sine-Gordon equation and the KdV equation.

For further information on B\"acklund and Darboux transformations we indicatively refer to \cite{GuChaohao, Matveev-Salle, Rog-Schief} (and the references therein).

\section{Darboux transformations}
In 1882 Jean Gaston Darboux \cite{Darboux} presented the so-called ``Darboux theorem''\index{Darboux theorem} which states that a Sturm-Liouville problem\index{Sturm-Liouville problem} is covariant with respect to a linear transformation. In the recent literature, this is called the \textit{Darboux transformation} \cite{Matveev-Salle, Rog-Schief}. The first book devoted to the relation between Darboux transformations\index{Darboux transformation(s)} and the soliton theory is that of Matveev and Salle \cite{Matveev-Salle}.

\subsection{Darboux's theorem}
Darboux's original result is related to the so-called \textit{one-dimensional, time-independent Schr\"odinger} equation,\index{time-independent Schr\"odinger equation} namely
\begin{equation}\label{Sturm-Liouville}
y''+(\lambda-u)y=0, \quad u=u(x),
\end{equation}
which can be found in the literature as a \textit{Sturm-Liouville problem}\index{Sturm-Liouville problem} of finding eigenvalues and eigenfunctions. Moreover, we refer to $u$ as a \textit{potential function}, or just \textit{potential}.

In particular we have the following.

\begin{theorem}(Darboux)\label{DarbouxTheorem}\index{Darboux theorem}
Let $y_1=y_1(x)$ be a particular integral of the Sturm-Liouville problem \index{Sturm-Liouville problem} (\ref{Sturm-Liouville}), for the value of the spectral parameter $\lambda=\lambda_1$. Consider also the following (Darboux) transformation
\begin{equation}\label{DT}
y\mapsto y[1]:=\left(\frac{d}{dx}-l_1\right)y,
\end{equation}
of an arbitrary solution, $y$, of (\ref{Sturm-Liouville}),  where $l_1= l_1(y_1)=y_{1,x}y_1^{-1}$ is the logarithmic derivative of $y_1$. Then, $y[1]$ obeys the following equation
\begin{subequations}\label{transEq}
\begin{align}
y''[1]+&(\lambda-u[1])y[1]=0,\label{transEq-1}
\intertext{where $u[1]$ is given by}
&u[1]=u-2l_1'.\label{transEq-2}
\end{align}
\end{subequations}
\end{theorem}

\begin{proof}
Substitution of $y[1]$ in (\ref{DT}) into (\ref{transEq}) implies
\begin{equation}\label{eqyy'}
(u-2l_1'-u[1])y'+(u'-l_1 u-l_1''+l_1 u[1])y=0,
\end{equation}
where we have used (\ref{Sturm-Liouville}) to express $y''[1]$ and $y'''[1]$ in terms of $y$ and $y'$. Now, since $y$ in (\ref{eqyy'}) is arbitrary, it follows that
\begin{equation}\label{eqsuld}
u[1]=u-2l_1',\quad u'-l_1 u-l_1''+l_1 u[1]=0.
\end{equation}
Now, substitution of the first equation of (\ref{eqsuld}) to the second implies
\begin{equation}\label{intld}
u-l_1^2-l_1'=\lambda_1=\text{const.},
\end{equation}
after one integration with respect to $x$. Equation (\ref{intld}) is identically satisfied due to the definition of the logarithmic function, $l_1$, and the fact that $y_1$ obeys (\ref{Sturm-Liouville}).
\end{proof}

Darboux's theorem states that function $y[1]$ given in (\ref{DT}) obeys a Sturm-Liouville problem of the same structure with (\ref{Sturm-Liouville}), namely the same equation (\ref{Sturm-Liouville}) but with an updated potential $u[1]$. In other words, equation (\ref{Sturm-Liouville}) is covariant with respect to the Darboux transformation, $y\mapsto y[1]$, $u\mapsto u[1]$.

\subsection{Darboux transformation for the KdV equation and Crum's theorem}\index{KdV equation!Darboux transformation}
The significance of the Darboux theorem lies in the fact that transformation (\ref{DT}) maps solutions of a Sturm-Liouville equation (\ref{Sturm-Liouville}) to other solutions of the same equation, which allows us to construct hierarchies of such solutions. At the same time, the theorem provides us with a relation between the ``old'' and the ``new'' potential. In fact, if the potential $u$ obeys a nonlinear ODE\nomenclature{ODE}{Ordinary differential equation} (or more importantly a nonlinear PDE\nomenclature{PDE}{Partial differential equation}\footnote{Potential $u$ may depend on a temporal parameter $t$, namely $u=u(x,t)$.}), then relation (\ref{transEq}) may allow us to construct new non-trivial solutions starting from trivial ones, such as the zero solution.

\begin{example}\normalfont
Consider the Sturm-Liouville equation (\ref{Sturm-Liouville}) in the case where the potential, $u$, satisfies the KdV equation.\index{KdV equation} Therefore, both the eigenfunction $y$ and the potential $u$ depend on $t$, which slips into their expressions as a parameter.

In this case, equation (\ref{Sturm-Liouville}) is nothing else but the spatial part of the Lax pair for the KdV equation\index{Lax pair!for KdV equation} that we have seen in the previous chapter; recall:
\begin{equation}\label{x-LaxKdV}
\mathcal{L} y=\lambda y \quad \text{or}\quad y_{xx}+(\lambda-u(x,t))y=0.
\end{equation}

Now, according to theorem \ref{DarbouxTheorem}, for a known solution of the KdV \nomenclature{KdV}{Korteweg de Vries} equation,\index{KdV equation} say $u$, we can solve (\ref{Sturm-Liouville}) to obtain $y=y(x,t;\lambda)$. Evaluating at $\lambda=\lambda_1$, we get $y_1(x,t)=y(x,t;\;\lambda_1)$ and thus, using equation (\ref{transEq-2}), a new potential $u[1]$. Therefore, we simultaneously obtain new solutions, $(y[1],u[1])$, for both the linear equation (\ref{x-LaxKdV}) and the KdV equation\footnote{Potential $u[1]$ is a solution of the KdV equation,\index{KdV equation} since it can be readily shown that the pair $(y[1],u[1])$ also satisfies the temporal part of the Lax pair for KdV.}, which are given by
\begin{subequations}\label{eqy1u1}
\begin{align}
y[1]&=(\partial_x-l_1)y,\label{eqy1u1-1}\\
u[1]&=u-2l_{1,x},        \label{eqy1u1-2}
\end{align}
\end{subequations}
respectively.

Now, applying the Darboux transformation once more, we can construct a second solution of the KdV equation\index{KdV equation} in a fully algebraic manner. Specifically, first we consider the solution $y_2[1]$, which is $y[1]$ evaluated at $\lambda=\lambda_2$, namely
\begin{equation}
y_2[1]=(\partial_x-l_1)y_2.
\end{equation}
where $y_2=y(x,t;\lambda_2)$
Then, we obtain a second pair of solutions, $(y[2],u[2])$, for (\ref{x-LaxKdV}) and the KdV equation,\index{KdV equation} given by
\begin{subequations}
\begin{align}
y[2]&=(\partial_x-l_2)y[1]\stackrel{(\ref{eqy1u1-1})}{=}(\partial_x-l_2)(\partial_x-l_1)y,\\
u[2]&=u[1]-2l_{2,x}\stackrel{(\ref{eqy1u1-2})}{=}u-2(l_{1,x}+l_{2,x}).
\end{align}
\end{subequations} 

This procedure can be repeated successively, in order to construct hierarchies of solutions for the KdV equation,\index{KdV equation} namely
\begin{equation}
(y[1],u[1])\rightarrow(y[2],u[2])\rightarrow\cdots\rightarrow(y[n],u[n])\rightarrow\cdots,
\end{equation}
where $(y[n],u[n])$ are given by
\begin{equation}\label{eqynun}
y[n]=\left(\prod_{k=1}^{\stackrel{\curvearrowleft}{n}}(\partial_x-l_k)\right)y,\quad u[n]=u-2\sum_{k=1}^n(l_{k,x}),
\end{equation} 
where ``$\curvearrowleft$" indicates that the terms of the above ``product" are arranged from the right to the left.
\end{example}

We must note that Crum in 1955 \cite{Crum} derived more practical and elegant expressions for $y[n]$ and $u[n]$, in (\ref{eqynun}), which are formulated in the following generalisation of the Darboux theorem \ref{DarbouxTheorem}.

\begin{theorem}\label{Crum}(Crum) Let $y_1, y_1,\dots, y_n$ be particular integrals of the Sturm-Liouville equation (\ref{Sturm-Liouville}), corresponding to the eigenvalues $\lambda_1, \lambda_2,\ldots, \lambda_n$. Then, the following function
\begin{equation}
y[n]=\frac{W[y_1, y_2,\dots, y_n, y]}{W[y_1, y_2,\dots, y_n]},
\end{equation}
where $W[y_1, y_2,\dots, y_n]$ denotes the Wronskian determinant of functions $y_1, y_1,\dots, y_n$, obeys the following equation
\begin{equation}
y_{x,x}[n]+(\lambda-u[n])y[n]=0,
\end{equation}
where the potential $u[n]$ is given by
\begin{equation}
u[n]=u-2\frac{d^2}{dx^2}\ln(W[y_1, y_2,\dots, y_n]).
\end{equation}
\end{theorem}

\begin{remark}\normalfont
For $n=1$, Crum's theorem \ref{Crum} coincides with Darboux's theorem \ref{DarbouxTheorem}.
\end{remark}

In this thesis, we understand Darboux transformations as gauge-like transformations\index{gauge-like transformation(s)} which depend on a spectral parameter. In fact, as we shall see in the next chapter, their dependence on the spectral parameter is essential to construct discrete integrable systems.

\section{B\"acklund transformations}
As mentioned earlier, B\"acklund transformations \nomenclature{BT}{B\"acklund transformation}\index{B\"acklund transformation(s)} originate in differential geometry in the 1880s and, in particular, they arose as certain transformations between surfaces.

In the theory of integrable systems, they are seen as relations between solutions of the same PDE (auto-BT)\index{auto-B\"acklund transformation(s)} or as relations between solutions of two different PDEs (hetero-BT).\index{hetero-B\"acklund transformation(s)} Regarding the nonlinear equations which have Lax representation, Darboux transformations apply to the associated linear problem (Lax pair), while B\"acklund transformations\index{B\"acklund transformation(s)} are related to the nonlinear equation itself. Therefore, unlike DTs, BTs\index{B\"acklund transformation(s)} do not depend on the spectral parameter which appears in the definition of the Lax pair.\index{Lax pair} Yet, both DTs \nomenclature{DT}{Darboux transformation} and BTs\index{B\"acklund transformation(s)} serve the same purpose; they are used to construct non-trivial solutions starting from trivial ones.

\begin{definition} (BT-loose Def.) \index{B\"acklund transformation(s)}Consider the following partial differential equations for $u$ and $v$:
\begin{subequations}
\begin{align}
F(u,u_x,u_t,u_{xx},u_{xt},\ldots)&=0,\\
G(v,v_x,v_t,v_{xx},v_{xt},\ldots)&=0.
\end{align}
\end{subequations}
Consider also the following pair of relations 
\begin{equation}
\mathcal{B}_i(u,u_x,u_t,\ldots,v,v_x,v_t,\ldots)=0,
\end{equation}
between $u$, $v$ and their derivatives. If $\mathcal{B}_i=0$ is integrable for $v$, $\mod \langle F=0 \rangle$, and the resulting $v$ is a solution of $G=0$, and vice versa, then it is an hetero-B\"acklund transformation.\index{hetero-B\"acklund transformation(s)} Moreover, if $F\equiv G$, the relations $\mathcal{B}_i=0$ is an auto-B\"acklund transformation.\index{auto-B\"acklund transformation(s)}
\end{definition}

The simplest example of BT\index{B\"acklund transformation(s)} are the well-known Cauchy-Riemann\index{Cauchy-Riemann relations} relations in complex analysis, for the analyticity of a complex function, $f=u(x,t)+v(x,t)i$.

\begin{example}\normalfont (Laplace equation)
Functions $u=u(x,t)$ and $v=v(x,t)$ are harmonic, namely
\begin{equation}\label{Laplace}
\nabla^2u=0, \quad \nabla^2v=0,
\end{equation}
if the following Cauchy-Riemann relations\index{Cauchy-Riemann relations} hold
\begin{equation}\label{Cauchy-Riemann}
u_x=v_t, \quad u_t=-v_x.
\end{equation}
The latter equations constitute an auto-BT\index{B\"acklund transformation(s)!for Laplace equation} for the Laplace equation (\ref{Laplace})\index{Laplace equation} and can be used to construct solutions of the same equations, starting with known ones. For instance, consider the simple solution $v(x,t)=xt$. Then, according to (\ref{Cauchy-Riemann}), a second solution of (\ref{Laplace}), $u$, has to satisfy $u_x=x$ and $u_t=-t$. Therefore, $u$ is given by
\begin{equation}
u=\frac{1}{2}(x^2-t^2).
\end{equation}
\end{example}

However, even though Laplace's equation is linear, the same idea works for nonlinear equations. 


\subsection{BT for sine-Gordon equation and Bianchi's permutability}\index{Bianchi's!permutability}
One of the first examples of BT was for the nonlinear sine-Gordon equation,\index{sine-Gordon equation}\index{sine-Gordon equation!B\"acklund transformation}\index{B\"acklund transformation(s)!for SG equation}\nomenclature{SG}{sine-Gordon}
\begin{equation}\label{SG}\index{sine-Gordon equation}
u_{xt}=\sin u, \quad u=u(x,t).
\end{equation}

Let us now consider the following well-known relations
\begin{align}\label{BTSG}
\mathcal{B}_\alpha:
\begin{cases}
\left(\frac{u+v}{2}\right)_x&=\alpha \sin\left(\frac{u-v}{2}\right),\\
\left(\frac{u-v}{2}\right)_t&=\frac{1}{\alpha} \sin\left(\frac{u+v}{2}\right),
\end{cases}
\end{align}
between functions $u=u(x,t)$ and $v=v(x,t)$.

We have the following.

\begin{proposition}
Relations (\ref{BTSG}) constitute an auto-BT between the solutions $u=u(x,t)$ and $v=v(x,t)$ of the SG equation (\ref{SG}).\index{sine-Gordon equation}\index{auto-B\"acklund transformation(s)!for SG equation}
\end{proposition}
\begin{proof}
Differentiating the first equation of (\ref{BTSG}) with respect to $t$ and the second with respect to $x$, we obtain
\begin{subequations}\label{d-BTSG}
\begin{align}
\left(\frac{u+v}{2}\right)_{xt}&=\cos\left(\frac{u-v}{2}\right)\sin\left(\frac{u+v}{2}\right),\label{d-BTSG-1}\\
\left(\frac{u-v}{2}\right)_{tx}&=\cos\left(\frac{u+v}{2}\right)\sin\left(\frac{u-v}{2}\right),\label{d-BTSG-2}
\end{align}
\end{subequations}
where he have made use of (\ref{BTSG}). Now, we demand that the above equations are compatible, namely $u_{xt}=u_{tx}$ and $v_{xt}=v_{tx}$. Adding equations (\ref{d-BTSG}) by parts, we deduce that $u$ obeys the SG equation.\index{sine-Gordon equation} Moreover, the same is true for $v$ after subtraction of (\ref{d-BTSG}) by parts. Hence, (\ref{BTSG}) is an auto-B\"acklund transformation\index{auto-B\"acklund transformation(s)!for SG equation} for the SG equation.\index{sine-Gordon equation}
\end{proof}

\begin{remark}\normalfont
We shall refer to the first equation of (\ref{BTSG}) as the \textit{spatial part} (or $x$-\textit{part}) of the BT,\index{B\"acklund transformation(s)!spatial part ($x$-part)} while we refer to the second one as the \textit{temporal part} (or $t$-\textit{part}) \index{B\"acklund transformation(s)!temporal part ($t$-part)}of the BT.
\end{remark}

\subsubsection{Bianchi's permutability: Nonlinear superposition principle of solutions of the SG equation}\index{Bianchi's!permutability}\index{nonlinear superposition principle}
Starting with a function $u=u(x,t)$, such that $u_{xt}=\sin u$, one can construct a second solution of the SG equation, $u_1=\mathcal{B}_{\alpha_1}(u)$, using the spatial part of the BT (\ref{BTSG}), namely
\begin{equation}\label{BTsgu1}
\left(\frac{u_1+u}{2}\right)_x=\alpha_1 \sin\left(\frac{u_1-u}{2}\right).
\end{equation}
Moreover, using another parameter, $\alpha_2$, we can construct a second solution $u_2=\mathcal{B}_{\alpha_2}(u)$, given by
\begin{equation}\label{BTsgu2}
\left(\frac{u_2+u}{2}\right)_x=\alpha_2 \sin\left(\frac{u_2-u}{2}\right).
\end{equation}

Now, starting with the solutions $u_1$ and $u_2$, we can construct two new solutions $u_{12}$ and $u_{21}$ from relations $u_{12}=\mathcal{B}_{\alpha_2}(u_1)$ and
$u_{21}=\mathcal{B}_{\alpha_1}(u_2)$, namely
\begin{eqnarray}
\left(\frac{u_{12}+u_1}{2}\right)_x&=\alpha_2 \sin\left(\frac{u_{12}-u_1}{2}\right),\label{BTsgu12}\\
\left(\frac{u_{21}+u_2}{2}\right)_x&=\alpha_1 \sin\left(\frac{u_{21}-u_2}{2}\right),\label{BTsgu21}
\end{eqnarray}
as represented schematically in Figure \ref{BianchiPer}-(a).

\begin{figure}[ht]
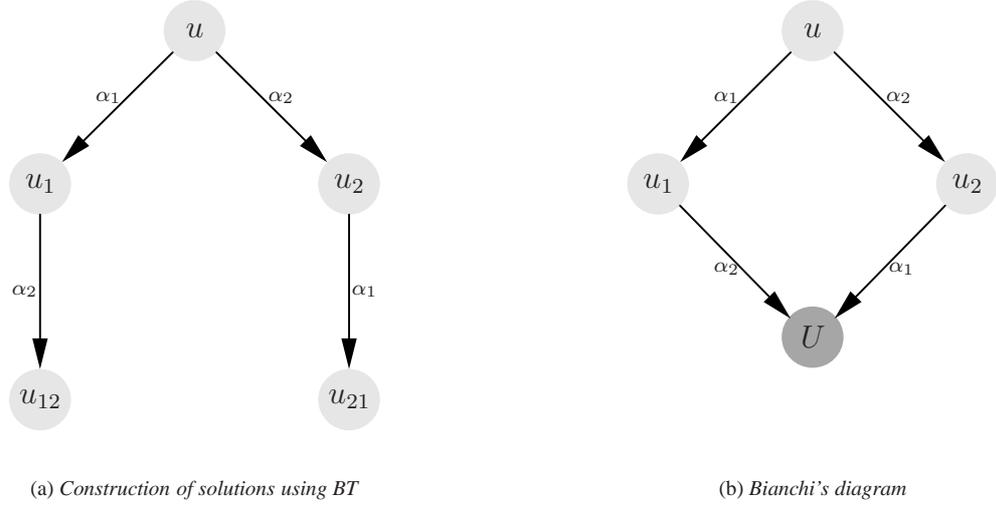

\centering
\centertexdraw{ \setunitscale 0.8
\move (0 1)\fcir f:.9 r:.2  \move (-1 0)\fcir f:.9 r:.2 \move (1 0)\fcir f:.9 r:.2 \move (-1 -1.41)\fcir f:.9 r:.2 \move (1 -1.41)\fcir f:.9 r:.2
\move (4 1)\fcir f:.9 r:.2  \move (3 0)\fcir f:.9 r:.2 \move (5 0)\fcir f:.9 r:.2 \move (4 -1)\fcir f:.65 r:.2 
\textref h:C v:C \htext(0 1){$u$} \textref h:C v:C \htext(-1 0){$u_1$} \textref h:C v:C \htext(1 0){$u_2$} \textref h:C v:C \htext(-1 -1.41){$u_{12}$} \textref h:C v:C \htext(1 -1.41){$u_{21}$}
\textref h:C v:C \htext(4 1){$u$} \textref h:C v:C \htext(3 0){$u_1$} \textref h:C v:C \htext(5 0){$u_2$} \textref h:C v:C \htext(4 -1){$U$}
\move (-.14 .86) \arrowheadtype t:F \avec(-.86 .14) \move (.14 .86) \arrowheadtype t:F \avec(.86 .14)
\move (-1 -.2) \arrowheadtype t:F \avec(-1 -1.21) \move (1 -.2) \arrowheadtype t:F \avec(1 -1.21)
\move (3.86 .86) \arrowheadtype t:F \avec(3.14 .14) \move (4.14 .86) \arrowheadtype t:F \avec(4.86 .14)
\move (3.14 -.14) \arrowheadtype t:F \avec(3.86 -0.88)
\move (4.86 -.14) \arrowheadtype t:F \avec (4.14 -0.88)
\textref h:C v:C \scriptsize{\htext(0 -2){(a) \textit{Construction of solutions using BT}}}
\textref h:C v:C \scriptsize{\htext(4 -2){(b) \textit{Bianchi's diagram}}}
\textref h:C v:C \htext(-.56 .56){$\alpha_1$}
\textref h:C v:C \htext(.56 .56){$\alpha_2$}
\textref h:C v:C \htext(-1.1 -.705){$\alpha_2$}
\textref h:C v:C \htext(1.1 -.705){$\alpha_1$}
\textref h:C v:C \htext(3.44 .56){$\alpha_1$}
\textref h:C v:C \htext(4.56 .56){$\alpha_2$}
\textref h:C v:C \htext(3.44 -.56){$\alpha_2$}
\textref h:C v:C \htext(4.576 -.56){$\alpha_1$}
}
\caption{Bianchi's permutability}\label{BianchiPer}\index{Bianchi's!permutability}
\end{figure}

Nevertheless, the above relations need integration in order to derive the actual solutions $u_1$, $u_2$ and, in retrospect, solutions $u_{12}$ and $u_{21}$. Yet, having at our disposal solutions $u_1$ and $u_2$, a new solution can be constructed using Bianchi's permutativity\index{Bianchi's!permutability} (see Figure \ref{BianchiPer}-(b)) in a purely algebraic way. Specifically, we have the following.

\begin{proposition}
Imposing the condition $u_{12}=u_{21}$, the BTs \{(\ref{BTsgu1})-(\ref{BTsgu21})\} imply the following solution of the SG equation:
\begin{equation}\label{U-solSG}
U=u+4\arctan\left[\frac{\alpha_1+\alpha_2}{\alpha_2-\alpha_1}\tan\left(\frac{u_1-u_2}{4}\right)\right].
\end{equation}
\end{proposition}
\begin{proof}
First, we subtract (\ref{BTsgu2}) and (\ref{BTsgu21}) from (\ref{BTsgu1}) and (\ref{BTsgu12}), respectively, and then subtract by parts the resulting equations to obtain
\begin{equation*}
\alpha_1\left(\sin\left(\frac{u_1-u}{2}\right)+\sin\left(\frac{u_{21}-u_2}{2}\right)\right)=\alpha_2\left(\sin\left(\frac{u_2-u}{2}\right)+\sin\left(\frac{u_{12}-u_1}{2}\right)\right).
\end{equation*}

Now, the above equation becomes
\begin{equation}
\alpha_1\sin\left(\frac{U-u+u_1-u_2}{4}\right)=\alpha_2\sin\left(\frac{U-u-(u_1-u_2)}{4}\right),
\end{equation}
where we have used the well-known identity $\sin A+\sin B=2\sin(\frac{A+B}{2})\cos(\frac{A-B}{2})$.

Finally, using the identity $\sin(A\pm B)=\sin A \cos B\pm \sin B\cos A$, for $A=(U-u)/4$ and $B=u_1-u_2$, we deduce
\begin{equation}
\tan\left(\frac{U-u}{4}\right)=\frac{\alpha_1+\alpha_2}{\alpha_2-\alpha_1}\tan\left(\frac{u_1-u_2}{4}\right),
\end{equation}
which can be solved for $U$ to give (\ref{U-solSG}).
\end{proof}

\begin{remark}\normalfont
One can verify that $U$ given in (\ref{U-solSG}) satisfies both equations (\ref{BTsgu12}) (where $U=u_{12}$) and (\ref{BTsgu21}) (where $U=u_{21}$) modulo equations (\ref{BTsgu1}) and (\ref{BTsgu2}). Moreover, it satisfies the corresponding temporal part of equations (\ref{BTsgu12}) and (\ref{BTsgu21}).
\end{remark}

\begin{remark}\normalfont
Relation (\ref{U-solSG}) is nothing else but a nonlinear superposition principle \index{nonlinear superposition principle} for the production of solutions of the SG equation.
\end{remark}

\subsection{B\"acklund transformation for the PKdV equation}\index{KdV equation!B\"acklund transformation}\index{B\"acklund transformation(s)!for the PKdV equation}
An auto-B\"acklund transformation associated to the PKdV\nomenclature{PKDV}{Potential Korteweg de Vries} equation is given by the following relations
\begin{align}\label{BTKdV}
\mathcal{B}_\alpha:
\begin{cases}
\left(u+v\right)_x&=2\alpha+\frac{1}{2}(u-v)^2,\\
\left(u-v\right)_t&=3(u_x^2-v_x^2)-(u-v)_{xxx},
\end{cases}
\end{align}
which was first presented in 1973 in a paper of Wahlquist and Estabrook \cite{Wahlquist-Estabrook}. In this section we show how we can construct algebraically a solution of the PKdV equation,\index{KdV equation} using Bianchi's permutability.\index{Bianchi's!permutability}

\subsubsection{Bianchi's permutability: Nonlinear superposition principle of solutions of the PKdV equation}\index{Bianchi's!permutability}\index{nonlinear superposition principle}
Let $u=u(x,t)$ be a function satisfying the PKdV equation.\index{KdV equation} Focusing on the spatial part of the BT (\ref{BTKdV}), we can construct two new solutions, from $u_1=\mathcal{B}_{\alpha_1}(u)$ and $u_2=\mathcal{B}_{\alpha_2}(u)$, i.e.
\begin{eqnarray}
\left(u_1+u\right)_x&=&2\alpha_1+\frac{1}{2}(u_1-u)^2,\label{BTKdVu1}\\
\left(u_2+u\right)_x&=&2\alpha_2+\frac{1}{2}(u_2-u)^2.\label{BTKdVu2}
\end{eqnarray}

Moreover, following \ref{BianchiPer}-(b), we can construct two more from relations $u_{12}=\mathcal{B}_{\alpha_2}(u_1)$ and $u_{21}=\mathcal{B}_{\alpha_1}(u_2)$, i.e.
\begin{eqnarray}
\left(u_{12}+u_1\right)_x&=&2\alpha_2+\frac{1}{2}(u_{12}-u_1)^2,\label{BTKdVu12}\\
\left(u_{21}+u_2\right)_x&=&2\alpha_1+\frac{1}{2}(u_{21}-u_2)^2.\label{BTKdVu21}
\end{eqnarray}

\begin{proposition}
Imposing the condition $u_{12}=u_{21}$, the BTs \{(\ref{BTKdVu1})-(\ref{BTKdVu21})\} imply the following solution of the PKdV equation:
\begin{equation}
U=u-4\frac{\alpha_1-\alpha_2}{u_1-u_2}.
\end{equation}
\end{proposition}

\begin{proof}
It is straightforward calculation; one needs to subtract (\ref{BTKdVu2}) and (\ref{BTKdVu21}) by (\ref{BTKdVu1}) and (\ref{BTKdVu12}), respectively, and subtract the resulting equations.
\end{proof}

In the next chapter --where we study Darboux transformations\index{Darboux transformation(s)} for particular NLS type equations-- we shall see that BT\index{B\"acklund transformation(s)} arise naturally in the derivation of DT.

\chapter{Darboux tranformations for NLS type equations and discrete integrable systems}\index{integrable system(s)!discrete}
\label{chap3} \setcounter{equation}{0}
\renewcommand{\theequation}{\thechapter.\arabic{equation}}

\section{Overview}
As we have seen in the previous chapter, Darboux and B{\"a}cklund transformations are closely related to the notion of integrability \cite{Rog-Schief}. They can be derived from Lax pairs in a systematic way, e.g. \cite{Cieslinski, Crampin}, and provide the means to construct classes of solutions for the integrable equations to which they are related. Moreover, they can be interpreted as differential-difference equations \cite{Adler-Yamilov, levi1981, levi1980, Frank} and their commutativity, also referred to as Bianchi's permutability theorem,\index{Bianchi's!permutability} leads to systems of difference equations \cite{Adler1996, Frank1984, Quispel1984}.

In this chapter we use Lax operators which are invariant under the action of the reduction group to derive Darboux transformations. We interpret the associated Darboux matrices as Lax matrices of a discrete Lax pair and construct systems of difference equations. 

More precisely, our starting point is the general Lax operator 
\begin{equation}
\mathcal{L} = D_x + U(p,q;\lambda).
\end{equation}
Here, the $2 \times 2$ matrix $U$ belongs to the Lie algebra ${\mathfrak{sl}}(2,{\mathbb{C}})$, depends implicitly on $x$ through potentials $p$, $q$, and depends rationally on the spectral parameter $\lambda$. Imposing the invariance of operator $\mathcal{L}$ under the action of the automorphisms of ${\mathfrak{sl}}(2,{\mathbb{C}})$, i.e. the reduction group, inequivalent classes of Lax operators can be constructed systematically. This classification of the corresponding automorphic Lie algebras was presented in \cite{Lombardo-Sanders}, which were also derived in \cite{BuryPhD} in a different way.

Because of this construction, the resulting Lax operators have specific $\lambda$-dependence and possess certain symmetries. We shall be assuming that the corresponding Darboux transformations inherit the same $\lambda$-dependence and symmetries, hence the derivation of these transformations is considerably simplified in this context. Once a Darboux transformation has been derived, one may construct algorithmically new fundamental solutions, i.e. solutions of the equation $\mathcal{L}\Psi = 0$, from a given initial one using this transformation. Moreover, combining two different Darboux transformations and imposing their commutativity, a set of algebraic relations among the potentials involved in $\mathcal{L}$ results as a necessary condition.

One may interpret these potentials as functions defined on a two dimensional lattice and, consequently, the corresponding algebraic relations among them as a system of difference equations. One interesting characteristic of the resulting discrete systems is their multidimensional consistency \cite{ABS-2004, ABS-2005, Bobenko-Suris,Frank3,Frank5,Frank4}.\index{multidimensional consistency} This means that these systems can be extended into a three dimensional lattice in a consistent way and, consequently, in an infinite dimensional lattice. 

Another property of these systems, following from their derivation, is that they admit symmetries. The latter are nothing else but the B{\"a}cklund transformation of the corresponding continuous system to which the Lax operator $\mathcal{L}$ is related.

The chapter is organised as follows: In the next section we briefly explain what is a reduction group, what automorphic Lie algebras are and we list the cases of the PDEs we study; the nonlinear Schr\"odinger equation, the derivative nonlinear Schr\"odinger equation and a deformation of the derivative nonlinear Schr\"odinger equation. In section 3 we present the general scheme we follow to derive Darboux matrices and construct systems of difference equations. Finally, section 4 is devoted to the derivation of Darboux matrices for NLS type equations, while section 5 deals with employing these Darboux matrices to  construct discrete integrable systems\index{integrable system(s)!discrete} and their integrable reductions.

The results of this chapter appear in \cite{SPS}.


\section{The reduction group and automorphic Lie algebras}\index{reduction group}\index{automorphic Lie algebras}
The reduction group \index{reduction group}was first introduced in \cite{Mikhailov3, Mikhailov2}. It is a discrete group of automorphisms of a Lax operator,\index{Lax operator(s)} and its elements are simultaneous automorphisms of the corresponding Lie algebra and fractional-linear transformations of the spectral parameter.

Automorphic Lie algebras \index{automorphic Lie algebras} were introduced in \cite{LombardoPhD,Lombardo} and studied in \cite{BuryPhD,Bury-Sasha,LombardoPhD,Lombardo,Lombardo-Sanders}. These algebras constitute a subclass of infinite dimensional Lie algebras and their name is due to their construction which is very similar to the one for automorphic functions. 

Following Klein's classification \cite{Klein} of finite groups of fractional-linear transformations\index{fractional-linear transformations} of a complex variable, in \cite{BuryPhD, Bury-Sasha} it has been shown that in the case of $2\times 2$ matrices, which we study in this chapter, the essentially different reduction groups are
\begin{itemize}
\item the trivial group (with no reduction);
\item the cyclic reduction group $\field{Z}_2$\index{$\field{Z}_2$ reduction group}\index{reduction group!$\field{Z}_2$ (cyclic)} (leading to the Kac-Moody algebra\index{Kac-Moody algebra} $A_1^1$);
\item the Klein reduction group $\field{Z}_2\times \field{Z}_2 \cong \field{D}_2$.\index{$\field{Z}_2\times\field{Z}_2$ reduction group}\index{reduction group!$\field{Z}_2\times\field{Z}_2$ (Klein)}
\end{itemize}
Reduction groups $\field{Z}_2$ and $\field{D}_2$ have both \textit{degenerate}\index{degenerate orbits} and \textit{generic} orbits\index{generic orbits}. Degenerate are those orbits that correspond to the fixed points of the fractional-linear transformations\index{fractional-linear transformations} of the spectral paramater, while the others are called generic. 

Now, the following Lax operators 
\begin{subequations}
\begin{align}
\label{Lax-NLS}
&\mathcal{L} =D_x +\lambda\left(\begin{array}{cc} 1 & 0 \\ 0 & -1\end{array}\right)+\left(\begin{array}{cc} 0 & 2p \\ 2q & 0\end{array}\right),\\
\label{Lax-DNLS}
&\mathcal{L}=D_x+\lambda^{2} \left(\begin{array}{cc} 1 & 0 \\ 0 & -1\end{array}\right)+\lambda \left(\begin{array}{cc} 0 & 2p \\ 2q & 0\end{array}\right),\\
\label{Lax-dDNLS}
&\mathcal{L}=D_x+(\lambda^2-\lambda^{-2})\left(\begin{array}{cc} 1 & 0 \\ 0 & -1\end{array}\right)+\lambda \left(\begin{array}{cc} 0 & 2\,p\\ 2\,q & 0\end{array}\right)+\lambda^{-1} \left(\begin{array}{cc} 0 & 2\,q\\ 2\,p & 0\end{array}\right),
\end{align}
\end{subequations}
constitute all the essential different Lax operators, with poles of minimal order, invariant with respect to the generators of $\field{Z}_2$ and $\field{D}_2$ groups with degenerate orbits\index{degenerate orbits}. In what follows, we study the Darboux transformations\index{Darboux transformation(s)} for all the above cases.

Operator (\ref{Lax-NLS}) is associated with the NLS\nomenclature{NLS}{Nonlinear Schr\"odinger} equation  \cite{ZS},
\begin{equation}\label{NLS}
p_t=p_{xx}+4p^2q,\qquad q_t=-q_{xx}-4pq^2,
\end{equation}\index{nonlinear Schr\"odinger (NLS)!equation}
while (\ref{Lax-DNLS}) and (\ref{Lax-dDNLS}) are associated with the DNLS\nomenclature{DNLS}{Derivative nonlinear Schr\"odinger} equation \cite{Kaup-Newell},
\begin{equation}\label{DNLS}
p_t=p_{xx}+4(p^2q)_x,\qquad q_t=-q_{xx}+4(pq^2)_x.
\end{equation}\index{derivative nonlinear Schr\"odinger (DNLS)!equation}
and a deformation of the DNLS equation \cite{MSY}
\begin{equation}\label{dDNLS}
p_t=p_{xx}+8(p^2q)_x-4q_x,\qquad q_t=-q_{xx}+8(pq^2)_x-4p_x,
\end{equation}\index{deformation of DNLS!equation}
respectively.


\section{General framework}
In this section we present the general framework for the derivation of Darboux matrices related to Lax operators  of AKNS\nomenclature{AKNS}{Ablowitz, Kaup, Newell and Segur}\index{Lax operator(s)!AKNS type} type. Moreover, we explain how we can employ these Darboux matrices to construct discrete integrable systems.\index{integrable system(s)!discrete}

\subsection{Derivation of Darboux matrices}
The Lax operators \index{Lax operator(s)} which we consider in the rest of this thesis are of the following AKNS form
\begin{equation} \label{L-operator}
\mathcal{L}(p,q;\lambda)=D_x+U(p,q;\lambda),
\end{equation}
where $U$ is a $2 \times 2$ traceless matrix which belongs in the Lie algebra ${\mathfrak{sl}}(2,{\mathbb{C}})$, depends implicitly on $x$ through the potential functions $p$ and $q$ and is a rational function in the spectral parameter $\lambda \in \field{C}$. We also require that the dependence on the spectral parameter is nontrivial\footnote{By nontrivial we mean that $\lambda$ cannot be eliminated by a gauge transformation\index{gauge transformation}. For instance, matrix $\left(\begin{array}{cc} * & \lambda^{-1} \\  & * \end{array}\right)$ has trivial dependence on $\lambda$ since
$$\left(\begin{array}{cc} \lambda^{1/2} &  \\  & \lambda^{-1/2} \end{array}\right)\left(\begin{array}{cc} * & \lambda^{-1} \\  & * \end{array}\right)\left(\begin{array}{cc} \lambda^{1/2} &  \\  & \lambda^{-1/2} \end{array}\right)^{-1}=\left(\begin{array}{cc} * & 1 \\  & * \end{array}\right)$$
}.

\begin{remark}\normalfont 
In the forthcoming analysis we shall only be needing the spatial part of the Lax pair\index{Lax pair!spatial part ($x$-part)} of the associated PDEs.
\end{remark}

In what follows, by Darboux transformation \index{Darboux transformation(s)} we understand a map
\begin{equation}\label{DarbouxDef}
\mathcal{L}\rightarrow \mathcal{\tilde{L}}=M\mathcal{L}M^{-1},
\end{equation}
where $\mathcal{\tilde{L}}$ has exactly the same form as $\mathcal{L}$ but updated with new potentials $p_{10}$ and $q_{10}$, namely
\begin{equation}
\tilde{\mathcal{L}}=D_x+U(p_{10},q_{10};\lambda).
\end{equation}
Matrix $M$ in (\ref{DarbouxDef}) is an invertible matrix called the \textit{Darboux matrix}\index{Darboux matrix(-ces)}.

According to definition (\ref{DarbouxDef}), a Darboux matrix may depend on any of the potential functions $p$, $q$, $p_{10}$ and $q_{10}$, and the spectral parameter $\lambda$. Moreover, given a Lax operator\index{Lax operator(s)} $\mathcal{L}$, we can calculate the Darboux matrix using the following.
\begin{proposition}\label{DarbouxMatrix}
Given a Lax operator \index{Lax operator(s)} of the form (\ref{L-operator}), the Darboux matrix, $M$, satisfies the following equation
\begin{equation}\label{DMeq}
D_x M+U_{10}M-MU=0,
\end{equation} 
where $U_{10}=U(p_{10},q_{10};\lambda)$.
\end{proposition} 
\begin{proof}
By definition of the Darboux transformation we have that $\tilde{\mathcal{L}}M=M\mathcal{L}$, namely
\begin{equation}
(D_x+U(p_{10},q_{10};\lambda))M=M(D_x+U(p,q;\lambda)).
\end{equation}
Therefore, since $D_x\cdot M=M_x+MD_x$, the above equation implies that $M$ must obey equation (\ref{DMeq}).
\end{proof}

Although the above proposition can be used to determine $M$, this cannot be done in full generality without making choices for $M$ and analysing its dependence on the spectral parameter $\lambda$. 

Our first choice will be based on the following.

\begin{proposition}
A composition of two Darboux matrices for an operator of the form (\ref{L-operator}) is a Darboux matrix for the same operator.
\end{proposition}
\begin{proof}
Let $M$ and $K$ be Darboux matrices for an operator $\mathcal{L}$ of the form (\ref{L-operator}). Then, 
\begin{equation}\label{MLKL}
\tilde{\mathcal{L}}M=M\mathcal{L},\quad \hat{\mathcal{L}}K=K\mathcal{L}.
\end{equation}
by definition of the Darboux matrix. Now,
\begin{equation}
KM\mathcal{L}\stackrel{(\ref{MLKL})}{=}K\tilde{\mathcal{L}}M\stackrel{(\ref{MLKL})}{=}\hat{\tilde{\mathcal{L}}}KM.
\end{equation}
Therefore, $\hat{\tilde{\mathcal{L}}}=KM\mathcal{L}(KM)^{-1}$ which proves the statement.
\end{proof}

We define the \textit{rank} of a Darboux transformation\index{Darboux transformation(s)!rank}\index{rank of Darboux transformation} to be the rank of the matrix which appears as coefficient of the higher power of the spectral parameter.
In the next sections, we shall be assuming that the Darboux tranformations are of rank 1; in fact, in some examples, Darboux transformations of full rank can be written as composition of Darboux transformations of rank 1. 

The second choice is related to the form of the corresponding Lax operators. Since the Lax operators\index{Lax operator(s)} we deal with have rational dependence on the spectral parameter, we impose without loss of generality, that the same holds for matrix $M$ as well. Moreover, we employ any symmetries of the Lax operator, $\mathcal{L}$, as symmetries inherited to $M$. Specifically, if the the Lax operator\footnote{For simplicity of the notation, we sometimes omit the dependence on the potentials $p$ and $q$, i.e. $\mathcal{L}(p,q;\lambda)\equiv\mathcal{L}(\lambda)$.} $\mathcal{L}(\lambda)$ satisfies a relation of the form
\begin{equation}
\mathcal{L}(\lambda)=\Sigma(\lambda)\mathcal{L(\sigma(\lambda))}\Sigma(\lambda)^{-1},
\end{equation}
for some invertible function $\sigma(\lambda)$ and some invertible matrix $\Sigma(\lambda)$, then we shall be assuming that $M$ must obey the same relation, namely
\begin{equation}\label{Msym}
M(\lambda)=\Sigma(\lambda)M(\sigma(\lambda))\Sigma(\lambda)^{-1}.
\end{equation}
Relation (\ref{Msym}) imposes some restrictions on the form of matrix $M$ and reduces the number of functions involved in it.

Now, let $M$ be a Darboux matrix for the operator $\mathcal{L}$, and $\Psi=\Psi(x,\lambda)$ a fundamental solution of the linear equation 
\begin{equation}\label{LinearEq}
\mathcal{L}\Psi(x,\lambda)=0. 
\end{equation}
Then, we have the following.
\begin{proposition}\label{FS-Det}
Matrix $M$ maps fundamental solutions of (\ref{LinearEq}) to fundamental solutions of $\tilde{\mathcal{L}}\Psi=0$. Moreover, the determinant of $M$ is independent of $x$.
\end{proposition}
\begin{proof}
Let $M$ map $\Psi$ to $\Psi_{10}$, namely $\Psi_{10}=M\Psi$. Then, according to (\ref{DarbouxDef})
\begin{equation}
\tilde{\mathcal{L}}\Psi_{10}=M\mathcal{L}M^{-1}\Psi_{10}=M\mathcal{L}M^{-1}(M\Psi)=M\mathcal{L}\Psi=0,
\end{equation}
i.e. $\Psi_{10}$ is a solution of $\tilde{\mathcal{L}}\Psi=0$. Moreover, $\Psi_{10}$ is fundamental, since $\Psi$ is fundamental, $\det{M}\neq 0$ and $\Psi_{10}=M\Psi$.

Now, recall Liouville's formula\footnote{It is also known as Abel-Jacobi-Liouville identity.}\index{Liouville's formula}\index{Abel-Jacobi-Liouville identity} for solutions of the linear equation $L\Psi=0$, given by
\begin{equation}
\det\Psi(x,t;\lambda)=\det\Psi(x_0,t;\lambda)\exp\left(-\int^{x}_{x_0}\tr U(p(\xi),q(\xi);\lambda)d\xi\right).
\end{equation}
Since $U$ is traceless, from the above formula we deduce that the determinants of $\Psi$ and $\Psi_{10}$ are non-zero and independent of $x$. Hence, the relation $\Psi_{10}=M\Psi$ implies that $\partial_x(\det(M))=0$.
\end{proof}


\subsection{Discrete Lax pairs and discrete systems}
Starting with a fundamental solution of equation (\ref{LinearEq}), say $\Psi$, we can employ two Darboux matrices\index{Darboux matrix(-ces)}\index{Darboux matrix(-ces)} to derive two new fundamental solutions $\Psi_{10}$ and $\Psi_{01}$ as follows
\begin{equation} \label{dis-LP}
\Psi_{10}=M(p,q,p_{10},q_{10};\lambda)\Psi \equiv M\Psi,\quad \Psi_{01}=M(p,q,p_{01},q_{01};\lambda)\Psi \equiv K\Psi.
\end{equation}\index{discrete!Lax pair}
Then, a third solution can be derived in a purely algebraic way as shown in Figure \ref{bianchi}.

\begin{figure}[ht]
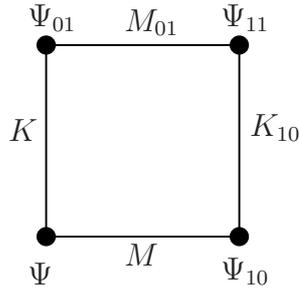

\centertexdraw{ \setunitscale 0.5
\linewd 0.02 \arrowheadtype t:F 
\htext(0 0.5) {\phantom{T}}
\move (-1 -2) \lvec (1 -2) 
\move(-1 -2) \lvec (-1 0) \move(1 -2) \lvec (1 0) \move(-1 0) \lvec(1 0)
\move (1 -2) \fcir f:0.0 r:0.1 \move (-1 -2) \fcir f:0.0 r:0.1
 \move (-1 0) \fcir f:0.0 r:0.1 \move (1 0) \fcir f:0.0 r:0.1  
\htext (-1.2 -2.5) {$\Psi$} \htext (.8 -2.5) {$\Psi_{10}$} \htext (-.2 -2.3) {$M$}
\htext (-1.2 .15) {$\Psi_{01}$} \htext (.8 .15) {$\Psi_{11}$} \htext (-.2 .1) {$M_{01}$}
\htext (-1.4 -1) {$K$} \htext (1.1 -1) {$K_{10}$}}
\caption{{Bianchi commuting diagram}} \label{bianchi}\index{Bianchi's!commuting diagram}
\end{figure}

Specifically, starting with two fundamental solutions of equation (\ref{LinearEq}), $\Psi_{10}$ and $\Psi_{01}$, one can construct two new fundamental solutions using Darboux matrices
\begin{equation*}
M(p_{10},q_{10},p_{11},q_{11};\lambda)\equiv K_{10} \quad M(p_{01},q_{01},p_{11},q_{11};\lambda)\equiv M_{01},
\end{equation*}
namely the following
\begin{eqnarray*}
&&\mathcal{X}:=M_{01}\Psi_{01}=M_{01}K\Psi\equiv M(p_{01},q_{01},p_{11},q_{11};\lambda)M(p,q,p_{01},q_{01};\lambda)\Psi,\\
&&\mathcal{Y}:=K_{10}\Psi_{10}=K_{10}M\Psi\equiv M(p_{10},q_{10},p_{11},q_{11};\lambda)M(p,q,p_{10},q_{10};\lambda)\Psi.
\end{eqnarray*}
Imposing that these two different solutions coincide, i.e. $\mathcal{X}=\mathcal{Y}=\Psi_{11}$ (see Figure \ref{bianchi}), the following condition must hold
\begin{equation}\label{CompCond}
M_{01}K-K_{10}M=0.
\end{equation}
If the latter condition is written out explicitly, it results in algebraic relations among the various potentials involved.

We can interpret the above construction in a discrete way. Particularly, let us assume that $p$ and $q$ are functions depending not only on $x$ but also on two discrete variables $n$ and $m$, i.e. $p=p(x;n,m)$ and $q=q(x;n,m)$. Furthermore, we define the \textit{shift operators} $\mathcal{S}$ and $\mathcal{T}$ acting on a function $f=f(n,m)$ as 
\begin{equation}
\mathcal{S}f(n,m)=f(n+1,m), \qquad \mathcal{T}f(n,m)=f(n,m+1).
\end{equation}
We shall refer to $\mathcal{S}$ and $\mathcal{T}$ as the \textit{shift operators in the} $n$ and the $m$-\textit{direction}, respectively. 

Now, we expect that the shift operators $\mathcal{S}$ and $\mathcal{T}$ commute with each other and with the differential operator $\partial_x$. In particular, we have the following.
\begin{proposition}
The shift operators $\mathcal{S}$ and $\mathcal{T}$ commute with each other and also with the partial differential operator on solutions of $\mathcal{L}\Psi=0$.
\end{proposition}
\begin{proof}
Let $\Psi$ be a solution of $\mathcal{L}\Psi=0$. For the commutativity of $\mathcal{S}$ and $\mathcal{T}$ we have that 
\begin{equation}
\mathcal{T}\mathcal{S}\Psi=\mathcal{T}(M\Psi)=M_{01}\Psi_{01}=M_{01}K\Psi,
\end{equation}
and on the other hand
\begin{equation}
\mathcal{S}\mathcal{T}\Psi=\mathcal{S}(K\Psi)=K_{10}\Psi_{10}=K_{10}M\Psi,
\end{equation}
which proves that $\left[S,T\right]\Psi=0$, due to (\ref{CompCond}).

Now, we prove the commutativity between the shift operator $\mathcal{S}$ and the partial differential operator, and the proof is exactly the same for $\mathcal{T}$. We basically need to show that $\mathcal{S}\Psi_x=(\mathcal{S}\Psi)_x$. Indeed,
\begin{equation}\label{Lrhs}
(\mathcal{S}\Psi)_x=\partial_x\Psi_{10}\stackrel{(\ref{dis-LP})}{=}(M\Psi)_x=M_x\Psi+M\Psi_x\stackrel{(\ref{L-operator})}{=}M_x\Psi-MU\Psi.
\end{equation}
On the other hand, we have that
\begin{equation}\label{Llhs}
\mathcal{S}\Psi_x\stackrel{(\ref{L-operator})}{=}-\mathcal{S}(U\Psi)=-U_{10}\Psi_{10}\stackrel{(\ref{dis-LP})}{=}-U_{10}M\Psi.
\end{equation}
Due to proposition \ref{DarbouxMatrix}, the right hand sides of (\ref{Llhs}) and (\ref{Lrhs}) are equal which completes the proof.
\end{proof}

In addition, we interpret the shifts of $p$ and $q$ with
\begin{equation}
p_{ij}=p(x;n+i,m+j),\quad q_{ij}=q(x;n+i,m+j), \quad p_{00}\equiv p,~~q_{00}\equiv q,
\end{equation}
respectively. 

In this notation, system (\ref{dis-LP}) can be considered as a discrete Lax pair, and equation (\ref{CompCond}) is nothing but its compatibility condition. Furthermore, the resulting polynomials from condition (\ref{CompCond}) define a system of partial difference equations for $p$ and $q$.

\begin{note}\normalfont
For the sake of simplicity, in the rest of the thesis, we adopt the following notation: the derivative with respect to $x$ of a scalar object with lower indices, say $\partial_x p_{ij}$, will be denoted by $p_{ij,x}$. Moreover, for a matrix with lower index, say $M_0$, with $M_{0,ij}$ we shall denote its $(i,j)$-element.
\end{note}

Now let $\textbf{u}$ be a solution of a system of difference equations, $\Delta(\textbf{u})=0$. Moreover, let $\overline{\textbf{u}}$ be given by
\begin{equation}\label{eqRe}
\overline{\textbf{u}}:=R_\epsilon (\textbf{u})= \textbf{u}+\epsilon\, r(\textbf{u}),
\end{equation}
were $\epsilon$ is an infinitesimal parameter. We have the following (\cite{MWX}).

\begin{definition}
We shall say that $R_\epsilon(\textbf{u})$ constitutes an \textit{infinitensimal symmetry}\index{infinitensimal symmetry} --or just a \textit{symmetry}\index{symmetry}-- of $\Delta(\textbf{u})=0$, if  $\,\overline{\textbf{u}}:=R_\epsilon(\textbf{u})$ is also a solution up to order $\epsilon^2$, namely
\begin{equation}\label{sym-bu}
\Delta(\overline{\textbf{u}})={\cal{O}}(\epsilon^2).
\end{equation}
\end{definition}

\begin{corollary}
Map $R_\epsilon(\textbf{u})$ defined in (\ref{eqRe}) is a symmetry of a system of difference equations, $\Delta(\textbf{u})$, if the following condition\footnote{Operator $D_{\Delta}:=\sum_{i,j\in\field{Z}}\frac{\partial \Delta}{\partial \textbf{u}_{ij}}\mathcal{S}^i\mathcal{T}^j$ is called the \textit{Fr\'echet derivative}\index{Fr\'echet derivative}.} 
\begin{equation}\label{FrechetDer}
\sum_{i,j\in\field{Z}}\frac{\partial \Delta}{\partial \textbf{u}_{ij}}\mathcal{S}^i\mathcal{T}^j r(\textbf{u})=0,
\end{equation}
is satisfied $\mod \langle \Delta(\textbf{u})=0 \rangle$.
\end{corollary}
\begin{proof}
If we expand $\Delta(\textbf{u}+\epsilon r(\textbf{u}))$ in series, then equating the $\epsilon$-terms, (\ref{sym-bu}) implies (\ref{FrechetDer}).
\end{proof}

In the next section we shall see for particular examples that, the derivation of Darboux matrices gives rise to particular differential-difference equations\index{differential-difference equation(s)} which possess first integrals.\index{first integral(s)} In some cases, the latter may be used to reduce the number of the dependent variables and derive scalar equations; some of them are of Toda type and some others are defined on a stencil of six or seven points. Additionally, the form of these systems allows us to pose an initial value problem on the staircase.

The derived system of differential-difference equations\index{differential-difference equation(s)} is, in general, of the form
\begin{equation}\label{Baecklund-x}
\mathcal{P}(p,q,\partial_xp,\partial_xq,\partial_xp_{10},\partial_xq_{10})=0.
\end{equation}
However, in all the examples of this thesis, system (\ref{Baecklund-x}) boils down to an evolutionary differential-difference system of equations of the form
\begin{equation}\label{Baecklund-x_local}
\textbf{p}_x=\mathcal{\textbf{Z}}(\textbf{p},\textbf{p}_{10}), \quad \textbf{p}=(p,q),
\end{equation}
and its shifted consequences.

For the above system of equations we have the following.
\begin{proposition}\label{d_sym}
The differential-difference equations (\ref{Baecklund-x_local})\index{differential-difference equation(s)} constitute generators of generalised symmetries\index{generalised symmetries} of the associated system of difference equations\index{difference equation(s)}.
\end{proposition}
\begin{proof}
Indeed,
\begin{equation}
\frac{d}{dx}(M_{01}K-K_{10}M)=M_{01,x}K+M_{01}K_x-K_{10,x}M-K_{10}M_x.
\end{equation}
Now, according to Proposition \ref{DarbouxMatrix} we have
\begin{eqnarray*}
&&M_x=MU-U_{10}M, \quad M_{01,x}=M_{01}U_{01}-U_{11}M_{01} \\
&&K_x=KU-U_{01}K, \quad K_{10,x}=K_{10}U_{10}-U_{11}K_{10}.
\end{eqnarray*}
Therefore,
\begin{equation}
\frac{d}{dx}(M_{01}K-K_{10}M)=U_{11}(K_{10}M-M_{01}K)+(M_{01}K-K_{10}M)U.
\end{equation}
The right-hand-side of the above is zero $\mod \langle K_{10}M-M_{01}K=0 \rangle$ which proves the statement.
\end{proof}

\begin{note}\normalfont
System (\ref{Baecklund-x_local}) is nothing but the $x$-part (spatial-part) of a B\"acklund transfo- rmation\index{B\"acklund transformation(s)} between solutions of the PDE associated to the Lax operator.
\end{note}

In the next section we derive Darboux matrices for particular PDEs of NLS type; the NLS equation, the DNLS equation and a deformation of the DNLS equation. 

In all the above cases, we find more than one Darboux matrix\index{Darboux matrix(-ces)}\index{Darboux matrix(-ces)}. At this point it is worth mentioning that, the interpretation of any pair of Darboux matrices as a discrete Lax pair does not always lead to a discrete integrable system.\index{integrable system(s)!discrete} On the contrary, in several cases, the compatibility condition (\ref{CompCond}) yields a trivial system. 

\section{NLS type equations}
In this section we study the Darboux transformations\index{Darboux transformation(s)} related to the NLS type equations discussed earlier.

As an illustrative example, we start with the NLS equation.

\subsection{The nonlinear Schr\"odinger equation}\label{sec-NLS}
The Lax operator \index{Lax operator(s)!for NLS equation}\index{nonlinear Schr\"odinger (NLS)!Lax operator} in this case is given by
\begin{subequations}\label{NLS-U}
\begin{align}
&\mathcal{L} := D_x + U(p,q;\lambda)=D_x +\lambda U^{1}+U^{0},
\intertext{where $U^{1}$ and $U^0$ are given by}
&U^1\equiv \sigma_3={\rm{diag}}(1,-1),\quad U^0=\left(\begin{array}{cc} 0 & 2p \\ 2q & 0\end{array}\right).
\end{align}
\end{subequations}
Operator $\mathcal{L}$ is the spatial part of the Lax pair for the nonlinear Schr\"odinger equation.\index{Lax pair!spatial part ($x$-part)!for NLS equation}

The NLS equation has the following scaling symmetry
\begin{equation}
p_{10} = \alpha \beta^{-1} p,\quad q_{10} = \beta \alpha^{-1} q.
\end{equation}
A spectral parameter independent Darboux matrix\index{nonlinear Schr\"odinger (NLS)!spectral parameter independent Darboux matrix}\index{Darboux matrix(-ces)!for NLS} corresponding to the above symmetry is given by the following constant matrix
\begin{equation}
M = \left(\begin{array}{cc} \alpha & 0 \\ 0 & \beta \end{array} \right),\quad \alpha \beta \ne 0.
\end{equation}

Now, we seek Darboux matrices which depend on the spectral parameter, $\lambda$. The simplest $\lambda$-dependence is a linear dependence. In particular, we have the following.

\begin{proposition}\label{M-NLS-Prop}
Let $M$ be a Darboux matrix for the Lax operator (\ref{NLS-U}) and suppose it is linear in $\lambda$. Let also $M$ define a Darboux transformation\index{Darboux transformation} of rank 1. Then, up to a gauge transformation,\index{gauge transformation} $M$ is given by\index{nonlinear Schr\"odinger (NLS)!Darboux matrix}\index{Darboux matrix(-ces)!for NLS}
\begin{equation}\label{M-NLS}
M(p,q_{10},f):=\lambda \left(\begin{array}{cc}1 & 0\\0 & 0\end{array}\right)+\left(\begin{array}{cc}f & p\\q_{10} & 1 \end{array}\right),
\end{equation}
where the potentials $p$ and $q$ satisfy the following differential-difference equations\index{differential-difference equation(s)}
\begin{subequations} \label{nls-comp-cond}
\begin{eqnarray}
&& \partial_x f=2 (pq-p_{10}q_{10}),\label{nls-comp-cond-1}\\
&& \partial_x p =2 (p f -p_{10}), \label{nls-comp-cond-2}\\
&& \partial_x q_{10}=2 ( q-q_{10}f).\label{nls-comp-cond-3}
\end{eqnarray}
\end{subequations}
Moreover, matrix (\ref{M-NLS}) degenerates to
\begin{equation}\label{M-degen}
M_c(p,f)=\lambda \left(\begin{array}{cc}1 & 0\\0 & 0\end{array}\right)+\left(\begin{array}{cc}f & p\\\frac{c}{p} & 0 \end{array}\right),\quad f=\frac{p_x}{2p}.
\end{equation}
\end{proposition}

\begin{proof}
Let us suppose that $M$ is of the form $M=\lambda M_{1}+M_{0}$. Substitution of $M$ to equation (\ref{DMeq}) implies a second order algebraic equation in $\lambda$. Equating the coefficients of the several powers of $\lambda$ equal to zero, we obtain the following system of equations
\begin{subequations}
\begin{align}
\lambda^2: & \quad \left[\sigma_3,M_1\right]=0,\label{M-syst-1}\\
\lambda^1:   & \quad M_{1,x}+\left[\sigma_3,M_0\right]+U^{0}_{10}M_1-M_1U^0=0,\label{M-syst-2}\\
\lambda^0: & \quad M_{0,x}+U^0_{10}M_0-M_0U^0=0,\label{M-syst-3}
\end{align}
\end{subequations}
where with $\left[\sigma_3,M_1\right]$ we denote the commutator of $\sigma_3$ and $M_1$.

Equation (\ref{M-syst-1}) implies that $M_1$ must be diagonal, i.e. $M_1=\diag(c_1,c_2)$. Then, we substitute $M_1$ to (\ref{M-syst-2}).

Now, for simplicity of the notation, we denote the $(1,1)$ and $(2,2)$ entries of $M_0$ by $f$ and $g$ respectively. Then, it follows from equation (\ref{M-syst-3}) that the entries of matrix $M_0$ must satisfy the following equations 
\begin{subequations} \label{fgpq10}
\begin{eqnarray}
\partial_x f&=&2 (M_{0,12}q-p_{10}M_{0,21}),\label{fgpq10-1}\\
\partial_x g&=&2 (M_{0,21}p-q_{10}M_{0,12}),\label{fgpq10-2}\\
\partial_x M_{0,12} &=&2 (p f -g p_{10}), \label{fgpq10-3}\\
\partial_x M_{0,21}&=&2 ( q g-q_{10}f).\label{fgpq10-4}
\end{eqnarray}
\end{subequations}

The off-diagonal part of (\ref{M-syst-2}) implies that the $(1,2)$ and $(2,1)$ entries of matrix $M_0$ are given by
\begin{equation}
M_{0,12}=c_1 p-c_2 p_{10}, \quad M_{0,21}=c_1q_{10}-c_2 q.
\end{equation}
Additionally, from the diagonal part of (\ref{M-syst-2}) we deduce that $c_{1,x}=c_{2,x}=0$. Since the Darboux transformation is of rank one, namely $\rank M_1=1$, one of the constants $c_i$, $i=1,2$ must be zero. Thus, after rescaling we can choose either $c_1=1$, $c_2=0$ or $c_1=0$, $c_2=-1$. These two choices correspond to gauge equivalent Darboux matrices.

Indeed, the choice $c_1=1$, $c_2=0$ implies $M_{0,12}=p$ and $M_{0,21}=q_{10}$. Moreover, (\ref{fgpq10-2}) implies that $g=\text{const.}=\alpha$, i.e.
\begin{equation}
M_0=\left(\begin{array}{cc}f & p\\
q_{10} & \alpha \end{array}\right).
\end{equation}
In this case the Darboux matrix is given by 
\begin{equation}\label{M-NLS-const.}
\lambda \left(\begin{array}{cc}1 & 0\\0 & 0\end{array}\right)+\left(\begin{array}{cc}f & p\\q_{10} & \alpha \end{array}\right),
\end{equation} 
where, according to (\ref{fgpq10}), its entries satisfy
\begin{equation}\label{nls-comp-cond-alpha}
\partial_x f=2 (pq-p_{10}q_{10}),\quad \partial_x p =2 (p f -\alpha p_{10}),\quad \partial_x q_{10}=2 ( cq -q_{10}f).
\end{equation}
Now, if $\alpha\neq 0$, it can be rescaled to $\alpha=1$ and thus the Darboux matrix in this case is given by (\ref{M-NLS}) where its entries obey (\ref{nls-comp-cond}).

Similarly, the second choice, $c_1=0$, $c_2=-1$, leads to the following Darboux matrix
\begin{equation}\label{M-NLS-2}
N(p_{10},q,g)=\lambda \left(\begin{array}{cc}0 & 0\\0 & -1\end{array}\right)+\left(\begin{array}{cc}1 & p_{10}\\q & g \end{array}\right).
\end{equation}
The above is gauge equivalent to (\ref{M-NLS}), since $\sigma_1 N(p_{10},q,g) \sigma_1^{-1}$ is of the form (\ref{M-NLS}).

Therefore, a linear in $\lambda$ Darboux matrix is given by (\ref{M-NLS}) where its entries obey the differential-difference equations\index{differential-difference equation(s)} (\ref{nls-comp-cond}).

In the case where $\alpha=0$, from (\ref{nls-comp-cond-alpha}) we deduce
\begin{equation}\label{pxq}
p_x=2fp,\quad q_{10,x}=-2fq_{10}.
\end{equation}
Thus, the Darboux matrix in this case is given by
\begin{equation}
M(p,q_{10},f):=\lambda \left(\begin{array}{cc}1 & 0\\0 & 0\end{array}\right)+\left(\begin{array}{cc}f & p\\q_{10} & 0 \end{array}\right).
\end{equation}
In addition, after an integration with respect to $x$, equations (\ref{pxq}) imply that $q_{10}=c/p$. Hence, the Darboux matrix in this case is given by (\ref{M-degen}).
\end{proof}

It is straightforward to show that system (\ref{nls-comp-cond}) admits the following first integral\index{first integral(s)}
\begin{equation} \label{nls-const}
\partial_x \left(f-p\,q_{10} \right)\,=\,0\,.
\end{equation}
which implies that $\partial_x\det M=0$.


\subsection{The derivative nonlinear Schr\"odinger equation: $\field{Z}_2$-reduction group}
The Lax operator in this case is given by
\begin{subequations} \label{SL2-Lax-Op} \index{Lax operator(s)!for DNLS equation}\index{derivative nonlinear Schr\"odinger (DNLS)!Lax operator}
\begin{align}
&\mathcal{L}=D_x+\lambda^{2} U^2+\lambda U^1,
\intertext{where}
U^2=\sigma_3=&\diag(1,-1) \quad \text{and} \quad U^1=\left(\begin{array}{cc} 0 & 2p \\ 2q & 0\end{array}\right).
\end{align}
\end{subequations}
This is the spatial part of the Lax pair for the DNLS\index{Lax pair!spatial part ($x$-part)!for DNLS equation} equation (\ref{DNLS}) \index{derivative nonlinear Schr\"odinger (DNLS)!equation}, and it is invariant under the transformation
\begin{equation}\label{reduction-Z2}\index{$\field{Z}_2$-reduction group}
s_1(\lambda): \mathcal{L}(\lambda) \rightarrow \sigma_{3}\mathcal{L}(-\lambda) \sigma_{3}.
\end{equation}
As a matter of fact, the above involution generates the reduction group which is isomorphic to the $\field{Z}_2$ group.

As in the case of the NLS equation, the DNLS equation has the following scaling symmetry
\begin{equation}
p_{10} = \alpha \beta^{-1} p,\quad q_{10} = \beta \alpha^{-1} q.
\end{equation}
A spectral parameter independent Darboux matrix\index{nonlinear Schr\"odinger (NLS)!spectral parameter independent Darboux matrix}\index{Darboux matrix(-ces)!for NLS} corresponding to the above symmetry is given by the following constant matrix
\begin{equation}\label{cDar}
M = \left(\begin{array}{cc} \alpha & 0 \\ 0 & \beta \end{array} \right),\quad \alpha \beta \ne 0.
\end{equation}

Now, we seek Darboux matrices with the same $\lambda$-dependence as the non-differential part of (\ref{SL2-Lax-Op}), namely of the form
\begin{equation}\label{M-Z2-gnef}
M=\lambda^2 M_2+\lambda M_1+M_0.
\end{equation}

\begin{lemma}\label{M2M1M0}
A second order matrix polynomial in $\lambda$, of the form (\ref{M-Z2-gnef}), is invariant under the involution (\ref{reduction-Z2}), iff $M_2$ and $M_0$ are diagonal matrices and $M_1$ is off-diagonal.
\end{lemma}
\begin{proof}
It is straightforward if we demand that $M$ satisfies the condition: $M(\lambda)=\sigma_3M(-\lambda)\sigma_3$.
\end{proof}

As mentioned earlier, we shall restrict ourselves to Darboux transformations of rank 1, and we have the following.

\begin{proposition}\label{Darboux-DNLScase}
Let $M$ in (\ref{M-Z2-gnef}) be a Darboux matrix for the Lax operator (\ref{SL2-Lax-Op}) that is invariant under the involution (\ref{reduction-Z2}) and $\rank{M_2}=1$. Then, up to a gauge transformation\index{gauge transformation}, $M$ is given by
\begin{equation} \label{DT-sl2-gen}\index{derivative nonlinear Schr\"odinger (DNLS)!Darboux matrix}\index{Darboux matrix(-ces)!for DNLS}
M(p,q_{10},f;c_1,c_2) := \lambda^{2}\left(\begin{array}{cc} f & 0\\ 0 & 0\end{array}\right)+\lambda\left(\begin{array}{cc} 0 & fp\\ fq_{10} & 0\end{array}\right)+\left(\begin{array}{cc} c_1 & 0\\ 0 & c_2 \end{array}\right),
\end{equation}
where $p$ and $q$ satisfy the following differential-difference equations\index{differential-difference equation(s)}
\begin{subequations} \label{sl2-D-sym-gen}
\begin{align}
\partial_x f &= 2 f \left(pq-p_{10}q_{10} \right),\label{sl2-D-sym-gen-1}\\
\partial_x p &= 2 p \left(p_{10}q_{10}-pq \right) - 2\frac{c_2 p_{10} - c_1 p}{f},\label{sl2-D-sym-gen-2}\\
\partial_x q_{10} &= 2 q_{10}\left(p_{10}q_{10}-pq \right) - 2\frac{c_1 q_{10}- c_2 q}{f}.\label{sl2-D-sym-gen-3}
\end{align}
\end{subequations}
\end{proposition}
\begin{proof}
According to Lemma (\ref{M2M1M0}), matrix $M_2$ should be diagonal. Additionally, since $\rank{M_2}=1$, we consider $M_2$ to be of the form $M_2=\diag(f,0)$. The choise $M_2=\diag(0,g)$ leads to a gauge equivalent Darboux transformation.

Now, \ref{DMeq} implies that $M$ satisfies the following system of equations
\begin{subequations}
\begin{align}
\lambda^4: & \quad \left[\sigma_3,M_2\right]=0,\label{MZ2-sys-1}\\
\lambda^3: & \quad \left[\sigma_3,M_1\right]+U^1_{10}M_2-M_2U^1=0,\label{MZ2-sys-2}\\
\lambda^2: & \quad M_{2,x}+\left[\sigma_3,M_0\right]+U^1_{10}M_1-M_1U^1=0,\label{MZ2-sys-3}\\
\lambda^1: & \quad M_{1,x}+U^1_{10}M_0-M_0U^1=0,\label{MZ2-sys-4}\\
\lambda^0: & \quad M_{0,x}=0.\label{MZ2-sys-5}
\end{align}
\end{subequations}
The first one, (\ref{MZ2-sys-1}), is satisfied automatically since $M_2$ is diagonal. Moreover, the last one, (\ref{MZ2-sys-5}), implies that $M_0$ is constant, i.e. $M_0=\diag(c_1,c_2)$, where $c_1,c_2\in\field{C}$.

From equation (\ref{MZ2-sys-2}) we determine the entries of the off-diagonal matrix $M_1$, namely
\begin{equation}
M_{1,12}=fp \quad \text{and} \quad M_{1,21}=fq_{10}.
\end{equation}
That is, $M$ is given by (\ref{DT-sl2-gen}).

Equation (\ref{MZ2-sys-3}) implies the first differential-difference equation (\ref{sl2-D-sym-gen-1}), while equation (\ref{MZ2-sys-4}) implies the following
\begin{equation}\label{compEq}
(fp)_x=2c_1p-2c_2p_{10} \quad \text{and} \quad (fq_{10})_x=2c_2q-2c_1q_{10}.
\end{equation}
With use of (\ref{sl2-D-sym-gen-1}), the above equations can be rewritten in the form (\ref{sl2-D-sym-gen-2}) and (\ref{sl2-D-sym-gen-3}), respectively.
\end{proof}
 
A first integral\index{first integral(s)} of the differential-difference equations (\ref{sl2-D-sym-gen}) is given by
\begin{equation} \label{sl2-D-con-det-gen}
\partial_x \left(f^{2}pq_{10}-c_2 f\right)=0.
\end{equation}
The above integral guarantees that the determinant of the Darboux matrix (\ref{DT-sl2-gen}) is independent of $x$.

\begin{remark}\normalfont
If the constants $c_1$ and $c_2$ in (\ref{DT-sl2-gen}) are non-zero, then we can rescale them to 1, by composing with a Darboux matrix (\ref{cDar}) and changing $f\rightarrow f\alpha^{-1}$ and $q_{10}\rightarrow q_{10}\alpha\beta^{-1}$.
\end{remark}

However, if either of $c_1$ or $c_2$ is zero, we can bring the differential-difference equations (\ref{sl2-D-sym-gen}) into polynomial form.

\subsubsection{Case I: $c_1=c_2=0$. Modified Volterra chain}
In this case, from equations (\ref{compEq}), we obtain $(fp)_x=(fq_{10})_x=0$. An integration of the latter implies that $f=1/p$ and $q_{10}=p$, where we have set the constants of integration equal to 1. Hence, the Darboux matrix (\ref{DT-sl2-gen}) degenerates to
\begin{equation} \label{DT-sl2-degen}\index{derivative nonlinear Schr\"odinger (DNLS)!Darboux matrix}\index{Darboux matrix(-ces)!for DNLS}
M_{deg}(p) := \lambda^{2}\left(\begin{array}{cc} 1/p & 0\\ 0 & 0\end{array}\right)+\lambda\left(\begin{array}{cc} 0 & 1\\ 1 & 0\end{array}\right).
\end{equation} 

Moreover, the corresponding differential-difference equations become
\begin{equation} \label{Z2-case1-BT}
q_{10} = p \quad \partial_x p=2p^2\left(p_{10}-p_{-10} \right),
\end{equation}
and the first integral (\ref{sl2-D-con-det-gen})\index{first integral(s)} holds identically. The resulting differential-difference equations (\ref{Z2-case1-BT}) constitute the modified Volterra chain \cite{Yamilov}.\index{modified Volterra chain}

\subsubsection{Case II: $c_1=1$ and $c_2=0$}
In this case the Darboux matrix simplifies to
\begin{equation} \label{Z2-case2-DM-M}
M(p,q_{10},f) := \lambda^2 \left(\begin{array}{cc} f & 0 \\ 0 & 0 \end{array} \right) + \lambda \left(\begin{array}{cc} 0 & f p \\ f q_{10} & 0 \end{array} \right) + \left(\begin{array}{cc} 1 & 0 \\ 0 & 0\end{array} \right).
\end{equation}

The differential-difference equations (\ref{sl2-D-sym-gen}) become
\begin{subequations} \label{Z2-case2-BT}
\begin{align}
\partial_x f &= 2 f (pq-p_{10} q_{10}),\\
\partial_x p &= 2 p \left( p_{10} q_{10}-p q\right) + \frac{2 p}{f}, \\
\partial_x q_{10} &= 2 q_{10} \left(p_{10} q_{10}-p q\right) -\frac{2 q_{10}}{f},
\end{align}
\end{subequations}
and the first integral (\ref{sl2-D-con-det-gen})\index{first integral(s)} can be rewritten as
\begin{equation} \label{Z2-case2-fi}
\partial_x\left(f^2 pq_{10}\right) = 0.
\end{equation}

Now, in order to express $f$ in terms of $p$ and $q$, avoiding any square roots, we make the following point transformation
\begin{equation}
p = u^2,\quad q = v_{-10}^2.
\end{equation}
Thus, the first integral\index{first integral(s)} (\ref{Z2-case2-fi}) implies $f^2u^2v^2= 1$ and, subsequently, $f$ is given by $f=\pm 1/uv$. Moreover, for this $f$, system (\ref{Z2-case2-BT}) can be rewritten in a polynomial form as
\begin{equation}
\partial_x u = u (u_{10}^2 v^2-u^2 v_{-10}^2) \pm u^2 v,\quad \partial_x v = v (u_{10}^2 v^2 - u^2 v_{-10}^2) \mp u v^2.
\end{equation}

\subsection{A deformation of the derivative nonlinear Schr\"odinger equa-\\tion: Dihedral reduction group}
In this case, the Lax operator\footnote{The full Lax pair of the associated PDE can be found in \cite{BuryPhD, Bury-Sasha}.} is given by
\begin{subequations}\label{L-dDNLS}\index{Lax operator(s)!for deformation of DNLS equation}\index{deformation of DNLS!Lax operator}
\begin{align}
&\mathcal{L}=D_x+\lambda^2U^2+\lambda U^1+\lambda^{-1} U^{-1}+\lambda^{-2}U^{-2},
\intertext{where}
U^2&\equiv -U^{-2}=\sigma_3, \quad U^{1}\equiv (U^{-1})^T=\left(\begin{array}{cc} 0 & 2\,p\\ 2\,q & 0\end{array}\right).
\end{align}
\end{subequations}

Operator (\ref{L-dDNLS}) is the spatial part of the Lax pair of the deformation of the DNLS equation\index{Lax pair!spatial part ($x$-part)!for the deformation of DNLS equation} (\ref{dDNLS}) \index{deformation of DNLS!equation}, and it is invariant under the following transformations
\begin{subequations}  \label{reduct_group-D2}
\begin{align}
 s_1(\lambda)&:\mathcal{L}(\lambda) \rightarrow \sigma_3\mathcal{L}(-\lambda)\sigma_3,\\
 s_2(\lambda)&:\mathcal{L}(\lambda) \rightarrow \sigma_1\mathcal{L}(\frac{1}{\lambda})\sigma_1,\quad \sigma_1=\left(\begin{array}{cc} 0 & 1 \\ 1 & 0 \end{array} \right).
\end{align}
\end{subequations}
The above involutions generate the reduction group which is isomorphic to $\field{Z}_2\times\field{Z}_2\cong \field{D}_2$ (dihedral group)\index{$\field{D}_2$-reduction group}\index{reduction group!dihedral}. 

Equation (\ref{dDNLS}) has the obvious symmetry $(p,q)\rightarrow (-p,-q)$. A spectral parameter independent Darboux matrix for (\ref{L-dDNLS})\index{deformation of DNLS!spectral parameter independent Darboux matrix}\index{Darboux matrix(-ces)!for deformation of DNLS}, which corresponds to the latter symmetry, is given by $\sigma_3$. For a $\lambda$-dependent Darboux matrix, we seek a matrix with the same dependence on the spectral parameter as the non-differential part of $\mathcal{L}$ in (\ref{L-dDNLS}). Specifically, we are seeking for $M$ of the form
\begin{equation}\label{M-DNLS-Ansatz}
M=\lambda^2M_2+\lambda M+M_0+\lambda^{-1}M_{-1}+\lambda^{-2}M_{-2}.
\end{equation}

\begin{lemma}\label{M-DNLS-red}
A matrix of the form (\ref{M-DNLS-Ansatz}), which is invariant under the involutions $s_1$ and $s_2$ in (\ref{reduct_group-D2}), is given by
\begin{equation*}
M=\lambda^2\left(\begin{array}{cc} \alpha & 0\\ 0 & \beta\end{array}\right)+\lambda\left(\begin{array}{cc} 0 & \gamma\\ \delta & 0\end{array}\right)+
\left(\begin{array}{cc} \mu & 0\\ 0 & \nu\end{array}\right)+
\lambda^{-1}\left(\begin{array}{cc} 0 & \delta\\ \gamma & 0\end{array}\right)+\lambda^{-2}\left(\begin{array}{cc} \beta & 0\\ 0 & \alpha\end{array}\right).
\end{equation*}
\end{lemma}
\begin{proof}
Relation $M(\lambda)=\sigma_3 M(-\lambda)\sigma_3$ implies that $M_0$, $M_2$ and $M_{-2}$ must be diagonal, whereas $M_1$ and $M_{-1}$ must be off-diagonal. Moreover, from relation $M(\lambda)=\sigma_1 M(1/\lambda)\sigma_1$, we obtain: $M_{2,11}=M_{-2,22}$, $M_{2,22}=M_{-2,11}$, $M_{1,12}=M_{1,21}$ and $M_{1,21}=M_{1,12}$.
\end{proof}

\begin{proposition}\index{Darboux matrix(-ces)!for deformation of DNLS}\index{deformation of DNLS!Darboux matrix}
Let $M$ be a Darboux matrix of the form (\ref{M-DNLS-Ansatz}) for the Lax operator (\ref{L-dDNLS}) and $\rank M_2=1$. Moreover, suppose that it is invariant under the involutions $s_1$ and $s_2$. Then, up to a gauge transformation\index{gauge transformation} and an additive constant, $M$ is given by
\begin{equation}  \label{Dih-M}
M(p,q_{10},f,g) := f\left(\left(\begin{array}{cc} \lambda^2&0 \\0 & \lambda^{-2} \end{array}\right) + \lambda  \left(\begin{array}{cc} 0 & p \\q_{10} & 0 \end{array}\right) + g I + \lambda^{-1}  \left(\begin{array}{cc} 0& q_{10} \\ p &0 \end{array}\right)\right)
\end{equation}
where $I$ is the identity matrix and $p$, $q$, $f$ and $g$ satisfy the following system of differential-difference equations\index{differential-difference equation(s)}
\begin{subequations}\label{Dih-D-cond}
\begin{eqnarray} 
\partial_x p &=& 2\Big((p_{10} q_{10}-pq) p + (p-p_{10}) g  + q- q_{10}\Big), \\
\partial_x q_{10} &=& 2\Big((p_{10} q_{10}-pq) q_{10}+(q-q_{10}) g + p- p_{10} \Big),\\
\partial_x g &=& 2 \Big((p_{10} q_{10}-pq) g + (p-p_{10})p +  (q-q_{10}) q_{10} \Big),\\
\partial_x f &=& 2 (pq-p_{10} q_{10}) f.
\end{eqnarray}
\end{subequations}
\end{proposition}
\begin{proof}
From equation (\ref{DMeq}) we deduce the system of equations 
\begin{subequations}\label{eqDar-dDNLS}
\begin{align}
\lambda^4: & \quad\left[\sigma_3,M_2\right]=0,\label{eqDar-dDNLS-1}\\
\lambda^3: & \quad\left[\sigma_3,M_1\right]+U^1_{10}M_2-M_2U^1=0,\label{eqDar-dDNLS-2}\\
\lambda^2: & \quad M_{2,x}+U^1_{10}M_1-M_1U^1=0,\label{eqDar-dDNLS-3}\\
\lambda^1: & \quad M_{1,x}+\left[\sigma_3,M_{-1}\right]+U^1_{10}M_0-M_0U^1+U^{-1}_{10}M_2-M_2U^{-1}=0,\label{eqDar-dDNLS-4}\\
\lambda^0: & \quad M_{0,x}+U^1_{10}M_1-M_{-1}U^{1}+U^{-1}_{10}M_1-M_1U^{-1}=0,\label{eqDar-dDNLS-5}
\end{align}
\end{subequations}
for the entries of $M$. Because of the symmetry, the negative powers of $\lambda$ correspond to the same equations as the above.

Equation (\ref{eqDar-dDNLS-1}) is automatically satisfied as $M_2$ is diagonal. From equation (\ref{eqDar-dDNLS-2}) we identify the entries $\gamma$ and $\delta$ of matrix $M_1$ in terms of the potentials, namely
\begin{equation}\label{gamma-delta}
\gamma=\alpha p-\beta p_{10} \quad \text{and} \quad \delta=\alpha q_{10}-\beta q,
\end{equation}
whereas from (\ref{eqDar-dDNLS-4}) we can express their derivatives in terms of the potentials and the entries of matrices $M_2$ and $M_0$ as
\begin{subequations}\label{Der-gamma-delta}
\begin{align}
\gamma_x&=2\left(\alpha q-\beta q_{10}+\mu p-\nu p_{10}-\delta\right), \\
\delta_x&=2\left(\beta p-\alpha p_{10}+\nu q-\mu q_{10}+\gamma\right).
\end{align}
\end{subequations}

Now, from equation (\ref{eqDar-dDNLS-3}) we obtain the derivatives of $\alpha$ and $\beta$ in terms of the potentials, $\gamma$ and $\delta$, namely
\begin{equation}\label{der-alpha-beta}
\alpha_x=2(\gamma q-\delta p_{10}) \quad \text{and} \quad \beta_x=2(\delta p-\gamma q_{10}).
\end{equation}
Substituting (\ref{gamma-delta}) into the above equations, we deduce that $\alpha\beta=$const. Moreover, since $\rank M_2=1$, either $\alpha$ or $\beta$ has to be zero. Thus, we can choose $\beta=0$ and, then, $\alpha$ can be any arbitrary function of $x$, say $\alpha=f(x)$. The case $\alpha=0$ and $\beta = f(x)$ leads to a Darboux matrix which is gauge equivalent to the former. 

Now, $\alpha=f(x)$ and $\beta=0$  imply that $\gamma=fp$ and $\delta=fq_{10}$ from (\ref{gamma-delta}).

The last equation of (\ref{eqDar-dDNLS}) implies
\begin{equation}\label{mu-nu}
\mu_x=\nu_x=2f\left(p(p-p_{10})+q(q-q_{10})\right).
\end{equation}
Therefore, after a simple integration, we deduce that $\mu=\nu$ (where we have assumed that the integration constant is zero).

Now, the first equation of (\ref{der-alpha-beta}) can be rewritten as
\begin{equation}
f_x=2f(pq-p_{10}q_{10}).
\end{equation}
Using the above, and since we have expressed $\gamma$ and $\delta$ in terms of $f$, $p$ and $q_{10}$, equations (\ref{Der-gamma-delta}) can be rewritten as
\begin{subequations}\label{pxq10x}
\begin{align}
p_x&=2\Big((p_{10}q_{10}-pq)p+q-q_{10}+\frac{\mu}{f}(p-p_{10})\Big),\\
q_{10,x}&=2\Big((p_{10}q_{10}-pq)q_{10}+p-p_{10}+\frac{\mu}{f}(q-q_{10})\Big).
\end{align}
\end{subequations}
Now, without any loss of generality, we can set\footnote{This choice was made in order to retrieve polynomial expressions in (\ref{pxq10x}).} $\mu=fg$, which proves the statement.
\end{proof}

It is straightforward to show that the quantities 
\begin{equation} \label{dih-D-fi}
\Phi_1 := f^2 \left( g- pq_{10} \right), \quad \Phi_2 := f^2 \left(g^2+1 - p^2-q_{10}^2\right),
\end{equation}
are first integrals\index{first integral(s)} for the system of equations (\ref{Dih-D-cond}), namely $\partial_x \Phi_i = 0$, $i=1,2$. In fact, these first integrals\index{first integral(s)} imply that matrix $M$ has constant determinant, since
\begin{equation}
\det M = \left(\lambda^2 +\lambda^{-2} \right) \Phi_1 + \Phi_2.
\end{equation}

\section{Derivation of discrete systems and initial value prob-\\lems}
In this section we employ the Darboux matrices derived in the previous section to derive discrete integrable systems.\index{integrable system(s)!discrete} We shall present only the pairs of Darboux matrices which lead to genuinely non-trivial discrete integrable systems.\index{integrable system(s)!discrete}
For these systems we consider an initial value problem on the staircase.

\subsection{Nonlinear Schr\"odinger equation and related discrete systems}
Having derived two Darboux matrices\index{Darboux matrix(-ces)} for operator (\ref{NLS-U}), we focus on the one given in (\ref{M-NLS}) and consider the following discrete Lax pair
\begin{equation}\label{dLaxP}
\Psi_{10} = M \Psi,\quad \Psi_{01} = K\Psi,
\end{equation}
where $M$ and $K$ are given by
\begin{subequations} \label{NLS-disc-LP}
\begin{align}
&M\equiv M(p,q_{10},f)=\lambda \left(\begin{array}{cc} 1 & 0\\0 & 0 \end{array}\right)+\left(\begin{array}{cc} f & p\\q_{10} & 1\end{array}\right),\\
&K\equiv M(p,q_{01},g)= \lambda \left(\begin{array}{cc} 1 & 0\\ 0 & 0\end{array}\right)+\left(\begin{array}{cc} g & p\\ q_{01} & 1\end{array}\right).
\end{align}
\end{subequations}
The compatibility condition of (\ref{NLS-disc-LP}) results to
\begin{subequations} \label{nls-comp}
\begin{eqnarray}
 f_{01} -f - \left( g_{10}-g\right) = 0,&&\label{nls-comp-1}\\
 f_{01}\,g-fg_{10}-p_{10}q_{10}+p_{01}q_{01}=0,&&\\
 p \left(f_{01}-g_{10} \right)-p_{10}+p_{01}=0,&&\\
 q_{11}\left(f-g\right)-q_{01}+q_{10}=0.&&
\end{eqnarray}
\end{subequations}
This system can be solved either for $(p_{01},q_{01},f_{01},g)$ or for $(p_{10},q_{10},f,g_{10})$. In either of these cases, we derive two solutions. The first one is
\begin{equation} \label{triv-sol}
p_{10} = p_{01},\quad q_{10} = q_{01}, \quad f = g,\quad g_{10} = f_{01},
\end{equation}
which is trivial and corresponds to $M(p,q_{10},f)=M(p,q_{01},g)$.

The second solution is given by
\begin{subequations} \label{LPcompatEq}
\begin{eqnarray}
\hspace{-.3cm}&& p_{01} = \frac{q_{10} p^2 + (g_{10} - f) p + p_{10}}{1+p q_{11}},\quad 
q_{01} = \frac{p_{10}{ q_{11}}^{2} + (f-g_{10}) q_{11} + q_{10}}{1+p q_{11}},  \\
\hspace{-.3cm}&& f_{01} = \frac{q_{11} (p_{10} + p g_{10}) + f - p q_{10}}{1+ p q_{11}},\quad 
g = \frac{q_{11} (p f- p_{10}) + g_{10}+p q_{10}}{1+pq_{11}}. 
\end{eqnarray}
\end{subequations}

The above system has some properties which take their rise in the derivation of the Darboux matrix.\index{Darboux matrix(-ces)} In particular, we have the following.

\begin{proposition}
System (\ref{LPcompatEq}) admits two first integrals,\index{first integral(s)} $\mathcal{F}:=f-pq_{10}$ and $\mathcal{G}:=g-pq_{01}$, and the following conservation law\index{conservation law(s)}
	\begin{equation}\label{conlaw}
	(\mathcal{T}-1)f=(\mathcal{S}-1)g
	\end{equation}
\end{proposition}
\begin{proof}
Relation (\ref{nls-const}) suggests that
	\begin{equation}\label{1stInts}
	({\cal{T}}-1)\left(f-pq_{10}\right)=0\quad {\mbox{and}} \quad ({\cal{S}}-1)\left(g-pq_{01}\right)=0,
	\end{equation}
which can be verified with straightforward calculation, using equations (\ref{nls-comp-1}). Thus, $F=f-pq_{10}$ and $G:=g-pq_{01}$ are first integrals.\index{first integral(s)} Moreover, equation (\ref{nls-comp-1}) can be written in the form of the conservation law (\ref{conlaw}).\index{conservation law(s)}
\end{proof}

\begin{corollary}
The following relations hold.
\begin{equation} \label{nls-fi}
f-pq_{10}=\alpha(n)\quad {\mbox{and}} \quad g-pq_{01}=\beta(m).
\end{equation}
\end{corollary}

\begin{remark}\normalfont
In view of relations (\ref{nls-fi}), we can interpret functions $f$ and $g$ as being given on the edges of the quadrilateral where system (\ref{LPcompatEq}) is defined, and, consequently, consider system (\ref{LPcompatEq}) as a vertex-bond system \cite{HV}. 
\end{remark}

\subsubsection{Initial value problem on the staircase}
Our choice to solve system (\ref{nls-comp}) for $p_{01}$, $q_{01}$, $f_{01}$ and $g$ is motivated by the initial value problem related to system (\ref{LPcompatEq}). Suppose that initial values for $p$ and $q$ are given on the vertices along a staircase as shown in Figure \ref{fig-ivp}. Functions $f$ and $g$ are given on the edges of this initial value configuration in a consistent way with the first integrals\index{first integral(s)} (\ref{nls-fi}). In particular, horizontal edges carry the initial values of $f$ and vertical edges the corresponding ones of $g$.

\begin{figure}[ht]
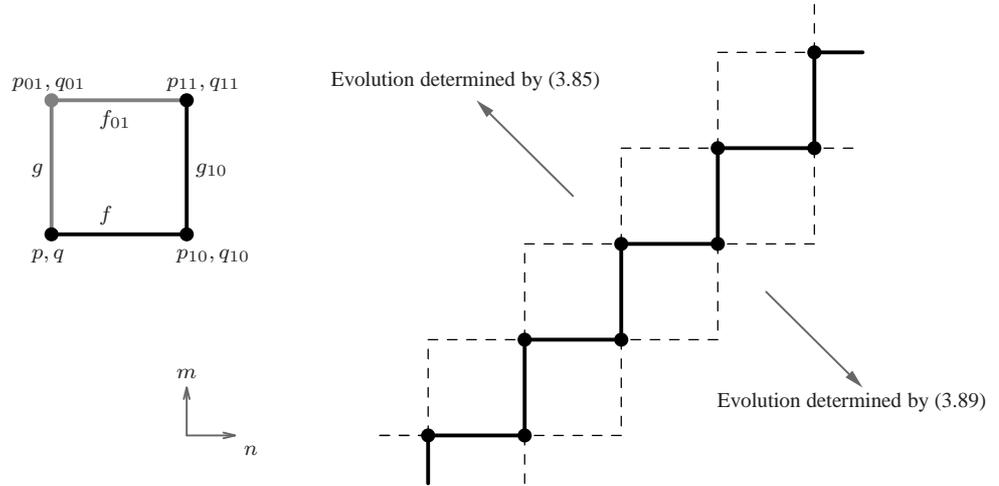

\centertexdraw{
\setunitscale 0.5
\move(-4.5 -2) \linewd 0.02 \setgray 0.4 \arrowheadtype t:V \arrowheadsize l:.12 w:.06 \avec(-4.5 -1.5) 
\move(-4.5 -2) \arrowheadtype t:V  \avec(-4 -2)
\arrowheadsize l:.20 w:.10
\move(-.5 .5) \linewd 0.02 \setgray 0.4 \arrowheadtype t:F \avec(-1.5 1.5) 
\move(1.5 -.5) \linewd 0.02 \setgray 0.4 \arrowheadtype t:F \avec(2.5 -1.5) 
\setgray 0.5
\linewd 0.04 \move (-4.5 1.5)  \lvec (-5.9 1.5) \lvec (-5.9 .1)
\move (-5.9 1.5) \fcir f:0.5 r:0.075
\htext (-6.1 .7) {\scriptsize{$g$}}
\htext (-4.4 .7) {\scriptsize{$g_{10}$}}
\setgray 0.0
\linewd 0.04 \move (-2 -2.5) \lvec (-2 -2) \lvec (-1 -2) \lvec (-1 -1) \lvec (0 -1) \lvec (0 0) \lvec (1 0) \lvec(1 1) \lvec (2 1) \lvec(2 2) \lvec(2.5 2)
\move (-5.9 .1) \lvec (-4.5 .1) \lvec (-4.5 1.5)  
\linewd 0.015 \lpatt (.1 .1 ) \move (-2 -2) \lvec (-2 -1) \lvec(-1 -1) \lvec (-1 0) \lvec (0 0) \lvec (0 1) \lvec(1 1) \lvec (1 2) \lvec (2 2) \lvec (2 2.5)
\move(-2.5 -2) \lvec(-2 -2) \move(-2.5 -2) \lvec(-2 -2)
\move (-1 -2.5) \lvec (-1 -2) \lvec(0 -2) \lvec(0 -1) \lvec(1 -1) \lvec(1 0) \lvec(2 0) \lvec(2 1) \lvec(2.5 1)
\move (-2 -2) \fcir f:0.0 r:0.075 \move (-1 -2) \fcir f:0.0 r:0.075
\move (-1 -1) \fcir f:0.0 r:0.075 \move (0 -1) \fcir f:0.0 r:0.075
\move (0 0) \fcir f:0.0 r:0.075 \move (1 0) \fcir f:0.0 r:0.075  
\move (1 1) \fcir f:0.0 r:0.075 \move (2 1) \fcir f:0.0 r:0.075
\move (2 2) \fcir f:0.0 r:0.075
\move (-5.9 .1) \fcir f:0.0 r:0.075 \move (-4.5 .1) \fcir f:0.0 r:0.075 \move (-4.5 1.5) \fcir f:0.0 r:0.075
\htext (-3.9 -2.2) {\scriptsize{$n$}}
\htext (-4.6 -1.4) {\scriptsize{$m$}}
\htext (-6.1 -.2) {\scriptsize{$p,q$}}
\htext (-4.6 -.2) {\scriptsize{$p_{10},q_{10}$}}
\htext (-5.4 0.2) {\scriptsize{$f$}}
\htext (-5.4 1.2) {\scriptsize{$f_{01}$}}
\htext (-6.3 1.6) {\scriptsize{$p_{01},q_{01}$}}
\htext (-4.7 1.6) {\scriptsize{$p_{11},q_{11}$}}
\htext (-3 1.6) {{\scriptsize{Evolution determined by (\ref{LPcompatEq})}}}
\htext (1 -1.75) {{\scriptsize{Evolution determined by (\ref{LPcompatEq-2})}}}
}
\caption{{Initial value problem and direction of evolution}} \label{fig-ivp}
\end{figure}

With these initial conditions, the values of $p$ and $q$ can be uniquely determined at every vertex of the lattice, while $f$ and $g$ on the corresponding edges. This is obvious from the rational expressions (\ref{LPcompatEq}) defining the evolution above the staircase, cf. Figure \ref{fig-ivp}. 

For the evolution below the staircase, one has to use 
\begin{subequations}\label{LPcompatEq-2}
\begin{eqnarray} 
p_{10} &=& \frac{q_{01} p^2 + (f_{01} - g) p + p_{01}}{1+p\, q_{11}},\quad 
q_{10} = \frac{p_{01}{ q_{11}}^{2} + (g-f_{01})q_{11} + q_{01}}{1+pq_{11}},\\
g_{10} &=& \frac{q_{11} (p_{01} + p f_{01}) + g - p q_{01}}{1+ p q_{11}},\quad 
f = \frac{q_{11} (p g- p_{01}) + f_{01}+p q_{01}}{1+p\,q_{11}},
\end{eqnarray}
\end{subequations}
which uniquely defines the evolution below the staircase as indicated in Figure \ref{fig-ivp}.

\begin{remark}\normalfont
We could consider more general initial value configurations of staircases of lengths $\ell_1$ and $\ell_2$ in the $n$ and $m$ lattice direction, respectively. Such initial value problems are consistent with evolutions (\ref{LPcompatEq}), (\ref{LPcompatEq-2}) determining the values of all fields uniquely at every vertex and edge of the lattice.
\end{remark}

\subsubsection{Derivation of an Adler-Yamilov type of system}
Now, using first integrals\index{first integral(s)} we can reduce system (\ref{LPcompatEq}) to an \textit{Adler-Yamilov type}\index{Adler-Yamilov!type of system} of system as those in \cite{Adler-Yamilov}. Specifically, we have the following.

\begin{proposition}
System (\ref{LPcompatEq}) can be reduced to the following non-autonomous Adler-Yamilov type of system for $p$ and $q$:
\begin{equation} \label{nls-pq-sys}
p_{01}=p_{10}-\frac{\alpha(n)-\beta(m)}{1+ pq_{11}}p,\quad q_{01}=q_{10}+\frac{\alpha(n)-\beta(m)}{1+ pq_{11}}q_{11}.
\end{equation}
\end{proposition}
\begin{proof}
The proof is straightforward if one uses relations (\ref{nls-fi}) to replace $f$ and $g$ in system (\ref{LPcompatEq}).
\end{proof}

\subsubsection{Derivation of discrete Toda equation}
Now, we will use two different Darboux matrices associated with the NLS equation to construct the discrete Toda equation \cite{Suris}\index{discrete!Toda equation}\index{Toda (type) equation(s)}. 

In fact, we introduce a discrete Lax pair as (\ref{dLaxP}), with $M=M_1(p,f)$ in (\ref{M-degen}) and $K=M(p,q_{01},g)$ in (\ref{M-NLS}). That is, we consider the following system
\begin{subequations}
\begin{align}
\Psi_{10} &= \left(\lambda \left(\begin{array}{cc} 1 & 0\\0 & 0 \end{array}\right)+\left(\begin{array}{cc} f & p\\ \frac{1}{p} & 0 \end{array}\right)\right) \Psi,\\
\Psi_{01} &= \left( \lambda \left(\begin{array}{cc} 1 & 0\\ 0 & 0\end{array}\right)+\left(\begin{array}{cc} g & p\\ q_{01} & 1\end{array}\right) \right)\Psi,
\end{align}
\end{subequations}
and impose its compatibility condition.

From the coefficient of the $\lambda$-term in the latter condition we extract the following equations
\begin{subequations}
\begin{align}
f-f_{01}&=g-g_{10},\label{l1-term-1}\\
p_{01}&=\frac{1}{q_{11}}\label{l1-term-2}.
\end{align}
\end{subequations}
Additionally, the $\lambda^0$-term of the compatibility condition implies
\begin{subequations}
\begin{align}
f_{01}g-g_{10}f&=\frac{p_{10}}{p}-p_{01}q_{01},\label{l0-term-1}\\
g_{10}-f_{01}&=\frac{p_{01}}{p},\label{l0-term-2}\\
g-f&=\frac{p_{01}}{p}.\label{l0-term-3}
\end{align}
\end{subequations}

Now, recall from the previous section that, using (\ref{nls-fi}), the quantities $g$ and $g_{10}$ are given by
\begin{equation}
g=\beta(m)+pq_{01} \quad \text{and} \quad g_{10}=\beta(m)+p_{10}q_{11}.
\end{equation}
We substitute $g$ and $g_{10}$ into (\ref{l0-term-2}) and (\ref{l0-term-3}), and then replace $p$ and its shifts using (\ref{l1-term-2}). Then, we can express $f$ and $f_{01}$ in terms of the potential $q$ and its shifts:
\begin{subequations}\label{ff01}
\begin{align}
f&=\frac{q_{01}}{q_{10}}-\frac{q_{10}}{q_{11}}+\beta(m),\label{ff01-1}\\
f_{01}&=\frac{q_{11}}{q_{20}}-\frac{q_{10}}{q_{11}}+\beta(m).\label{ff01-2}
\end{align}
\end{subequations}

\begin{proposition}
The compatibility of system (\ref{ff01}) yields a fully discrete Toda type equation.
\end{proposition}\index{discrete!Toda equation}\index{Toda (type) equation(s)}
\begin{proof}
Applying the shift operator $\mathcal{T}$ on both sides of (\ref{ff01-1}) and demanding that its right-hand side agrees with that of (\ref{ff01-2}), we obtain
\begin{equation}
\frac{q_{11}}{q_{20}}-\frac{q_{02}}{q_{11}}+\frac{q_{11}}{q_{12}}-\frac{q_{10}}{q_{11}}=\beta(m+1)-\beta(m).
\end{equation}
Then, we make the transformation
\begin{equation}
q\rightarrow \exp(-w_{-1,-1}),
\end{equation}
which implies the following discrete Toda type equation \index{discrete!Toda equation}\index{Toda (type) equation(s)}
\begin{equation}\label{DToda}\index{discrete!Toda equation}\index{Toda (type) equation(s)}
e^{w_{1,-1}-w}-e^{w-w_{-1,1}}+e^{w_{0,1}-w}-e^{w-w_{0,-1}}=\beta(m+1)-\beta(m),
\end{equation}
and proves the statement.
\end{proof}

\begin{remark}\normalfont
The discrete Toda equation (\ref{DToda}) can be written in the form of a conservation law,\index{conservation law(s)}
\begin{equation}
(\mathcal{S}-1)e^{w_{0,-1}-w_{-10}}=(\mathcal{T}-1)(e^{w_{0,-1}-w_{-10}}-e^{w-w_{0,-1}}+\beta(m)).
\end{equation}
\end{remark}


\subsection{Derivative nonlinear Schr\"odinger equation and related discrete systems}
Now we consider the difference Lax pair 
\begin{equation} \label{sl2-ddLP}
\Psi_{10}\,=\,M(p,q_{10},f;c_1,c_2) \,\Psi\,,\quad \Psi_{01}\,=\,M(p,q_{01},g;1,1)\,\Psi\,, 
\end{equation}
where matrix $M$ is given in (\ref{DT-sl2-gen}) and at least one of the constants $c_1$, $c_2$ is different from zero. 

From the consistency condition of this system, we derive the following system of equations.
\begin{subequations} \label{sl2-res-eq}
\begin{eqnarray}
 &&fg_{10}-gf_{01}  = 0,\label{consLaw} \\
 &&f_{01}q_{11} - fq_{10} - c_1 g_{10}q_{11} + c_2 gq_{01}=0, \\
 &&f_{01}p_{01} - fp - c_2 g_{10}p_{10} + c_1 gp = 0,\\
 &&f_{01}-f- c_1 (g_{10}-g) - f g_{10} p_{10}q_{10} + g f_{01}p_{01}q_{01} = 0.
\end{eqnarray}
\end{subequations}

As in the case of the nonlinear Schr\"odinger equation, we can solve equations (\ref{sl2-res-eq}) either for $p_{01}$, $q_{01}$, $f_{01}$ and $g$ or for $p_{10}$, $q_{10}$, $f$ and $g_{10}$, motivated by the evolution of the intial value problem on the staircase (cf. Figure \ref{fig-ivp}). 

Specifically, the first branch is the trivial solution given by 
\begin{equation}
p_{01}=\frac{c_2}{c_1}, \quad q_{01}=\frac{c_1}{c_2}q_{10}, \quad f_{01}=c_1g_{10}, \quad g=\frac{1}{c_1}f
\end{equation}
The second branch involves rational expressions of the remaining variables, and it is given by 
\begin{subequations} \label{SL2-sol}
\begin{align}
p_{01} &=\frac{A}{f\,B^2}\,\left(f^2 p^2 q_{10} + c_2 f p (g_{10} p_{10} q_{10}-1) - c_2^2 g_{10} p_{10} + c_1 c_2 g_{10} p \right), \\
g &= g_{10}\,\frac{A}{B}, \\ 
f_{01} &=f\, \frac{B}{A},\\ 
q_{01} &= \frac{B}{g_{10} A^2}\,\left(f (q_{11}-q_{10} + g_{10} p_{10} q_{10} q_{11}) + c_1 g_{10} q_{11} (g_{10} p_{10} q_{11}-1)\right),
\intertext{where $A$ and $B$ are given by the following expressions}
&A:=fp q_{11} + c_2 (g_{10} p_{10} q_{11}-1),\quad\text{and}\quad B:= f p q_{10} + c_1 g_{10} p q_{11} - c_2.
\end{align}
\end{subequations}

When either $c_1$ or $c_2$ is equal to 0, then system (\ref{sl2-res-eq}) admits a unique non-trivial solution and it is given by the above expressions if we set $c_1$ or $c_2$ equal to 0 accordingly.

Now, as in the case of NLS equation, the derivation of the Darboux matrix gives rise to some integrability properties for system (\ref{sl2-res-eq}). Specifically, we have the following.

\begin{proposition}\label{1stInt-z2}
System (\ref{sl2-res-eq}) admits two first integrals,\index{first integral(s)} $\mathcal{F}:=f^2pq_{10}-c_2f$ and $\mathcal{G}:=g^2pq_{01}-g$, and the following conservation law\index{conservation law(s)}
	\begin{equation}\label{conlawZ2}
	(\mathcal{T}-1)\ln f=(\mathcal{S}-1)\ln g.
	\end{equation}
\end{proposition}
\begin{proof}
Indeed, relation (\ref{sl2-D-con-det-gen}) suggests that
	\begin{equation}
	({\cal{T}}-1)\left(f^2pq_{10}-c_2f\right)=0\quad {\mbox{and}} \quad ({\cal{S}}-1)\left(g^2pq_{01}-g\right)=0.
	\end{equation}
This can be verified with straightforward calculation using equations (\ref{sl2-res-eq}). Thus, $F=f-pq_{10}$ and $G:=g-pq_{01}$ are first integrals.\index{first integral(s)} Moreover, from equation (\ref{consLaw}) we get $\ln(fg_{10})=\ln(gf_{01})$ which implies
\begin{equation}
\ln f_{01}-\ln f=\ln g_{10}-\ln g.
\end{equation}
The latter equation can be written in the form of the conservation law (\ref{conlaw}).\index{conservation law(s)}
\end{proof}

\subsubsection{Derivation of a six-point difference equation}\index{six-point difference equation}\index{difference equation(s)!six-point}
We now use Darboux matrices $M_{deg}(p)$ and $M(p,q_{01},g;1,1)$ in (\ref{DT-sl2-degen}) and (\ref{DT-sl2-gen}), respectively, to define the following discrete Lax pair
\begin{subequations}\label{DLaxP-Z2}
\begin{align}
\Psi_{10}&=\left(\lambda^{2}\left(\begin{array}{cc} 1/p & 0\\ 0 & 0\end{array}\right)+\lambda\left(\begin{array}{cc} 0 & 1\\ 1 & 0\end{array}\right)\right)\Psi,\label{DLaxP-Z2-1}\\
\Psi_{01}&= \left(\lambda^{2}\left(\begin{array}{cc} g & 0\\ 0 & 0\end{array}\right)+\lambda\left(\begin{array}{cc} 0 & gp\\ gq_{01} & 0\end{array}\right)+\left(\begin{array}{cc} 1 & 0\\ 0 & 1 \end{array}\right)\right)\Psi.\label{DLaxP-Z2-2}
\end{align}
\end{subequations}

The compatibility of the above system implies the following equations
\begin{equation}\label{pgg10}
g_{10}=\frac{p}{p_{01}}g, \quad g_{10}=\frac{p}{q_{11}}g, \quad \frac{1}{p_{01}}+gq_{01}=g_{10}p_{10}+\frac{1}{p}.
\end{equation}
From the first two of the above equations we conclude that $q_{11}=p_{01}$ and, thus, $q_{10}=p$ and $q_{01}=p_{-1,1}$. Additionally, we use the third equation of (\ref{pgg10}) to express field $g$ in terms of $p$ and its shifts. Then, the first equation of (\ref{pgg10}) can be rewritten as the following six-point difference equation \index{six-point difference equation}\index{difference equation!six-point}
\begin{equation}\label{sixpointEq}
\frac{p_{01}-p}{p_{01}(p_{01}p_{-11}-pp_{10})}=\frac{p_{11}-p_{10}}{p_{10}(p_{11}p_{01}-p_{10}p_{20})}
\end{equation}

Equation (\ref{pgg10}) can be solved uniquely for any of the $p$ and its shifts, apart from $p_{10}$ and $p_{01}$. This allows us to define uniquely the evolution of the initial data placed on a double staircase, as it is shown in Figure (\ref{fig-ivp-Z2-p})

\begin{figure}[ht]
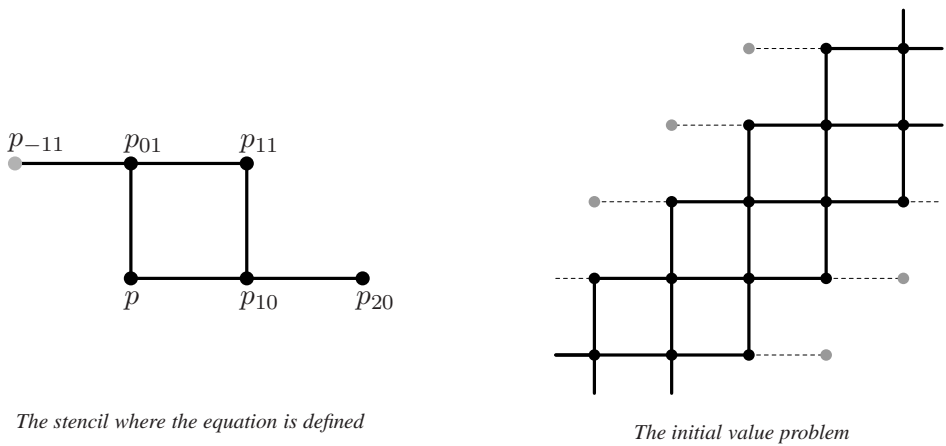

\centertexdraw{
\setunitscale 0.4
\setgray 0.0
\linewd 0.04 \move (-8 -1) \lvec (-6.5 -1) \lvec (-6.5 0.5) \lvec (-8 0.5) \lvec (-8 -1) \move (-8 0.5) \lvec (-9.5 0.5) \move (-6.5 -1) \lvec (-5 -1)
\move (-8 -1) \fcir f:0.0 r:0.095 \move (-6.5 -1) \fcir f:0.0 r:0.095
\move (-6.5 0.5) \fcir f:0.0 r:0.095 \move (-8 0.5) \fcir f:0.0 r:0.095
\move (-9.5 0.5) \fcir f:0.7 r:0.095 \move (-5 -1) \fcir f:0.0 r:0.095
\move (-2 -2.5) \lvec (-2 -2) \lvec (-1 -2) \lvec (-1 -1) \lvec (0 -1) \lvec (0 0) \lvec (1 0) \lvec(1 1) \lvec (2 1) \lvec(2 2) \lvec(2.5 2)
\linewd 0.04  \move (-2 -2) \lvec (-2 -1) \lvec(-1 -1) \lvec (-1 0) \lvec (0 0) \lvec (0 1) \lvec(1 1) \lvec (1 2) \lvec (2 2) \lvec (2 2.5) \move(-2.5 -2) \lvec(-2 -2) \move(-2.5 -2) \lvec(-2 -2)
\move (-1 -2.5) \lvec (-1 -2) \lvec(0 -2) \lvec(0 -1) \lvec(1 -1) \lvec(1 0) \lvec(2 0) \lvec(2 1) \lvec(2.5 1)
\linewd 0.001
\lpatt (.05 .05) \move (-2 0) \lvec (-1 0) \move (0 1) \lvec (-1 1) \move (0 2) \lvec (1 2) \move (-2.5 -1) \lvec (-2 -1) 
\move  (0 -2) \lvec (1 -2) \move (1 -1) \lvec (2 -1) \move (2 0) \lvec (2.5 0)
\lpatt  (1 1)
\move (-2 -2) \fcir f:0.0 r:0.075 \move (-2 -1) \fcir f:0.0 r:0.075
\move (-1 -2) \fcir f:0.0 r:0.075 \move (0 -2) \fcir f:0.0 r:0.075
\move (-1 -1) \fcir f:0.0 r:0.075 \move (-1 0) \fcir f:0.0 r:0.075 \move (0 -1) \fcir f:0.0 r:0.075 \move (1 -1) \fcir f:0.0 r:0.075
\move (0 0) \fcir f:0.0 r:0.075 \move (0 1) \fcir f:0.0 r:0.075 \move (1 0) \fcir f:0.0 r:0.075 \move (2 0) \fcir f:0.0 r:0.075 
\move (1 1) \fcir f:0.0 r:0.075 \move (2 1) \fcir f:0.0 r:0.075
\move (2 2) \fcir f:0.0 r:0.075 \move (1 2) \fcir f:0.0 r:0.075
\move (-2 0) \fcir f:0.6 r:0.075 
\move (-1 1) \fcir f:0.6 r:0.075 \move (0 2) \fcir f:0.6 r:0.075
\move (1 -2) \fcir f:0.6 r:0.075 \move (2 -1) \fcir f:0.6 r:0.075
\htext (-8.1 -1.4) {\small{$p$}}
\htext (-6.6 -1.4) {\small{$p_{10}$}}
\htext (-8.1 0.65) {\small{$p_{01}$}}
\htext (-6.6 0.65) {\small{$p_{11}$}}
\htext (-5.1 -1.4) {\small{$p_{20}$}}
\htext (-9.6 0.65) {\small{$p_{-11}$}}
\htext (-9.5 -3) {\scriptsize{\it{The stencil where the equation is defined}}}
\htext (-1.5 -3.15) {\scriptsize{\it{The initial value problem}}}
}
\caption{{\small{The stencil of six points and the initial value problem for equation (\ref{sixpointEq})}}} \label{fig-ivp-Z2-p}
\end{figure}

A first integral\index{first integral(s)} of equation (\ref{pgg10}) is given by the following.

\begin{corollary}
Equation (\ref{pgg10}) admits a first integral\index{first integral(s)} $\mathcal{G}_p$ given by
\begin{equation}
\mathcal{G}_p=\frac{(p_{01}-p)(p_{10}-p_{-11})}{(p_{01}p_{-11}-pp_{10})^2}.
\end{equation}
\end{corollary}
\begin{proof}
We express all the fields in $\mathcal{G}$, which is given in proposition \ref{1stInt-z2}, in terms of $p$ and its shifts.
\end{proof}  


\subsection{A deformation of the derivative nonlinear Schr\"odinger equa-\\tion and related discrete systems}
Here we employ Darboux matrix (\ref{Dih-M}) to introduce the following discrete Lax pair
\begin{equation}
\Psi_{10}=M(p,q_{10},f,g)\Psi, \quad \Psi_{01}=M(p,q_{01},u,v)\Psi.
\end{equation}
The compatibility condition of this Lax pair leads to an equation solely for $f$ and $u$, given by
\begin{equation}\label{conLaw-dDNLS}
f_{01}u-u_{10}f=0,
\end{equation}
and a system of equations for the remaining fields, given by
\begin{subequations}\label{DS-dDNLS}
\begin{align}
p(p_{01}-p)+(q_{01}-q_{10})q_{11}+g_{01}v-gv_{10}=&0,\\
g_{01}-g+v-v_{10}+p_{01}q_{01}-p_{10}q_{10}=&0,\\
p_{01}-p_{10}+g_{01}q_{01}+(v-g)q_{11}-v_{10}q_{10}=&0,\\
q_{01}-q_{10}+g_{01}p-gp_{10}+vp_{01}-v_{10}p=&0.
\end{align}
\end{subequations}

\begin{remark}\normalfont
We shall consider the case where the value of $\Phi_1$ in (\ref{dih-D-fi}) is nonzero. In the oposite case, we deduce that $g=pq_{10}$ and similarly $v=pq_{01}$; this is due to the fact that $f$ must be nonzero ($f=0$ implies $M=0$). Then, system (\ref{DS-dDNLS}) has only the trivial solution $p_{01}=p_{10}$ and $q_{01}=q_{10}$.
\end{remark}

\begin{proposition}
Equation (\ref{conLaw-dDNLS}) can be written in the form of the following conservation law\index{conservation law(s)}
\begin{equation}\label{conLaw-ds-D2}
(\mathcal{T}-1)\ln (g-pq_{10})=(\mathcal{S}-1)\ln (v-pq_{01}).
\end{equation}
\end{proposition}
\begin{proof}
The value of $\Phi_1$ can be rescaled to 1 and thus
\begin{equation}
f^2=\frac{1}{g-pq_{10}}, \quad u^2=\frac{1}{v-pq_{01}}.
\end{equation}
Then, we substitute the above back to equation (\ref{conLaw-dDNLS}), and the latter can be written in the form (\ref{conLaw-ds-D2}).
\end{proof}

Motivated by the initial value problem on the staircase, system (\ref{DS-dDNLS}) can be solved either for ($p_{01},q_{01},g_{01},v$) or ($p_{10},q_{10},g,v_{10}$). However, we present this the solution in the Appendix because of its length. Some properties of system (\ref{DS-dDNLS}) are given in the following.

\begin{proposition}
System (\ref{DS-dDNLS}) admits two first integrals\index{first integral(s)}
\begin{equation}\label{1stInt-DS-D2}
(\mathcal{T}-1)\frac{g-pq_{10}}{g^2+1-p^2-q_{10}^2}=0, \quad (\mathcal{S}-1)\frac{v-pq_{01}}{v^2+1-p^2-q_{01}^2}=0,
\end{equation}
and a conservation law given by\index{conservation law(s)}
\begin{equation}\label{conLaw-DS-D2}
(\mathcal{T}-1)(g+pq)=(\mathcal{S}-1)(v+pq).
\end{equation}
\end{proposition}
\begin{proof}
Eliminating $f$ from $\Phi_1$ and $\Phi_2$ we obtain
\begin{equation}
\partial_x\left(\frac{g-pq_{10}}{g^2+1-p^2-q_{10}^2} \right)=0.
\end{equation}
This suggests that relations (\ref{1stInt-DS-D2}) constitute first integrals \index{first integral(s)} for system (\ref{DS-dDNLS}) and can be readily shown. Moreover, it is straightforward to show that the first equation of (\ref{DS-dDNLS}) can be written in the form (\ref{conLaw-DS-D2}).
\end{proof}

In what follows, we use the first integrals (\ref{1stInt-DS-D2}) \index{first integral(s)}to reduce  system (\ref{DS-dDNLS}).


\subsubsection{Derivation of a discrete Toda type equation} \index{discrete!Toda equation}\index{Toda (type) equation(s)}
Here we consider the case where the first integrals\index{first integral(s)} (\ref{conLaw-DS-D2}) have the values
\begin{equation}
\frac{g-pq_{10}}{g^2+1-p^2-q_{10}^2}=0,\quad \text{and} \quad \frac{v-pq_{01}}{v^2+1-p^2-q_{01}^2}=\frac{1}{2},
\end{equation}
which implies the following algebraic equations
\begin{equation}\label{eqs-g-v}
g-pq_{10}=0, \quad (v-1+p-q_{01})(v-1-p+q_{01})=0.
\end{equation}
For the latter we choose the solution
\begin{equation}\label{sub-g-v}
g=pq_{10},\quad \text{and} \quad v=p-q_{01}+1.
\end{equation}

Substitution of the above expressions into system (\ref{DS-dDNLS}) we obtain a system of equations for $p$, $q$ and their shifts. Motivated by the initial value problem on the staircase as in (\ref{fig-ivp}), we solve this system either for $p_{01}$, $q_{01}$ or for $p_{10}$, $q_{10}$. In particular, we solve for $p_{01}$, $q_{01}$ and, in order to simplify the retrieved exressions, we make the point transformation 
\begin{equation}\label{pointtrans}
(p,q)\mapsto (p-1,q-1).
\end{equation}
Then, we come up with
\begin{equation}\label{sys-p01q01}
p_{01}=\frac{p_{10}q_{10}}{q_{11}}, \quad q_{01}=\frac{(p-2)(q_{10}-2)}{p_{10}q_{10}-2q_{11}}q_{11}+2.
\end{equation}

System (\ref{sys-p01q01}) admits the conservation law\index{conservation law(s)}
\begin{equation}
(\mathcal{T}-1)(p-1)(q_{10}+q-2)=(\mathcal{S}-1)((p-1)q-q_{01}),
\end{equation}
which results from (\ref{conLaw-DS-D2}) after substitution (\ref{sub-g-v}) followed by the point transformation (\ref{pointtrans}). Moreover, the first equation of (\ref{sys-p01q01}) can be written in the form of a conservation law,\index{conservation law(s)} namely
\begin{equation}\label{2ndconlaw}
(\mathcal{T}-1)\ln(pq_{10})=(\mathcal{S}-1)\ln(p).
\end{equation}

\begin{proposition}
System (\ref{sys-p01q01}) implies the following discrete Toda type equation
\begin{equation}\label{TodaEq}\index{discrete!Toda equation}\index{Toda (type) equation(s)}
e^{w_{0,-1}-w}-e^{w-w_{01}}+e^{w_{0,-1}-w_{1,-1}}-e^{w-w_{1,-1}}=\frac{1}{2}(e^{w_{-11}-w_{01}}-e^{w_{0,-1}-w_{1,-1}}).
\end{equation}
\end{proposition}
\begin{proof}
Conservation law (\ref{2ndconlaw})\index{conservation law(s)} suggests that we can introduce a new variable, $w$, via the relations
\begin{equation}
\ln(pq_{10})=w_{0,-1}-w_{1,-1}, \quad \ln(p)=w_{0,-1}-w,
\end{equation}
and therefore
\begin{equation}
p=\exp(w_{0,-1}-w), \quad q=\exp(w_{-10}-w_{0,-1}).
\end{equation}
It can be readily shown that, applying the above transformation to equations (\ref{sys-p01q01}), the first one is identically satisfied, whereas the second is written in the form (\ref{TodaEq}).
\end{proof}

Equation (\ref{TodaEq}) can be written in a conserved form as
\begin{eqnarray*}
&&(\mathcal{T}-1)(e^{w_{-10}-w}-e^{w-w_{1,-1}}+2e^{w_{0,-1}-w}+e^{w_{-10}-w_{0,-1}}+e^{w_{0,-1}-w_{1,-1}})=\\
&&(\mathcal{S}-1)(e^{w_{-10}-w}-e^{w_{-11}-w}+e^{w_{-10}}-e^{w_{0,-1}}).
\end{eqnarray*}

\begin{remark}\normalfont
If we choose the solution $g=pq_{10}, v=q_{01}-p+1$ of (\ref{eqs-g-v}) instead of (\ref{sub-g-v}), and make the point transformation $(p,q)\rightarrow (p+1,q+1)$, we will derive the system
\begin{equation}
p_{01}=\frac{p_{10}q_{10}}{q_{11}}, \quad q_{01}=\frac{(p+2)(q_{10}+2)}{p_{10}q_{10}+2q_{11}}q_{11}-2,
\end{equation}
which after a transformation $p\rightarrow \exp(w-w_{0,-1}), q \rightarrow -\exp(w_{0,-1}-w_{-10})$ is written in the form (\ref{TodaEq}).
\end{remark}

\subsubsection{Derivation of a seven point scalar equation}
Let us now choose the following values for the first integrals (\ref{conLaw-DS-D2})\index{first integral(s)}
\begin{equation}
\frac{g-pq_{10}}{g^2+1-p^2-q_{10}^2}=-\frac{1}{2},\quad \text{and} \quad \frac{v-pq_{01}}{v^2+1-p^2-q_{01}^2}=\frac{1}{2},
\end{equation}
or, equivalently,
\begin{equation}
(g+1+p+q_{01})(g+1-p-q_{01})=0, \quad (v-1+p-q_{01})(v-1-p+q_{01})=0.
\end{equation}
There are obviously four set of solutions for $g$ and $v$, and we choose
\begin{equation}\label{subsol-g-v}
g=p+q_{10}-1,\quad \text{and} \quad v=p-q_{01}+1.
\end{equation}

Substitution of the above expressions into system (\ref{DS-dDNLS}), and making the point transformation $(p,q)\mapsto (p-1,q-1)$, implies a system of difference equations among the potentials $p$ and $q$ and their shifts. Its solution for $p_{01}$ and $q_{01}$ is given by
\begin{equation}\label{deqsp01q01}
p_{01}=p_{10}-q_{11}+2+\frac{q_{10}-2}{p}p_{10}, \quad q_{01}=\frac{p_{10}q_{10}-2q_{11}}{p_{10}(p+q_{10}-2)-pq_{11}}p.
\end{equation}
Similarly, using the same substitution and point transformation, conservation laws (\ref{conLaw-DS-D2}) and (\ref{conLaw-ds-D2}) become\index{conservation law(s)}
\begin{subequations}
\begin{align}
&(\mathcal{T}-1)pq=(\mathcal{S}-1)(pq-2q_{01}),\label{conslawdD2-1}
\intertext{and}
(\mathcal{T}-1)&\ln(p-2)(q_{10}-2)=(\mathcal{S}-1)\ln p(q_{01}-2),\label{conslawdD2-2}
\end{align}
\end{subequations}
respectively, which of course constitute conservation laws for system (\ref{deqsp01q01}). \index{conservation law(s)}
\vspace{.4 cm}
\begin{figure}[h!]
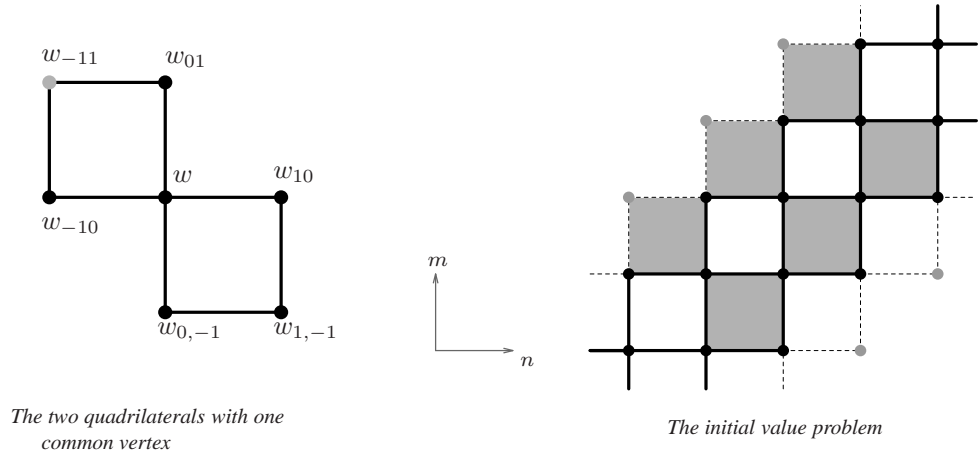

\centertexdraw{
\setunitscale 0.4
\move(-4.5 -2) \linewd 0.02 \setgray 0.4 \arrowheadtype t:V \arrowheadsize l:.12 w:.06 \avec(-4.5 -1) 
\move(-4.5 -2) \arrowheadtype t:V  \avec(-3.5 -2)
\setgray 1.0
\move (-2 -1) \lvec (-1 -1) \lvec (-1 0) \lvec (-2 0) \lvec (-2 -1) \lfill f:0.7
\move (-1 -2) \lvec (0 -2) \lvec (0 -1) \lvec (-1 -1) \lvec (-1 -2) \lfill f:0.7
\move (-1 0) \lvec (0 0) \lvec (0 1) \lvec (-1 1) \lvec (-1 0) \lfill f:0.7
\move (0 -1) \lvec (1 -1) \lvec (1 0) \lvec (0 0) \lvec (0 -1) \lfill f:0.7
\move (0 1) \lvec (1 1) \lvec (1 2) \lvec (0 2) \lvec (0 1) \lfill f:0.7
\move (1 0) \lvec (2 0) \lvec (2 1) \lvec (1 1) \lvec (1 0) \lfill f:0.7
\setgray 0.0
\linewd 0.04 \move (-8 -1.5) \lvec (-6.5 -1.5) \lvec (-6.5 0) \lvec (-9.5 0) \lvec (-9.5 1.5) \lvec(-8 1.5) \lvec (-8 -1.5)
\move (-8 -1.5) \fcir f:0.0 r:0.095 \move (-6.5 -1.5) \fcir f:0.0 r:0.095
\move (-6.5 0) \fcir f:0.0 r:0.095 \move (-8 0) \fcir f:0.0 r:0.095
\move (-9.5 0) \fcir f:0.0 r:0.095 \move (-9.5 1.5) \fcir f:0.7 r:0.095 \move (-8 1.5) \fcir f:0.0 r:0.095
\move (-2 -2.5) \lvec (-2 -2) \lvec (-1 -2) \lvec (-1 -1) \lvec (0 -1) \lvec (0 0) \lvec (1 0) \lvec(1 1) \lvec (2 1) \lvec(2 2) \lvec(2.5 2)
\linewd 0.04  \move (-2 -2) \lvec (-2 -1) \lvec(-1 -1) \lvec (-1 0) \lvec (0 0) \lvec (0 1) \lvec(1 1) \lvec (1 2) \lvec (2 2) \lvec (2 2.5) \move(-2.5 -2) \lvec(-2 -2) \move(-2.5 -2) \lvec(-2 -2)
\move (-1 -2.5) \lvec (-1 -2) \lvec(0 -2) \lvec(0 -1) \lvec(1 -1) \lvec(1 0) \lvec(2 0) \lvec(2 1) \lvec(2.5 1)
\linewd 0.001
\lpatt (.05 .05) \move (-2.5 -1) \lvec (-2 -1)  \lvec (-2 0) \lvec (-1 0) \lvec (-1 1) \lvec (0 1) \lvec (0 2) \lvec (1 2) \lvec (1 2.5)
\move (0 -2.5) \lvec (0 -2) \lvec (1 -2) \lvec (1 -1) \lvec (2 -1) \lvec (2 0) \lvec (2.5 0)
\lpatt  (1 1)
\move (-2 -2) \fcir f:0.0 r:0.075 \move (-2 -1) \fcir f:0.0 r:0.075
\move (-1 -2) \fcir f:0.0 r:0.075 \move (0 -2) \fcir f:0.0 r:0.075
\move (-1 -1) \fcir f:0.0 r:0.075 \move (-1 0) \fcir f:0.0 r:0.075 \move (0 -1) \fcir f:0.0 r:0.075 \move (1 -1) \fcir f:0.0 r:0.075
\move (0 0) \fcir f:0.0 r:0.075 \move (0 1) \fcir f:0.0 r:0.075 \move (1 0) \fcir f:0.0 r:0.075 \move (2 0) \fcir f:0.0 r:0.075 
\move (1 1) \fcir f:0.0 r:0.075 \move (2 1) \fcir f:0.0 r:0.075
\move (2 2) \fcir f:0.0 r:0.075 \move (1 2) \fcir f:0.0 r:0.075
\move (-2 0) \fcir f:0.6 r:0.075 
\move (-1 1) \fcir f:0.6 r:0.075 \move (0 2) \fcir f:0.6 r:0.075
\move (1 -2) \fcir f:0.6 r:0.075 \move (2 -1) \fcir f:0.6 r:0.075
\htext (-3.4 -2.2) {\scriptsize{$n$}}
\htext (-4.6 -.9) {\scriptsize{$m$}}
\htext (-8.1 -1.89) {\footnotesize{$w_{0,-1}$}}
\htext (-6.6 -1.89) {\footnotesize{$w_{1,-1}$}}
\htext (-7.9 .2) {\footnotesize{$w$}}
\htext (-6.6 .2) {\footnotesize{$w_{10}$}}
\htext (-9.6 -.5) {\footnotesize{$w_{-10}$}}
\htext (-9.6 1.7) {\footnotesize{$w_{-11}$}}
\htext (-8 1.7) {\footnotesize{$w_{01}$}}
\htext (-10 -3) {\scriptsize{\it{The two quadrilaterals with one}}}
\htext (-9.6 -3.3) {\scriptsize{\it{common vertex}}}
\htext (-1.5 -3.15) {\scriptsize{\it{The initial value problem}}}
}
\caption{{\small{The stencil of seven points and the initial value problem on the black and white lattice}}} \label{fig-ivp-y}
\end{figure}
\vspace{.4 cm}
\begin{proposition}
System (\ref{deqsp01q01}) implies the following seven-point scalar equation
\begin{eqnarray}\label{7pointeq}
&&(w-w_{10})(w_{0,-1}-w_{-10})\left(1+\frac{1}{w-w_{1,-1}}\right)\nonumber\\
&&+(w-w_{-10})(w_{10}-w_{01})\left(1-\frac{1}{w-w_{-11}}\right)=0.
\end{eqnarray}
\end{proposition}
\begin{proof}
Conservation law (\ref{conslawdD2-1}) \index{conservation law(s)} suggests that we can introduce a potential, $w$, using the following relations:
\begin{equation}
pq=4(w-w_{-10}),\quad pq-2q_{01}=4(w_{-11}-w_{-10}).
\end{equation}
Expressing $p$ and $q$ in terms of $w$, we obtain
\begin{equation}
p=2\frac{w-w_{-10}}{w_{0,-1}-w_{-10}},\quad q=2(w_{0,-1}-w_{-10}).
\end{equation}
Substitution of the above to the first equation of (\ref{deqsp01q01}), implies equation (\ref{7pointeq}).
\end{proof}

Equation (\ref{7pointeq}) involves seven lattice points, and it can be solved uniquely for any of the shifts of $w$, but not for $w$. These points can be placed on the vertices of two quadrilaterals with a common vertex; this allows us to consider an initial value problem where the initial values are placed placed across a double staircase, as in Figure \ref{fig-ivp-y}. Of course, since we can solve uniquely for both $w_{-11}$ and $w_{1,-1}$, the evolution can be uniquely determined in both directions.

\chapter{Introduction to Yang-Baxter maps}
\label{chap4} \setcounter{equation}{0}
\renewcommand{\theequation}{\thechapter.\arabic{equation}}

\section{Overview}
The original (quantum) Yang-Baxter equation \index{Yang-Baxter equation} \nomenclature{YB}{Yang-Baxter} originates in the works of Yang \cite{Yang} and Baxter \cite{Baxter}, and it has a fundamental role in the theory of quantum and classical integrable systems.

In this thesis we are interested in the study of the set-theoretical solutions of the Yang-Baxter equation.\index{Yang-Baxter (YB) equation} The first examples of such solutions  appeared in 1988, in a paper of Sklyanin \cite{Sklyanin}. However, the study of the set-theoretical solutions was formally proposed by Drinfiel'd in 1992 \cite{Drinfel'd}. Veselov, in \cite{Veselov}, proposed the more elegant term ``Yang-Baxter maps''\index{Yang-Baxter (YB) map(s)} for this type of solutions and, moreover, he connected them with integrable mappings \cite{Veselov, Veselov3}.

Yang-Baxter maps\index{Yang-Baxter (YB) map(s)} have been of great interest by many researchers in the area of Mathematical Physics. They are related to several concepts of integrability as, for instance, the multidimensionally consistent equations \cite{ABS-2004, ABS-2005, Bobenko-Suris,Frank3,Frank5,Frank4}. Especially, for those Yang-Baxter maps which admit Lax representation \cite{Veselov2}, there are corresponding hierarchies of commuting transfer maps\index{transfer map(s)} which preserve the spectrum of their monodromy matrix\index{monodromy matrix} \cite{Veselov,Veselov3}. 

In this chapter we give an introduction to the theory of Yang-Baxter maps.\index{Yang-Baxter (YB) map(s)}

In particular, this chapter is organised as follows: In the next section we briefly give an introduction to the Yang-Baxter equation\index{Yang-Baxter (YB) equation} and Yang-Baxter maps.\index{Yang-Baxter (YB) map(s)} We shall restrict our attention to the Yang-Baxter maps\index{Yang-Baxter (YB) map(s)} admitting \textit{Lax-representation} which we study in the next chapters of this thesis. In section 3 we discuss the connection between 3D consistent equations and Yang-Baxter maps,\index{Yang-Baxter (YB) map(s)} while section 4 deals with their classification. Finally, section 5 is devoted to the transfer dynamics \index{transfer dynamics}of Yang-Baxter maps\index{Yang-Baxter (YB) map(s)} and initial values problems on a two-dimensional lattice.

\section{The quantum Yang-Baxter equation}
Let $V$ be a vector space and $Y\in \End(V \otimes V)$ a linear operator. The Yang-Baxter equation\index{Yang-Baxter (YB) equation} is given by the following
\begin{equation}\label{YB_eq1}
Y^{12}\circ Y^{13} \circ Y^{23}=Y^{23}\circ Y^{13} \circ Y^{12},
\end{equation}
where $Y^{ij}$, $i,j=1,2,3$, $i\neq j$, denotes the action of $Y$ on the $ij$ factor of the triple tensor product $V \otimes V\otimes V$. In this form, equation (\ref{YB_eq1}) is known in the literature as the \textit{quantum YB equation}.\index{Yang-Baxter (YB) equation}

\subsection{Parametric Yang-Baxter maps}\index{Yang-Baxter (YB) map(s)!parametric}
Let us now replace the vector space $V$ by a set $A$, and the tensor product $V\otimes V$ by the Cartesian product $A \times A$. In what follows, we shall consider $A$ to be a finite dimensional algebraic variety in $K^N$, where $K$ is any field of zero characteristic, such as $\field{C}$ or $\field{Q}$.

Now, let $Y\in \End(A\times A)$ be a map defined by
\begin{equation}\label{Y-map}
Y:(x,y)\mapsto (u(x,y),v(x,y)).
\end{equation}
Furthermore, we define the maps $Y^{ij}\in \End(A\times A \times A)$ for $i,j=1,2,3,~i\neq j$, which appear in equation (\ref{YB_eq1}), by the following relations
\begin{subequations}\label{Yijs}
\begin{align}
 Y^{12}(x,y,z)&=(u(x,y),v(x,y),z), \\ 
 Y^{13}(x,y,z)&=(u(x,z),y,v(x,z)), \\ 
 Y^{23}(x,y,z)&=(x,u(y,z),v(y,z)).
\end{align}
\end{subequations}
Let also $Y^{21}=\pi Y \pi$, where $\pi\in\End(A\times A)$ is the permutation map: $\pi(x,y)=(y,x)$.

Map $Y$ is a YB map,\index{Yang-Baxter (YB) map(s)} if it satisfies the YB equation ($\ref{YB_eq1}$). \index{Yang-Baxter (YB) equation} Moreover, it is called \textit{reversible} if the composition of $Y^{21}$ and $Y$ is the identity map, i.e.
\begin{equation}\label{reversible}
Y^{21}\circ Y=Id.
\end{equation}

Now, let us consider the case where parameters are involved in the definition of the YB map.\index{Yang-Baxter (YB) map(s)} In particular we define the following map
\begin{equation}
Y_{a,b}:(x,y)\mapsto (u,v)\equiv (u(x,y;a,b),v(x,y;a,b)).
\end{equation}
This map is called \textit{parametric YB map}\index{Yang-Baxter (YB) map(s)!parametric} if it satisfies the \textit{parametric YB equation}\index{Yang-Baxter (YB) equation!parametric}
\begin{equation}\label{YB_eq}
Y^{12}_{a,b}\circ Y^{13}_{a,c} \circ Y^{23}_{b,c}=Y^{23}_{b,c}\circ Y^{13}_{a,c} \circ Y^{12}_{a,b}.
\end{equation}

One way to represent the map $Y_{a,b}$ is to consider the values $x$ and $y$ taken on the sides of the quadrilateral as in figure \ref{YBmap-Eq}-(a); the map $Y_{a,b}$ maps the values $x$ and $y$ to the values placed on the opposite sides of the quadrilateral, $u$ and $v$. 

Moreover, for the YB equation,\index{Yang-Baxter (YB) equation} we consider the values $x$, $y$ and $z$ taken on the sides of the cube as in figure \ref{YBmap-Eq}-(b). Specifically, by the definition \ref{Yijs} of the functions $Y^{ij}$, the map $Y^{23}_{b,c}$ maps 
\begin{equation}
(x,y,z)\stackrel{Y^{23}_{b,c}}{\rightarrow}(x,y^{(1)},z^{(1)}),
\end{equation}
using the right face of the cube. Then, map $Y^{13}_{a,c}$ maps
\begin{equation}
(x,y^{(1)},z^{(1)})\stackrel{Y^{13}_{a,c}}{\rightarrow}(x^{(1)},y^{(1)},z^{(2)})\equiv Y^{13}_{a,c} \circ Y^{23}_{b,c}(x,y,z),
\end{equation}
using the front face of the cube. Finally, map $Y^{12}_{a,b}$ maps
\begin{equation}
(x^{(1)},y^{(1)},z^{(2)})\stackrel{Y^{12}_{a,b}}{\rightarrow}(x^{(2)},y^{(2)},z^{(2)})\equiv Y^{12}_{a,b}\circ Y^{13}_{a,c} \circ Y^{23}_{b,c}(x,y,z),
\end{equation}
using the top face of the cube.

On the other hand, using the bottom, the back and the left face of the cube, the values $x$, $y$ and $z$ are mapped to the values $\hat{x}^{(2)}$, $\hat{y}^{(2)}$ and $\hat{z}^{(2)}$ via the map $Y^{23}_{b,c}\circ Y^{13}_{a,c} \circ Y^{12}_{a,b}$ which consists with the right hand side of equation, namely (\ref{Y-map})
\begin{equation}
Y^{23}_{b,c}\circ Y^{13}_{a,c} \circ Y^{12}_{a,b}(x,y,z)=(\hat{x}^{(2)},\hat{y}^{(2)},\hat{z}^{(2)}).
\end{equation}
Therefore, the map $Y_{a,b}$ satisfies the YB equation (\ref{YB_eq})\index{Yang-Baxter (YB) equation} if and only if $x^{(2)}=\hat{x}^{(2)}$, $y^{(2)}=\hat{y}^{(2)}$ and $z^{(2)}=\hat{z}^{(2)}$.

\begin{figure}[ht]
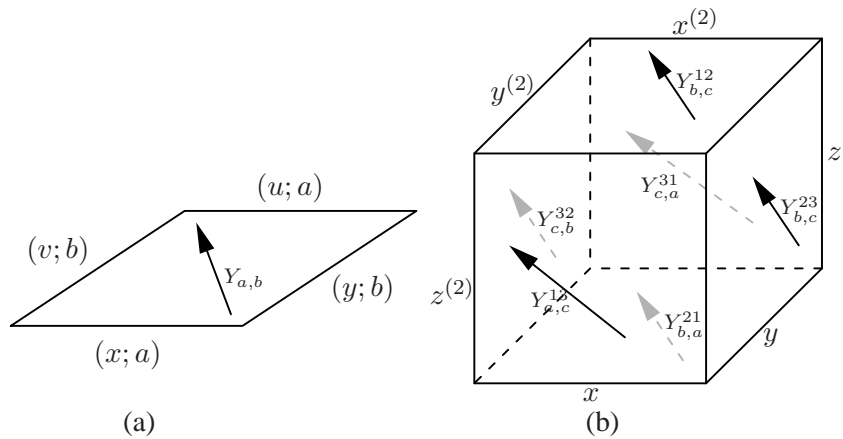

\centering
\centertexdraw{ 
\setunitscale 0.6
\move (-2 0.5)  \lvec(0 0.5) \lvec(1.5 1.5) \lvec(-0.5 1.5) \lvec(-2 0.5)
\textref h:C v:C \small{\htext(-1 0.25){$(x;a)$}}
\textref h:C v:C \htext(0.4 1.7){$(u;a)$} 
\textref h:C v:C \htext(1.03 0.82){$(y;b)$}
\textref h:C v:C \htext(-1.6 1.15){$(v;b)$}  
\textref h:C v:C \scriptsize{\htext(0 0.9){$Y_{a,b}$}}
\move (-0.1 0.6) \arrowheadtype t:F \avec(-0.4 1.4)
\move (2.7 1.1) \arrowheadtype t:F \lpatt(0.067 0.1)  
\setgray 0.7
\avec(2.3 1.7)
\textref h:C v:C \scriptsize{\htext(2.7 1.4){$Y^{32}_{c,b}$}}
\move (4.4 1.4) \arrowheadtype t:F \footnotesize{\avec(3.3 2.2)}
\textref h:C v:C \scriptsize{\htext(3.6 1.7){$Y^{31}_{c,a}$}}
\move (3.8 0.2) \arrowheadtype t:F \avec(3.4 0.8)
\textref h:C v:C \scriptsize{\htext(3.8 0.5){$Y^{21}_{b,a}$}}
\lpatt()
\setgray 0
\move (2 0)  \lvec(4 0) \lvec(5 1) 
\lpatt(0.067 0.1) \lvec(3 1) \lvec(2 0) 
\lpatt() \lvec(2 2) \lvec(3 3) 
\lpatt (0.067 0.1) \lvec(3 1) 
\lpatt() \move (3 3) \lvec(5 3) \lvec(4 2) \lvec(4 0) 
\move (2 2) \lvec(4 2)
\move (5 3) \lvec(5 1)
\textref h:C v:C \small{\htext(3 -0.1){$x$}}
\textref h:C v:C \small{\htext(4.55 0.4){$y$}}
\textref h:C v:C \small{\htext(5.1 2){$z$}}
\move (4.8 1.2) \arrowheadtype t:F \avec(4.4 1.8)
\textref h:C v:C \scriptsize{\htext(4.8 1.5){$Y^{23}_{b,c}$}}
\move (3.3 0.4) \arrowheadtype t:F \avec(2.3 1.2)
\textref h:C v:C \scriptsize{\htext(2.65 0.7){$Y^{13}_{a,c}$}}
\textref h:C v:C \small{\htext(1.8 0.8){$z^{(2)}$}}
\move (3.9 2.3) \arrowheadtype t:F \avec(3.5 2.9)
\textref h:C v:C \scriptsize{\htext(3.9 2.6){$Y^{12}_{b,c}$}}
\textref h:C v:C \small{\htext(3.9 3.15){$x^{(2)}$}} \htext(2.33 2.55){$y^{(2)}$}
\textref h:C v:C \htext(3.1 -0.35){(b)}
\textref h:C v:C \htext(-0.9 -0.35){(a)}
}
\caption{Cubic representation of (a) the parametric YB map\index{Yang-Baxter (YB) map(s)!parametric} and (b) the corresponding YB equation.}\label{YBmap-Eq}\index{Yang-Baxter (YB) equation}
\end{figure}

Most of the examples of YB maps \index{Yang-Baxter (YB) map(s)} which appear in this thesis are parametric. 

\begin{example}\normalfont
One of the most famous parametric YB maps \index{Yang-Baxter (YB) map(s)!parametric} is Adler's map \cite{Adler}
\begin{equation}\label{Adler_map}
(x,y)\stackrel{Y_{a,b}}{\rightarrow}(u,v)=\left(y-\frac{a-b}{x+y},x+\frac{a-b}{x+y}\right),
\end{equation}
which is related to the 3-D consistent discrete potential KDV equation \cite{Frank, PNC}.
\end{example}

\subsection{Matrix refactorisation problems and the Lax equation}
Let us consider the matrix $L$ depending on a variable $x$, a parameter $c$ and a \textit{spectral parameter} $\lambda$, namely $L=L(x;c,\lambda)$, such that the following matrix refactorisation problem\index{refactorisation problem(s)}
\begin{equation} \label{eqLax}
L(u;a,\lambda)L(v;b,\lambda)=L(y;b,\lambda)L(x;a,\lambda), \quad \text{for any $\lambda \in \field{C}$,}
\end{equation}
is satisfied whenever $(u,v)=Y_{a,b}(x,y)$. Then, $L$ is called Lax matrix\index{Lax matrix(-ces)} for $Y_{a,b}$, and (\ref{eqLax}) is called the \textit{Lax-equation}\index{Lax equation} or \textit{Lax-representation} for $Y_{a,b}$.

\begin{note}\normalfont
In the rest of this thesis we use the letter ``$L$" when referring to Lax matrices\index{Lax matrix(-ces)} of the refactorisation problem (\ref{eqLax})\index{refactorisation problem(s)} and the calligraphic ``$\mathcal{L}$" for Lax operators. Moreover, for simplicity of the notation, we usually omit the dependence on the spectral parameter, namely $L(x;a,\lambda)\equiv L(x;a)$. 
\end{note}

Since the Lax equation (\ref{eqLax})\index{Lax equation} does not always have a unique solution for $(u,v)$, Kouloukas and Papageorgiou in \cite{Kouloukas2} proposed the term \textit{strong Lax matrix}\index{Lax matrix(-ces)!strong} for a YB map.\index{Yang-Baxter (YB) map(s)} This is when the Lax equation\index{Lax equation} is equivalent to a map
\begin{equation}\label{unique-sol}
  (u,v)=Y_{a,b}(x,y).
\end{equation}
The uniqueness of refactorisation (\ref{eqLax}) is a sufficient condition for the solutions of the Lax equation\index{Lax equation} to define a reversible YB map \index{Yang-Baxter (YB) map(s)!reversible} \cite{Veselov3} of the form (\ref{unique-sol}). In particular, we have the following.

\begin{proposition}\label{PropVes}
(Veselov) Let $u=u(x,y)$, $v=v(x,y)$ and $L=L(x;\alpha)$ a matrix such that the refactorisation (\ref{eqLax}) is unique. Then, the map defined by (\ref{unique-sol}) satisfies the Yang-Baxter equation\index{Yang-Baxter (YB) equation} and it is reversible.
\end{proposition}
\begin{proof}
Due to the associativity of matrix multiplication and equation (\ref{eqLax}), we have
\begin{eqnarray}\label{uniquePr1}
&L(z;c)L(y;b)L(x;a)=L(y^{(1)};b)L(z^{(1)};c)L(x;a)=&\nonumber\\
&L(y^{(1)};b)L(x^{(1)};a)L(z^{(2)};c)=L(x^{(2)};a)L(y^{(2)};b)L(z^{(2)};c).&
\end{eqnarray}
On the other hand
\begin{eqnarray}\label{uniquePr2}
&L(z;c)L(y;b)L(x;a)=L(z;c)L(\hat{x}^{(1)};a)L(\hat{y}^{(1)};b)=&\nonumber\\
&L(\hat{x}^{(2)};a)L(\hat{z}^{(1)};c)L(\hat{y}^{(1)};b)=L(\hat{x}^{(2)};a)L(\hat{y}^{(2)};b)L(\hat{z}^{(2)};c).&
\end{eqnarray}
From the relations (\ref{uniquePr1}) and (\ref{uniquePr2}) follows that
\begin{equation}\label{trifactorisation}
L(x^{(2)};a)L(y^{(2)};b)L(z^{(2)};c)=L(\hat{x}^{(2)};a)L(\hat{y}^{(2)};b)L(\hat{z}^{(2)};c).
\end{equation}
Since the refactorisation (\ref{eqLax}) is unique, the above equation implies
\begin{equation}\label{YBprop}
x^{(2)}=\hat{x}^{(2)}, \qquad y^{(2)}=\hat{y}^{(2)}, \qquad \text{and} \quad z^{(2)}=\hat{z}^{(2)},
\end{equation}
which is the Yang-Baxter equation\index{Yang-Baxter (YB) equation}. 

For the reversibility, we need to show that
\begin{equation}
\pi Y \pi Y (x,y)=(x,y),
\end{equation}
or equivalently that
\begin{equation}\label{revrels}
\begin{cases}
x=v(v(x,y),u(x,y)),\\
y=u(v(x,y),u(x,y)).
\end{cases}
\end{equation}
Now, we have that
\begin{equation}\label{prrev}
L(u(x,y);a)L(v(x,y),b)=L(y;b)L(x;a), \quad \text{for any} ~~ x,y\in A.
\end{equation}
For $x=v(x,y)$ and $y=u(x,y)$ we have 
\begin{equation}\label{prrev-2}
L(u(x,y);a)L(v(x,y);b)=L(u(v(x,y),u(x,y));b)L(v(v(x,y),u(x,y)),a),
\end{equation}
where we have swapped $a$ with $b$. Since, the refactorisation is unique, (\ref{prrev}) and (\ref{prrev-2}) imply (\ref{revrels}).
\end{proof}

In the case where the map (\ref{unique-sol}) admits Lax representation (\ref{eqLax}), but it is not equivalent to (\ref{eqLax}), one may need to check the YB property separately. Yet, we can use the following \textit{trifactorisation criterion}.

\begin{corollary} (Kouloukas-Papageorgiou)
Let $u=u(x,y)$ and $v=v(x,y)$ and $L=L(x;\alpha)$ a matrix such that $L(u;a)L(v;b)=L(y;b)L(x;a)$. If equation (\ref{trifactorisation}) implies the relation (\ref{YBprop}), then the map defined by (\ref{unique-sol}) is a Yang-Baxter map.\index{Yang-Baxter (YB) map(s)}
\end{corollary}

In this thesis we are interested in those YB maps\index{Yang-Baxter (YB) map(s)} whose Lax representation involves matrices with rational dependence on the spectral parameter, as the following.

\begin{example}\normalfont
In terms of Lax matrices, Adler's map (\ref{Adler_map}) has the following strong Lax representation\index{Lax representation!strong} \cite{Veselov2, Veselov3}
\begin{equation}
L(u;a,\lambda)L(v;b,\lambda)=L(y;b,\lambda)L(x;a,\lambda), \quad \text{for any $\lambda \in \field{C}$,}
\end{equation}
where
\begin{equation}
L(x;a,\lambda)=
\left(\begin{matrix}
x & 1 \\
x^2-a & x
\end{matrix}\right)-\lambda \left(\begin{matrix}
0 & 0 \\
1 & 0
\end{matrix}\right).
\end{equation}
\end{example}




\section{Yang-Baxter maps and 3D consistent equations}\index{Yang-Baxter (YB) map(s)}
From the representation of the YB equation\index{Yang-Baxter (YB) equation} on the cube, as in Fig. \ref{YBmap-Eq}-(b), it is clear that the YB equation is essentially the same with the 3D consistency\index{3D consistency} condition with the fields lying on the edges of the cube. Therefore, one would expect that we can derive YB maps\index{Yang-Baxter (YB) map(s)} from equations having the 3D consistency\index{3D consistency} property.

The connection between YB maps\index{Yang-Baxter (YB) map(s)} and the multidimensional consistency\index{multidimensional consistency} condition for equations on quad graphs originates in the paper of Adler, Bobenko and Suris in 2003 \cite{ABS-2004}. However, a more systematic approach was presented in the paper of Papageorgiou, Tongas and Veselov \cite{PTV} a couple of years later and it is based on the symmetry analysis of equations on quad-graphs. In particular, the YB variables constitute invariants of their symmetry groups.

We present the example of the discrete potential KdV (dpKdV) equation \cite{PNC, Frank} which was considered in \cite{PTV}.

\begin{example}\normalfont
The dpKdV equation is given by
\begin{equation}\label{dpKdV}
(f_{11}-f)(f_{10}-f_{01})-a+b=0,
\end{equation}
where the fields are placed on the vertices of the square as in figure (\ref{dpKdVtoAdler}). We consider the values on the edges to be the difference of the values on the vertices, namely
\begin{equation}\label{invar}
x=f_{10}-f, \quad y=f_{11}-f_{10}, \quad u=f_{11}-f_{01} \quad \text{and} \quad v=f_{01}-f,
\end{equation}
as in figure (\ref{dpKdVtoAdler}). This choice of the variables is motivated by the fact that the dpKdV equation is invariant under the translation $f\rightarrow f+const.$ Now, the invariants (\ref{invar}) satisfy the following equation
\begin{equation}\label{eq1}
x+y=u+v.
\end{equation}
Moreover, the equation (\ref{dpKdV}) can be rewritten as
\begin{equation}\label{eq2}
(x+y)(x-v)=a-b.
\end{equation}

Solving (\ref{eq1}) and (\ref{eq2}), we obtain the Adler's map (\ref{Adler_map}).
\end{example}

\begin{figure}[ht]
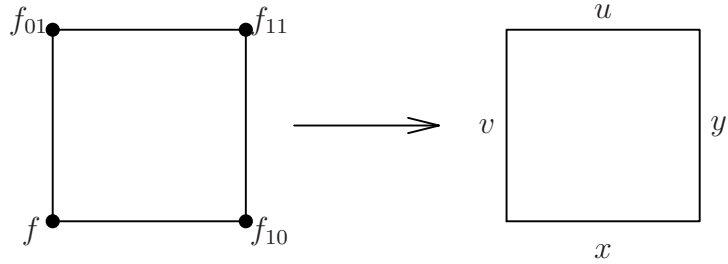

\centering
\centertexdraw{\setunitscale 0.5 \move (-1 0) \fcir f:0 r:0.075 \lvec(1 0)\fcir f:0 r:0.075 \lvec(1 2) \fcir f:0 r:0.075 \lvec (-1 2) \fcir f:0 r:0.075 \lvec (-1 0)
\textref h:C v:C \htext(-1.22 -0.1){$f$} \htext(-1.25 2.1){$f_{01}$} \htext(1.25 -0.1){$f_{10}$} \htext(1.25 2.1){$f_{11}$}
\move (1.5 1) \arrowheadtype t:V \avec(3 1)
\move (3.7 0) \lvec(5.7 0) \lvec(5.7 2) \lvec (3.7 2) \lvec (3.7 0)
\textref h:C v:C \htext(4.7 -0.3){$x$} \htext(5.9 1){$y$} \htext(4.7 2.2){$u$} \htext(3.5 1){$v$}
}
\caption{(a) dpKdV equation: fields placed on vertices (b) Adler's map: fields placed on the edges.}\label{dpKdVtoAdler}
\end{figure}

\begin{example}\normalfont
For the dpKdV equation, let us now consider a different combination for the variables assigned on the edges of the square, namely
\begin{equation}
x=ff_{10},\quad y=f_{10}f_{11}, \quad u=f_{01}f_{11}, \quad \text{and} \quad v=ff_{01}.
\end{equation}
This choice is motivated by the fact that equation (\ref{dpKdV}) is invariant under the change $f\rightarrow \epsilon f$, $f_{11}\rightarrow \epsilon f_{11}$, $f_{10}\rightarrow \epsilon^{-1} f_{10}$ and $f_{01}\rightarrow \epsilon^{-1} f_{01}$.

Now, the above variables satisfy the following equation
\begin{equation}\label{ex3}
xu=yv.
\end{equation}
On the other hand, we have that
\begin{equation}
y-x=f_{10}(f_{11}-f), \qquad u-v=f_{01}(f_{11}-f),
\end{equation}
and, therefore, the dpKdV equation can be rewritten as 
\begin{equation}\label{ex4}
y-x-(u-v)-a+b=0.
\end{equation}

Solving (\ref{ex3}) and (\ref{ex4}) for $u$ and $v$ we obtain the following map
\begin{equation}\label{YBexample}
(x,y)\longrightarrow (u,v)\equiv\left(y \left(1+\frac{a-b}{x-y}\right),\,x \left(1+\frac{a-b}{x-y}\right)\right),
\end{equation}
which is a parametric YB map.\index{Yang-Baxter (YB) map(s)!parametric}
\end{example}

\section{Classification of quadrirational YB maps: The $H$-list}\index{$H$-list}\index{Yang-Baxter (YB) map(s)!quadrirational}
All the quadrirational maps\index{quadrirational maps} in the $F$-list\index{$F$-list} presented in the first chapter satisfy the YB equation.\index{Yang-Baxter (YB) equation} However, in principle, their M\"obius-equivalent maps\index{M\"obius-equivalent maps} do not necessarily have the YB property, as in the following.

\begin{example}\normalfont
Consider the map $F_V$ of the $F$-list.\index{$F$-list} Under the change of variables
\begin{equation}\label{Mchange}
(x,y,u,v)\rightarrow(-x,-y,u,v),
\end{equation}
it becomes
\begin{equation}
(x,y)\rightarrow(-y-\frac{a-b}{x-y},-x-\frac{a-b}{x-y}).
\end{equation}
The above map does not satisfy the YB equation.\index{Yang-Baxter (YB) equation}
\end{example} 

In fact, all the maps of the $F$-list\index{$F$-list} lose the YB property under the transformation (\ref{Mchange}).

The quadrirational maps\index{quadrirational maps} which satisfy the YB equation \index{Yang-Baxter (YB) equation} were classified in \cite{PSTV}. Particularly, their classification is based on the following.

\begin{definition}
Let $\rho_\lambda:X\rightarrow X$ be a $\lambda$-parametric family of bijections. The parametric YB maps\index{Yang-Baxter (YB) map(s)!parametric} $Y_{a,b}$ and $\tilde{Y}_{a,b}$ are called equivalent, if they are related as follows
\begin{equation}\label{equivYB}
\tilde{Y}_{a,b}=\rho_{a}^{-1}\times \rho_{b}^{-1}~Y_{a,b}~\rho_a \times\rho_b.
\end{equation}
\end{definition}

\begin{remark}\normalfont
It is straightforward to show that the above equivalence relation is well defined; if $Y_{a,b}$ has the YB property, so does the map $\tilde{Y}_{a,b}$.
\end{remark}

The representative elements of the equivalence classes, with respect to the equivalence relation (\ref{equivYB}), are given by the following list.

\begin{theorem} Every quadrirational parametric YB map\index{Yang-Baxter (YB) map(s)!parametric}\index{Yang-Baxter (YB) map(s)!quadrirational} is equivalent (in the sense (\ref{equivYB})) to one of the maps of the F-list\index{$F$-list} or one of the maps of the following list
\begin{align}
 u&=yQ^{-1}, ~~\quad v=xQ,~\quad\qquad Q=\frac{(1-b)xy+(b-a)y+b(a-1)}{(1-a)xy+(a-b)x+a(b-1)}; \tag{$H_I$}\label{HI}\\
 u&=yQ^{-1}, ~~\quad v=xQ,~\quad\qquad Q=\frac{a+(b-a)y-bxy}{b+(a-b)x-axy}; \tag{$H_{II}$}\label{HII}\\
 u&=\frac{y}{a}Q,~\quad \quad v=\frac{x}{b}Q, \quad\qquad Q=\frac{ax+by}{x+y};\tag{$H_{III}$}\label{HIII}\\
 u&=yQ^{-1} \quad\quad v=xQ,~\quad\qquad Q=\frac{axy+1}{bxy+1}\tag{$H_{IV}$};\label{HIV}\\
 u&=y-P, ~~\quad v=x+P,\qquad P=\frac{a-b}{x+y}\tag{$H_{V}$}.\label{HV}
\end{align}
\end{theorem}

We refer to the above list as the $H$-list.\index{$H$-list} Note that, the map $H_V$ is the Adler's map (\ref{Adler_map}).

\section{Transfer dynamics of YB maps and initial value pro-\\blems}\index{transfer dynamics}
It is well known that, given a Yang-Baxter map,\index{Yang-Baxter (YB) map(s)} there is a hierarchy of commuting \textit{transfer maps},\index{transfer map(s)} which arise out of the consideration of initial value problems. The connection between the set-theoretical solutions of the YB equation \index{Yang-Baxter (YB) equation} and integrable mappings was first introduced by Veselov in \cite{Veselov, Veselov3}. In particular he showed that for those YB maps\index{Yang-Baxter (YB) map(s)} which admit Lax representation, there is a hierarchy of commuting transfer maps\index{transfer map(s)} which preserve the spectrum of their monodromy matrix.\index{monodromy matrix}

In this section we present the transfer maps,\index{transfer map(s)} which arise out of the consideration of the initial value problem on the staircase, as they were defined in \cite{kouloukasEnt}. Specifically, in \cite{kouloukasEnt} they considered YB maps,\index{Yang-Baxter (YB) map(s)} which admit a \textit{Lax pair} representation, $(L,M)$, namely YB maps\index{Yang-Baxter (YB) map(s)} which can be represented as 
\begin{equation}
L(u;a,\lambda)M(v;b,\lambda)=M(y;b,\lambda)L(x;a,\lambda).
\end{equation}
These maps are the so-called \textit{entwining} YB maps.\index{Yang-Baxter (YB) map(s)!entwining}

However, in this section we shall restrict ourselves to the case when $L\equiv M$. Therefore, we present the transfer maps\index{transfer map(s)} defined in \cite{kouloukasEnt} in the particular case when they admit Lax representation (\ref{eqLax}).

For a given a parametric YB map $Y_{a,b}$,\index{Yang-Baxter (YB) map(s)!parametric} we can consider a periodic initial value problem on the staircase as in \cite{PNC, kouloukasEnt}. Motivated by the fact that the YB map\index{Yang-Baxter (YB) map(s)} can be represented as a map mapping two successive edges of the quadrilateral to the opposite ones (as in figure \ref{YBmap-Eq}-(a)), we shall place the initial values on the edges of the staircase. In particular, let $x_1,x_2,\ldots,x_n$ and $y_1,y_2,\ldots,y_n$ be initial values assigned to the edges of the staircase with periodic boundary conditions 
\begin{equation}
x_{n+1}=x_1 \quad \text{and} \quad y_{n+1}=y_1,
\end{equation}
as in Figure \ref{staircase}. The edges with values $x_i$ and $y_i$, $i=1,\ldots,n$, carry the parameters $a$ and $b$ respectively.

The YB map $Y_{a,b}$\index{Yang-Baxter (YB) map(s)} maps the values $(x_i,y_i)$ to $(x_i^{(1)},y_i^{(1)})=Y_{a,b}(x_i,y_i)$. Then, the values on the several levels of the lattice will be given by
\begin{equation}
(x_i^{(k)},y_j^{(k)})=Y_{a,b}(x_i^{(k-1)},y_j^{(k-1)}), \quad j=i+k-1 \mod n.
\end{equation}

Now, for the $n$-periodic problem in Figure \ref{staircase} we define the \textit{transfer map}\index{transfer map(s)}
\begin{equation}
T_n:(x_1,x_2,\ldots,x_n,y_1,y_2,\ldots,y_n)\mapsto (x_1^{(1)},x_2^{(1)},\ldots,x_n^{(1)},y_2^{(1)},\ldots,y_n^{(1)},y_1^{(1)}),
\end{equation}
which maps the initial values $x_1,\ldots,x_n$ and $y_1,\ldots,y_n$ to the next level of the staircase. Note that $T_1^1\equiv Y_{a,b}$. Moreover, we define the $k$-\textit{transfer map}\index{transfer map(s)}
\begin{equation}
T_n^k:(x_1,\ldots,x_n,y_1,\ldots,y_n)\mapsto (x_1^{(k)},\ldots,x_n^{(k)},y_{r+1}^{(k)},\ldots,y_n^{(k)},y_1^{(k)},\ldots,y_r^{(k)}), \quad T_n^1\equiv T_n,
\end{equation}
where $r=k\mod n$. 
\begin{figure}[ht]
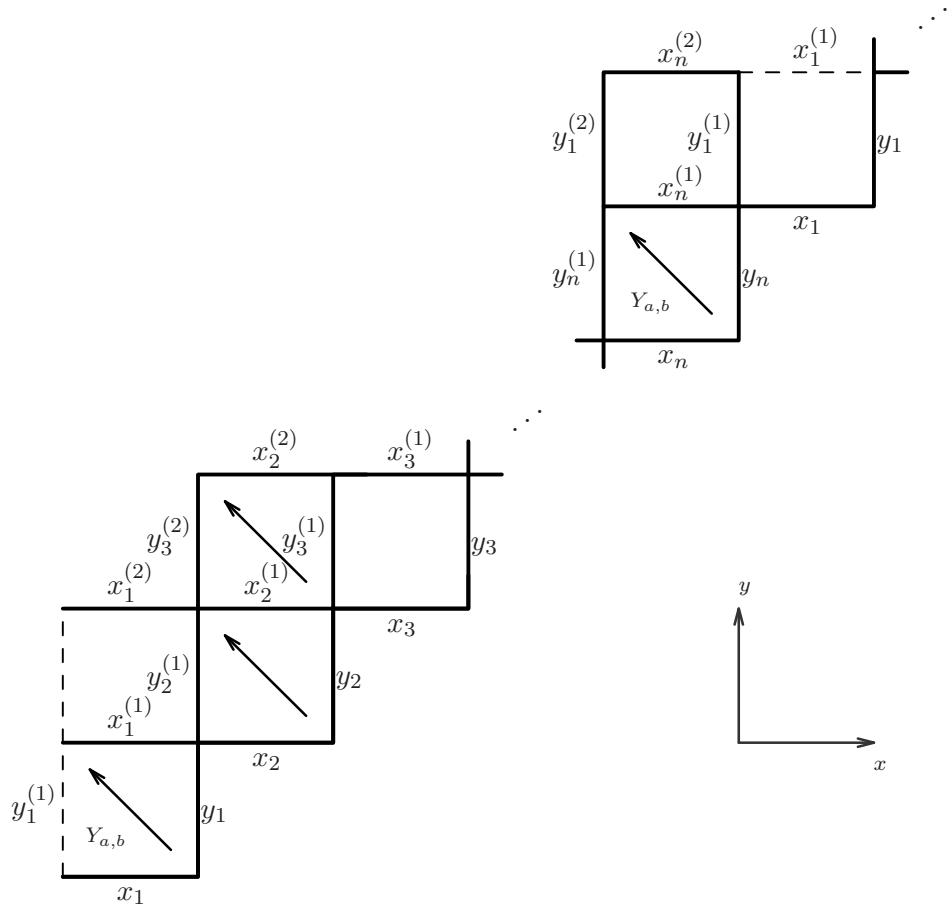

\centertexdraw{
\setunitscale 0.7
\move(3 -1)  \linewd 0.02 \setgray 0.2 \arrowheadtype t:V \arrowheadsize l:.12 w:.06 \avec(3 0) 
\move(3 -1) \arrowheadtype t:V  \avec(4 -1)
\setgray 0.0 
\move(-1.2 -1.8)  \arrowheadtype t:V \arrowheadsize l:.12 w:.06 \avec(-1.8 -1.2) 
\move(-0.2 -0.8)  \arrowheadtype t:V \arrowheadsize l:.12 w:.06 \avec(-0.8 -0.2) 
\move(2.8 2.2)  \arrowheadtype t:V \arrowheadsize l:.12 w:.06 \avec(2.2 2.8) 
\move(-0.2 0.2)  \arrowheadtype t:V \arrowheadsize l:.12 w:.06 \avec(-0.8 0.8) 
\setgray 0.0
\linewd 0.03 \move (-2 -2) \lvec (-1 -2) \lvec (-1 -1) \lvec (0 -1) \lvec (0 0) \lvec (1 0)  \lvec (1 1) \lvec(1 1.25) \move (1.8 2) \lvec (2 2) \lvec (3 2) \lvec (3 3) \lvec (4 3) \lvec (4 4)
\move (-2 -1) \lvec (-1 -1) \lvec (-1 0) \lvec (0 0)  \lvec (0 1) \lvec (1 1) \lvec (1.25 1) \move (2 1.8) \lvec (2 2) \lvec (2 3) \lvec (3 3) \lvec (3 4) 
\move (-2 -2) \lvec (-1 -2) \lvec (-1 -1) \lvec (0 -1) \lvec (0 0) \lvec (1 0) \lvec (1 0.25)
\move (2 3) \lvec (2 4) \lvec (3 4) 
\move (-2 0) \lvec (-1 0) \lvec (-1 1) \lvec (0 1) \lvec (0.25 1)
\move (4 4) \lvec (4.25 4)
\move (4 4) \lvec (4 4.25)

\linewd 0.015 \lpatt (.1 .1 ) \move (-2 -2) \lvec (-2 -1) \lvec(-2 0) \move (3 4) \lvec(4 4)
\htext (4 -1.2) {\scriptsize{$x$}}
\htext (3 .1) {\scriptsize{$y$}}
\htext (1.3 1.3) {$\iddots$}
\htext (4.3 4.3) {$\iddots$}
\htext (-1.83 -1.79) {\scriptsize{$Y_{a,b}$}}
\htext (2.2 2.2) {\scriptsize{$Y_{a,b}$}}

\htext (-1.6 -2.2) {\small{$x_1$}} \htext (-0.6 -1.2) {\small{$x_2$}} \htext (0.4 -0.2) {\small{$x_3$}} \htext (2.4 1.8) {\small{$x_n$}} \htext (3.4 2.8) {\small{$x_1$}}
\htext (-0.98 -1.6) {\small{$y_1$}} \htext (0.02 -0.6) {\small{$y_2$}} \htext (1.02 0.4) {\small{$y_3$}} \htext (3.02 2.4) {\small{$y_n$}} \htext (4.02 3.4) {\small{$y_1$}}

\htext (-1.67 -0.95) {\small{$x_1^{(1)}$}} \htext (-0.66 0.05) {\small{$x_2^{(1)}$}} \htext (0.4 1.05) {\small{$x_3^{(1)}$}} \htext (2.4 3.05) {\small{$x_n^{(1)}$}} \htext (3.4 4.05) {\small{$x_1^{(1)}$}}
\htext (-2.38 -1.6) {\small{$y_1^{(1)}$}} \htext (-1.38 -0.63) {\small{$y_2^{(1)}$}} \htext (-0.38 0.4) {\small{$y_3^{(1)}$}}  \htext (1.62 2.4) {\small{$y_n^{(1)}$}} \htext (2.62 3.4) {\small{$y_1^{(1)}$}} 

\htext (-1.67 0.05) {\small{$x_1^{(2)}$}} \htext (-0.6 1.05) {\small{$x_2^{(2)}$}}  \htext (2.4 4.05) {\small{$x_n^{(2)}$}}
 \htext (-1.38 0.4) {\small{$y_3^{(2)}$}}  \htext (1.62 3.4) {\small{$y_1^{(2)}$}}
 }
\caption{{Transfer maps corresponding to the $n$-periodic initial value problem}}\label{staircase}
\end{figure}

For the transfer map\index{transfer map(s)} $T_n$ we define the \textit{monodromy matrix}\index{monodromy matrix}
\begin{equation}\label{monMat}
\mathcal{M}({\textbf{x},\textbf{y}};\lambda)=\prod_{j=1}^{\stackrel{\curvearrowleft}{n}}L(\textbf{y}_j;b)L(\textbf{x}_j;a),
\end{equation}
where $\textbf{x}:=(x_1,\ldots,x_n)$, $\textbf{y}:=(y_1,\ldots,y_n)$ and the ``$\curvearrowleft$" indicates that the terms of the above product are placed from the right to the left. Similarly to \cite{Veselov}, in \cite{kouloukasEnt} was proven that the transfer map\index{transfer map(s)} $T_n$ preserves the spectrum of its monodromy matrix.\index{monodromy matrix} Therefore, the function $\tr(\mathcal{M}({\textbf{x},\textbf{y}};\lambda))$ is a generating function of invariants for the map $T_n$.

As we will see in the next chapter, the invariants of a map are essential for integrability claims.

\begin{example}\normalfont
For Adler's map (\ref{Adler_map}) we consider the transfer matrix\index{transfer matrix} of the two-periodic initial value problem, given by
\begin{equation}
T_2(x_1,x_2,y_1,y_2)=\left(y_1-\frac{a-b}{x_1+y_1},y_2-\frac{a-b}{x_2+y_2},x_2+\frac{a-b}{x_2+y_2},x_1+\frac{a-b}{x_1+y_1}\right).
\end{equation}
Moreover, the corresponding monodromy matrix\index{monodromy matrix} is given by
\begin{equation}
\mathcal{M}_2(\textbf{x},\textbf{y};\lambda)=L(y_2;b)L(x_2;a)L(y_1;b)L(x_1;a),
\end{equation}
where $\textbf{x}:=(x_1,x_2)$ and $\textbf{y}:=(y_1,y_2)$.

The trace of the monodromy matrix\index{monodromy matrix} is given by the following second order polynomial
\begin{equation}
\tr\left(\mathcal{M}_2(\textbf{x},\textbf{y};\lambda)\right)=2 \lambda^2-\left(I_1(\textbf{x},\textbf{y})\right)^2 \lambda-I_0(\textbf{x},\textbf{y}),
\end{equation}
where $I_0$ and $I_1$ are invariants of the map $T_2$, and they are given by
\begin{subequations}
\begin{align}
I_1(\textbf{x},\textbf{y})&=a f(\textbf{x},\textbf{y})+b f(\textbf{y},\textbf{x})-\prod_{k,l=1}^2 (x_k+y_l),  \\
I_2(\textbf{x},\textbf{y})&=x_1+x_2+y_1+y_2,
\intertext{where we have omitted the constant terms, and $f$ is given by}
f(\textbf{x},\textbf{y})&=x_1^2+x_2^2+2y_1y_2+(x_1+x_2)(y_1+y_2).
\end{align}
\end{subequations}
\end{example}

\chapter{Yang-Baxter maps related to NLS type equations}
\label{chap5} \setcounter{equation}{0}
\renewcommand{\theequation}{\thechapter.\arabic{equation}}

\section{Overview}
As explained in the previous chapter, the construction of Yang-Baxter maps which admit Lax representation\index{Lax representation} and the study of their integrability is important, as they are related to several concepts of integrability. 

The aim of this chapter is to construct Yang-Baxter maps using the Darboux matrices\index{Darboux matrix(-ces)} we presented in chapter 3, and study their integrability. Particularly, we are interested in the Liouville integrability\index{Liouville integrability} of these maps, as finite dimensional maps. We shall present six and four-dimensional YB maps\index{Yang-Baxter (YB) map(s)} corresponding to all the NLS type equations which we considered in chapter 3.

The chapter is organised as follows: In the following section we give the definitions of a Poisson manifold, the Poisson bracket\index{Poisson!bracket} and Casimir functions\index{Casimir function(s)} that we use in the later sections, for the convenience of the reader. However, for more information on Poisson geometry one could refer to \cite{Arnold, Marsden}. Moreover, we shall prove some basic consequences of the matrix refactorisation problems;\index{refactorisation problem(s)} the birationality of the deduced YB maps\index{Yang-Baxter (YB) map(s)} and the derivation of their invariants. Finally, we will give the definition of the complete integrability of a YB map.\index{Yang-Baxter (YB) map(s)} 

In section 3 we construct six-dimensional YB maps\index{Yang-Baxter (YB) map(s)} for the all the NLS type equations which we considered in chapter 3. Furthermore, in the cases of the NLS and the DNLS equations the six-dimensional maps can be restricted to four-dimensional YB maps\index{Yang-Baxter (YB) map(s)} on invariant leaves.\index{invariant leaves} These maps deserve our attention, as they are related to several aspects of integrability; they are integrable in the Liouville sense as finite dimensional maps, they can be used to construct integrable lattices and they also have applications to a recent theory of maps preserving functions with symmetries \cite{Allan-Pavlos}.

Finally, section 4 deals with the vector generalisations of the four-dimensional YB maps.\index{Yang-Baxter (YB) map(s)} However, the Liouville integrability of these generalisations is an open problem.

The results of this chapter appear in \cite{Sokor-Sasha}.

\section{Preliminaries}
In what follows, we consider $M$ to be a differentiable manifold with $\dim(M)=n$, and $\mathcal{C}^{\infty}(M)$ the space of smooth functions defined on $M$.

\subsection{Poisson manifolds and Casimir functions}\index{Casimir function(s)}
Let us start with the definition of the Poisson bracket.\index{Poisson!bracket}

\begin{definition}
A map $\left\{,\right\}:\mathcal{C}^{\infty}(M)\times \mathcal{C}^{\infty}(M)\rightarrow \mathcal{C}^{\infty}(M)$ is called Poisson bracket,\index{Poisson!bracket} if it possesses the following properties:
\begin{enumerate}
	\item $\left\{\alpha f+\beta g,h\right\}=\alpha \left\{f,h\right\}+\beta \left\{g,h\right\}$; \qquad (bilinear)
	\item $\left\{f,g\right\}=-\left\{g,f\right\}$; \qquad (antisymmetric)
	\item $\left\{\left\{f,g\right\},h\right\}+\left\{\left\{h,f\right\},g\right\}+\left\{\left\{g,h\right\},f\right\}=0$; \qquad (Jacobi identity)
	\item $\left\{f,gh\right\}=\left\{f,g\right\}h+\left\{f,h\right\}g$; \qquad (Leibnitz rule)
\end{enumerate}for any $f,g,h\in \mathcal{C}^{\infty}(M)$ and $\alpha,\beta \in \field{C}$. Moreover, the manifold $M$ equipped with the above Poisson bracket\index{Poisson!bracket} is called Poisson manifold, and it is denoted as $\left(M,\{,\}\right)$.
\end{definition}

Yet, the above definition is abstract and, in practice, the Poisson bracket\index{Poisson!bracket} is usually defined by the following map
\begin{equation}\label{Poissonmap}
(f,g)\stackrel{J}{\rightarrow}\left\{f,g\right\}=\nabla f\cdot J\cdot \left(\nabla g\right)^t,
\end{equation}
where $J$ is an antisymmetric matrix which satisfies the Jacobi identity and it is called the \textit{Poisson matrix}.\index{Poisson!matrix} It can be readily verified that the above relation defines a Poisson bracket.\index{Poisson!bracket} Specifically, we have the following.

\begin{proposition}
Let $J=J(\textbf{x})$, $\textbf{x}\in M$, an $n\times n$ matrix. Then, $J$ defines a Poisson bracket\index{Poisson!bracket} via the relation (\ref{Poissonmap}) iff it is antisymmetric and it satisfies the following
\begin{equation}
\sum_{i=1}^n\left(J_{ij}\partial_{x_i}J_{kl}+J_{il}\partial_{x_i}J_{jk}+J_{ik}\partial_{x_i}J_{lj}\right)=0, \quad j,k,l=1,\ldots,n,
\end{equation}
for any $x=(x_1,\ldots,x_n)\in M$.
\end{proposition}

\begin{corollary}
Any constant antisymmetric matrix defines a Poisson bracket\index{Poisson!bracket} via the relation (\ref{Poissonmap}).
\end{corollary}

Now, for any two smooth functions on $M$, we have the following.
\begin{definition}
The functions $f,g\in \mathcal{C}^{\infty}(M)$ are said to be in involution with respect to a Poisson bracket (\ref{Poissonmap})\index{Poisson!bracket} if $\{f,g\}=0$.
\end{definition}

\begin{definition}
A function $C=C(\textbf{x})\in \mathcal{C}^{\infty}(M)$ is called Casimir function\index{Casimir function(s)} if it is in involution with any arbitrary function with respect to the Poisson bracket,\index{Poisson!bracket} namely $\left\{C,f\right\}=0$, for any $f=f(\textbf{x})\in M$.
\end{definition}

\subsection{Properties of the YB maps\index{Yang-Baxter (YB) map(s)} which admit Lax representation}\index{Lax representation}
Since the Lax equation, (\ref{eqLax}),\index{Lax equation} has the obvious symmetry 
\begin{equation}\label{symL}
(u,\,v,\,a,\,b) \longleftrightarrow  (y,\,x ,\,b,\,a)
\end{equation}
we have the following.

\begin{proposition} \label{rationality}
If a matrix refactorisation problem (\ref{eqLax})\index{refactorisation problem(s)} yields a rational map (\ref{unique-sol}), then this map is birational.
\end{proposition}

\begin{proof} Let $Y:(x,y)\mapsto (u,v)$ be a rational map corresponding to a refactorisation problem (\ref{eqLax}),\index{refactorisation problem(s)} i.e.
\begin{equation}
x\mapsto u=\frac{n_1(x,y;a,b)}{d_1(x,y;a,b)}, \qquad y\mapsto v=\frac{n_2(x,y;a,b)}{d_2(x,y;a,b)},
\end{equation}
where $n_i$, $d_i$, $i=1,2$, are polynomial functions of their variables.

Due to the symmetry (\ref{symL}) of the refactorisation problem (\ref{eqLax}),\index{refactorisation problem(s)} the inverse map of $Y$, $Y^{-1}:(x,y)\mapsto (u,v)$, is also rational and, in fact,
\begin{equation}
u\mapsto x=\frac{n_1(v,u;b,a)}{d_1(v,u;b,a)}, \qquad v\mapsto y=\frac{n_2(v,u;b,a)}{d_2(v,u;b,a)}.
\end{equation} 
Therefore, $Y$ is a birational map.
\end{proof}

Now, for a Yang-Baxter map, $Y_{a,b}$, the quantity $M(x,y;a,b)=L(y;b)L(x;a)$ is called the \textit{monodromy} matrix.\index{monodromy matrix} The fact that the monodromy matrix\index{monodromy matrix} is a generating function of first integrals\index{first integral(s)} is well known from the eighties; for example, see \cite{FadTah}\footnote{Reprint of the 1987 edition.}. In particular, for the invariants of a YB map\index{Yang-Baxter (YB) map(s)} admiting Lax representation,\index{Lax representation} we have the following.

\begin{proposition}\label{genInv}
If $L=L(x,a;\lambda)$ is a Lax matrix\index{Lax matrix(-ces)} with corresponding YB map,\index{Yang-Baxter (YB) map(s)} $Y:(x,y)\mapsto (u,v)$, then the $\tr(L(y,b;\lambda)L(x,a;\lambda))$ is a generating function of invariants of the YB map.\index{Yang-Baxter (YB) map(s)}
\end{proposition}
\begin{proof}
Since, 
\begin{equation} \label{trace}
\tr(L(u,a;\lambda)L(v,b;\lambda))\overset{(\ref{eqLax})}{=}\tr(L(y,b;\lambda)L(x,a;\lambda))=\tr(L(x,a;\lambda)L(y,b;\lambda)),
\end{equation}
and the function $\tr(L(x,a;\lambda)L(y,b;\lambda))$ can be written as $\displaystyle\tr(L(x,a;\lambda)L(y,b;\lambda))=\sum_k \lambda^k I_k(x,y;a,b)$, from (\ref{trace}) follows that
\begin{equation}
I_i(u,v;a,b)=I_i(x,y;a,b),
\end{equation}
which are invariants for $Y$.
\end{proof}

Nevertheless, the above proposition does not guarantee that the generated invariants, $I_i(x,y;a,b)$, are functionally independent. Moreover, the number of the invariants we deduce from the trace of the monodromy matrix\index{monodromy matrix} may not be enough for integrability claims.

\subsection{Liouville integrability of Yang-Baxter maps}
The invariants of a YB map\index{Yang-Baxter (YB) map(s)} are essential towards its integrability in the Liouville sense. Here, we define the complete (Liouville) integrability \index{complete integrability}\index{Liouville integrability} of a YB map,\index{Yang-Baxter (YB) map(s)} following \cite{Fordy, Veselov4}.

Particularly, we have the following.

\newtheorem{CompleteIntegrability}{Definition}[section]
\begin{CompleteIntegrability}
A $2N$-dimensional Yang-Baxter map,  
\begin{equation}
Y:(x_1,\ldots,x_{2N})\mapsto (u_1,\ldots,u_{2N}), \quad u_i=u_i(x_1,\ldots,x_{2N}), \quad i=1,\ldots,2N, \nonumber
\end{equation}
is said to be completely integrable or Liouville integrable\index{completely integrable}\index{Liouville integrable} if
\begin{enumerate}
	\item there is a Poisson matrix,\index{Poisson!matrix} $J_{ij}=\left\{x_i,x_j\right\}$, of rank $2r$, which is invariant under the action of the YB map, namely $J_{ij}$ and $\tilde{J_{ij}}=\left\{u_i,u_j\right\}$ have the same functional form of their respective arguments,
	\item map $Y$ has $r$ functionally independent invariants, $I_i$, namely $I_i\circ Y=I_i$, which are in involution with respect to the corresponding Poisson bracket,\index{Poisson!bracket} i.e. $\left\{I_i,I_j\right\}=0$, $i,j=1,\ldots,r$, $i\neq j$,
	\item there are $k=2N-2r$ Casimir functions,\index{Casimir function(s)} namely functions $C_i$, $i=1,\ldots,k$, such that $\left\{C_i,f\right\}=0$, for any arbitrary function $f=f(x_1,...,x_{2N})$. These are invariant under $Y$, namely $C_i\circ Y=C_i$.
\end{enumerate}
\end{CompleteIntegrability} 

We will use this definition to study the integrability of the YB maps\index{Yang-Baxter (YB) map(s)} presented in the following section.

\section{Derivation of Yang-Baxter maps}
In chapter 3 we used Darboux transformations \index{Darboux transformation(s)}to construct integrable systems of discrete equations, which have the multidimensional consistency property.\index{multidimensional consistency} The compatibility condition of Darboux transformations\index{Darboux transformation(s)} around the square is exactly the same with the Lax equation (\ref{eqLax}).\index{Lax equation} Therefore, in this section, we use Darboux transformations\index{Darboux transformation(s)} to construct YB maps.\index{Yang-Baxter (YB) map(s)}

In particular, we consider Darboux matrices for the NLS type equations studied in chapter 3; the NLS equation, the DNLS equation and a deformation of the DNLS equation. For these Darboux matrices the refactorisation is not unique. Therefore, for the corresponding six-dimensional YB maps\index{Yang-Baxter (YB) map(s)} which are derived from the refactorisation problem,\index{refactorisation problem(s)} in principle, one needs to check the YB property separately. Yet, the entries of these Darboux matrices\index{Darboux matrix(-ces)} obey certain differential equations which possess first integrals.\index{first integral(s)} 

There is a natural restriction of the Darboux map on the affine variety corresponding to a level set of these first integrals.\index{first integral(s)} These restrictions make the refactorisation unique and this guarantees that the induced four-dimensional YB maps\index{Yang-Baxter (YB) map(s)} satisfy the YB equation \index{Yang-Baxter (YB) equation} and they are reversible \cite{Veselov3}. Moreover, we will show that the latter YB maps\index{Yang-Baxter (YB) map(s)} have Poisson structure.\index{Poisson!structure} 

However, the first integrals\index{first integral(s)} are not always very useful for the reduction because, in general, they are polynomial equations. In particular, in the cases of NLS and DNLS equations we present six-dimensional YB maps\index{Yang-Baxter (YB) map(s)} and their four-dimensional restrictions on invariant leaves.\index{invariant leaves} These four-dimensional restrictions are birational YB maps\index{Yang-Baxter (YB) map(s)} and we prove that they are integrable in the Liouville sense. In the case of the deformation of the DNLS equation, we present a six-dimensional YB map and a linear approximation to the four-dimensional YB map.\index{Yang-Baxter (YB) map(s)}

We start with the well known example of the Darboux transformation\index{Darboux transformation(s)} for the nonlinear Schr\"odinger equation \index{nonlinear Schr\"odinger (NLS)!equation} and construct its associated YB map.\index{Yang-Baxter (YB) map(s)}

\subsection{The Nonlinear Schr\"odinger equation}
Recall that, in the case of NLS equation,\index{nonlinear Schr\"odinger (NLS)!equation} the Lax operator\index{Lax operator(s)!for NLS equation} is given by
\begin{equation}
\mathcal{L}(p,q;\lambda)=D_x+\lambda U_{1}+U_{0},\quad \text{where} \quad U_1=\sigma_3,\quad U_0=\left(\begin{matrix}
        0 & 2p \\
        2q & 0
    \end{matrix}\right),
\end{equation}
where $\sigma_3$ is the standard Pauli matrix, i.e. $\sigma_3=\text{diag}(1,-1)$.

Moreover, a Darboux matrix \index{Darboux matrix(-ces)} for $\mathcal{L}$ is given by 
\begin{equation}\label{NLSDarboux}
  M=\lambda \left(
     \begin{matrix}
         1 & 0\\
         0 & 0
     \end{matrix}\right)+\left(
     \begin{matrix}
         f & p\\
         q_{10} & 1
     \end{matrix}\right).
\end{equation}
The entries of (\ref{NLSDarboux}) must satisfy the following system of equations
\begin{equation}\label{baecklundNLS}
\partial_x f=2(pq-p_{10}q_{10}), \qquad \partial_x p=2(pf-p_{10}), \qquad \partial_x q_{10}=2(q-q_{10}f),
\end{equation}
which admits the following first integral\index{first integral(s)}
\begin{equation}\label{integralNLS}
\partial_x(f-pq_{10})=0.
\end{equation}
This integral implies that $\partial_x \det M=0$. 

In correspondence with (\ref{NLSDarboux}), we define the matrix
\begin{equation} \label{3d-Darboux-NLS}
  M(\textbf{x};\lambda)=\lambda \left(
     \begin{matrix}
         1 & 0\\
         0 & 0
     \end{matrix}\right)+\left(
     \begin{matrix}
         X & x_1\\
         x_2 & 1
     \end{matrix}\right),
\qquad \textbf{x}=(x_1,x_2,X),
\end{equation}
and substitute it into the Lax equation (\ref{eqLax}) \index{Lax equation}
\begin{equation}\label{laxM}
M(\textbf{u};\lambda)M(\textbf{v};\lambda)=M(\textbf{y};\lambda)M(\textbf{x};\lambda),
\end{equation}
to derive the following system of equations
\begin{eqnarray}
&v_1 = x_1,\ u_2 = y_2,\ U +V = X + Y ,\ u_2 v_1 = x_1 y_2,& \nonumber \\
&u_1 +U v_1 = y_1+x_1 Y,\ u_1 v_2+U V = x_2 y_1+X Y,\ v_2+u_2 V = x_2 + X y_2.& \nonumber
\end{eqnarray}

The corresponding algebraic variety is a union of two six-dimensional components. The first one is obvious from the refactorisation problem (\ref{laxM}),\index{refactorisation problem(s)} and it corresponds to the permutation map
\begin{equation}
 \textbf{x}\mapsto \textbf{u}=\textbf{y}, \quad \textbf{y}\mapsto \textbf{v}=\textbf{x},  \nonumber
\end{equation}
which is a (trivial) YB map.\index{Yang-Baxter (YB) map(s)} The second one can be represented as a rational six-dimensional non-involutive map of $K^3\times K^3 \rightarrow K^3\times K^3$
\small
\begin{subequations}\label{NLS-3d} 
 \begin{align}
&x_1\mapsto u_1=\frac{y_1+x_1^2x_2-x_1X+x_1Y}{1+x_1y_2},\quad y_1\mapsto v_1=x_1, \\ 
&x_2\mapsto u_2=y_2,\quad\qquad\qquad\qquad\qquad \quad~~~ y_2\mapsto v_2=\frac{x_2+y_1y_2^2+y_2X-y_2Y}{1+x_1y_2},  \\
&X\mapsto U=\frac{y_1y_2-x_1x_2+X+x_1y_2Y}{1+x_1y_2}, ~~~Y\mapsto V=\frac{x_1x_2-y_1y_2+x_1y_2X+Y}{1+x_1y_2},
 \end{align}
\end{subequations}
\normalsize 
which, one can easily check that, satisfies the YB equation.\index{Yang-Baxter (YB) equation}

The trace of $M(\textbf{y};\lambda)M(\textbf{x};\lambda)$ is a polynomial in $\lambda$ whose coefficients are
\begin{equation}
\mbox{tr}(M(\textbf{y};\lambda)M(\textbf{x};\lambda))=\lambda^2+\lambda I_1(\textbf{x},\textbf{y})+I_2(\textbf{x},\textbf{y}), \nonumber
\end{equation}
where
\begin{equation}\label{NLS3dInv}
I_1(\textbf{x},\textbf{y})=X+Y \qquad \text{and} \qquad I_2(\textbf{x},\textbf{y})=x_2y_1+x_1y_2+XY,
\end{equation}
and those, according to proposition \ref{genInv}, are invariants for the YB map (\ref{NLS-3d}).\index{Yang-Baxter (YB) map(s)}

In the following section we show that the YB map (\ref{NLS-3d})\index{Yang-Baxter (YB) map(s)} can be restricted to a four-dimensional YB map\index{Yang-Baxter (YB) map(s)} which has Poisson structure.\index{Poisson!structure}


\subsubsection{Restriction on symplectic leaves: The Adler-Yamilov map}\index{Adler-Yamilov!map}\index{symplectic leaf(-ves)}
In this section, we show that map (\ref{NLS-3d}) can be restricted to the Adler-Yamilov map\index{Adler-Yamilov!map} on symplectic leaves,\index{symplectic leaf(-ves)} by taking into account the first integral,\index{first integral(s)} (\ref{integralNLS}), of the system (\ref{baecklundNLS}). 

In particular, we have the following.

\begin{proposition} For the six-dimensional map (\ref{NLS-3d}) we have the following:
\begin{enumerate}
	\item The quantities $\Phi =X-x_1x_2$ and $\Psi=Y-y_1y_2$ are its invariants (first integrals).\index{first integral(s)}
	\item It can be restricted to a four-dimensional map $Y_{a,b}:A_a\times A_b \longrightarrow A_a\times A_b$, where $A_a$, $A_b$ are level sets of the first integrals\index{first integral(s)} $\Phi$ and $\Psi$, namely
\begin{subequations}\label{symleaves}
\begin{align}
A_a&=\{(x_1,x_2,X)\in K^3; X=a+x_1x_2\}, \\
A_b&=\{(y_1,y_2,Y)\in K^3; Y=b+y_1y_2\}.
\end{align}
\end{subequations}
\end{enumerate}
Moreover, map $Y_{a,b}$ is the Adler-Yamilov map.\index{Adler-Yamilov!map}
\end{proposition}

\begin{proof}
\vspace{-1cm}
\begin{enumerate}
	\item It can be readily verified that (\ref{NLS-3d}) implies $U-u_1u_2=X-x_1x_2$ and $V-v_1v_2=Y-y_1y_2$. Thus, $\Phi$ and $\Psi$ are invariants, i.e.  first integrals of the map.\index{first integral(s)}
	\item The existence of the restriction is obvious. Using the conditions $X=x_1x_2+a$ and $Y=y_1y_2+b$, one can eliminate $X$ and $Y$ from (\ref{NLS-3d}). The resulting map, $\textbf{x}\rightarrow \textbf{u}(\textbf{x},\textbf{y})$, $\textbf{y}\rightarrow \textbf{v}(\textbf{x},\textbf{y})$, is given by
\begin{eqnarray} \label{YB_NLS}
(\textbf{x},\textbf{y})\overset{Y_{a,b}}{\longrightarrow }\left(y_1-\frac{a -b}{1+x_1y_2}x_1,y_2,x_1,x_2+\frac{a -b}{1+x_1y_2}y_2\right).
\end{eqnarray}
\end{enumerate}
Map (\ref{YB_NLS}) coincides with the Adler-Yamilov map.\index{Adler-Yamilov!map}
\end{proof}

Map (\ref{YB_NLS}) originally appeared in the work of Adler and Yamilov \cite{Adler-Yamilov}. Moreover, it appears as a YB map in \cite{kouloukas, PT}.\index{Yang-Baxter (YB) map(s)}

Now, one can use the condition $X=x_1x_2+a$ to eliminate $X$ from the Lax matrix\index{Lax matrix(-ces)} (\ref{3d-Darboux-NLS}), i.e. 
\begin{equation} \label{laxNLS}
M(\textbf{x};a,\lambda)=\lambda \left(
\begin{matrix}
 1 & 0\\
 0 & 0
\end{matrix}\right)+\left(
\begin{matrix}
 a+x_1x_2 & x_1\\
 x_2 & 1
\end{matrix}\right), \quad \textbf{x}=(x_1,x_2).
\end{equation}
The form of Lax matrix (\ref{laxNLS})\index{Lax matrix(-ces)} coincides with the well known Darboux transformation\index{Darboux transformation(s)} for the NLS equation (see \cite{Rog-Schief} and references therein). Now, Adler-Yamilov map \index{Adler-Yamilov!map}follows from the strong Lax representation\index{Lax representation}
\begin{equation} \label{lax_eq_NLS}
  M(\textbf{u};a,\lambda)M(\textbf{v};b,\lambda)=M(\textbf{y};b,\lambda)M(\textbf{x};a,\lambda).
\end{equation}
Therefore, the Adler-Yamilov map (\ref{YB_NLS})\index{Adler-Yamilov!map} is a reversible parametric YB map\index{Yang-Baxter (YB) map(s)!parametric} with strong Lax matrix\index{Lax matrix(-ces)} ($\ref{laxNLS}$). Moreover, it is easy to verify that it is not involutive.

For the integrability of this map we have the following

\begin{proposition}\label{cintA-Y}
The Adler-Yamilov map (\ref{YB_NLS}) is completely integrable.\index{Adler-Yamilov!map}\index{completely integrable}
\end{proposition}

\begin{proof}
From the trace of $M(\textbf{y};b,\lambda)M(\textbf{x};a,\lambda)$ we obtain the following invariants for the map (\ref{YB_NLS})
\begin{eqnarray}
&&I_1(\textbf{x},\textbf{y})=x_1 x_2+y_1 y_2+a+b, \\
&&I_2(\textbf{x},\textbf{y})=(a+x_1 x_2)(b+y_1 y_2)+x_1 y_2+x_2 y_1+1.
\end{eqnarray}
The constant terms in $I_1,I_2$ can be omitted. It is easy to check that $I_1,I_2$ are in involution with respect to a Poisson bracket\index{Poisson!bracket} defined as
\begin{equation}
\left\{x_1,x_2\right\}=\left\{y_1,y_2\right\}=1,\qquad \text{and all the rest}
\qquad \left\{x_i,y_j\right\}=0,
\end{equation}
and the corresponding Poisson matrix\index{Poisson!matrix} is invariant under the YB map (\ref{YB_NLS}).\index{Yang-Baxter (YB) map(s)} Therefore the map (\ref{YB_NLS}) is completely integrable.\index{completely integrable}
\end{proof}

The above proposition implies the following.

\begin{corollary}
The invariant leaves\index{invariant leaves} $A_a$ and $B_b$, given in (\ref{symleaves}), are symplectic.
\end{corollary}


\subsection{Derivative NLS equation: $\field{Z}_2$ reduction}
Recall that the Lax operator for the DNLS equation \cite{CLH, Kaup-Newell} \index{Lax operator(s)!for DNLS equation}\index{derivative nonlinear Schr\"odinger (DNLS)!equation} is given by
\begin{equation}
   \mathcal{L}(p,q;\lambda)=D_x+\lambda^{2} U_2+\lambda U_1,\qquad \text{where} \qquad U_2=\sigma_3,\qquad U_1=\left(\begin{matrix} 0 & 2p \\ 2q & 0\end{matrix}\right),
\end{equation}
and $\sigma_3$ is a Pauli matrix. The operator $\mathcal{L}$ is invariant with respect to the following 
\begin{equation}\label{sym_cond}
  \mathcal{L}(\lambda)\rightarrow\sigma_{3}\mathcal{L}(-\lambda) \sigma_{3},
\end{equation}
where $\mathcal{L}(\lambda)\equiv \mathcal{L}(p,q;\lambda)$. In particular, the involution ($\ref{sym_cond}$) generates the so-called reduction group\index{reduction group} \cite{Mikhailov2,Lombardo} and it is isomorphic to $\field{Z}_2$. 

The Darboux matrix\index{Darboux matrix(-ces)} in this case is given by 
\begin{equation} \label{affDarboux}
  M:=\lambda^{2}\left(\begin{matrix}
    f & 0\\
    0 & 0\end{matrix}\right)+\lambda\left(\begin{matrix}
    0 & f p\\
    fq_{10} & 0\end{matrix}\right)+\left(\begin{matrix}
    c & 0\\
    0 & 1\end{matrix}\right),
\end{equation} 
whose entries $p$, $q_{10}$ and $f$ obey the following system of equations
\begin{subequations}
\begin{align}\label{baecklundZ2-a}
\partial_x p&=2p(p_{10}q_{10}-pq)-\frac{2}{f}(p_{10}-cp), \\
\label{baecklundZ2-b}
\partial_x q_{10}&=2q_{10}(p_{10}q_{10}-pq)-\frac{2}{f}(cq_{10}-q), \\
\label{baecklundZ2-c}
\partial_x f&=2f(pq-p_{10}q_{10}). 
\end{align}
\end{subequations}
The system (\ref{baecklundZ2-a})-(\ref{baecklundZ2-c}) has a first integral\index{first integral(s)} which obliges the determinant of matrix (\ref{affDarboux}) to be $x$-independent, and it is given by
\begin{equation}\label{alpha-eq}
\partial_x (f^{2}pq_{10}-f)=0.
\end{equation}

Using the entries of (\ref{affDarboux}) as variables, namely $(p,q_{10},f;c)\rightarrow (x_1,x_2,X;1)$, we define the matrix
\begin{equation}\label{affDarboux-3D}
  M(\textbf{x};\lambda)=\lambda^{2}\left(\begin{matrix}
    X & 0\\
    0 & 0\end{matrix}\right)+\lambda\left(\begin{matrix}
    0 & x_1X\\
    x_2X & 0\end{matrix}\right)+\left(\begin{matrix}
    1 & 0\\
    0 & 1\end{matrix}\right), \quad \textbf{x}=(x_1,x_2,X).
\end{equation} 

The Lax equation\index{Lax equation} implies the following equations
\begin{equation} 
 \begin{array}[pos]{ccc}
u_1U+v_1V=x_1X+y_1Y, \quad u_2U+v_2V=x_2X+y_2Y,\\ 
UV=XY, \quad v_1UV=x_1XY, \quad u_2UV=y_2XY, \quad u_2v_1UV=x_1y_2XY,  \\
U+V+u_1v_2UV=X+Y+x_2y_1XY.   
\end{array}
\end{equation}

As in the case of nonlinear Schr\"odinger equation, the algebraic variety consists of two components. The first six-dimensional component corresponds to the permutation map
\begin{equation} \label{perMap}
 \textbf{x}\mapsto \textbf{u}=\textbf{y}, \quad \textbf{y}\mapsto \textbf{v}=\textbf{x},  
\end{equation}
and the second corresponds to the following six-dimensional YB map\index{Yang-Baxter (YB) map(s)}
\begin{equation}\label{f-g-h}
\begin{array}[pos]{lcl}
x_1 \mapsto u_1=f_1(\textbf{x},\textbf{y}),\qquad y_1 \mapsto v_1=f_2(\pi\textbf{y},\pi\textbf{x}), \\ 
x_2 \mapsto u_2=f_2(\textbf{x},\textbf{y}),\qquad y_2 \mapsto v_2=f_1(\pi\textbf{y},\pi\textbf{x}), \\ 
\,X \mapsto U=f_3(\textbf{x},\textbf{y}),\qquad \, Y \mapsto V=f_3(\pi\textbf{y},\pi\textbf{x}),
\end{array}
\end{equation}
where $\pi$ is the \textit{permutation function}, $\pi(x_1,x_2,X)=(x_2,x_1,X)$, $\pi^2=1$ and $f_1,f_2$ and $f_3$ are given by
\begin{subequations}
\begin{align} \label{f1-Aff}
\hspace{-0.5cm} f_1(\textbf{x},\textbf{y})&=\frac{-1}{f_3(\textbf{x},\textbf{y})}\frac{x_1X+(y_1-x_1)Y-x_1x_2y_1XY-x_1^2x_2X^2}{x_1 x_2 X+x_1 y_2 Y-1}, \\
\label{f2-Aff}
\hspace{-0.5cm} f_2(\textbf{x},\textbf{y})&=y_2, \\ 
\label{f3-Aff}
\hspace{-0.5cm} f_3(\textbf{x},\textbf{y})&=\frac{x_1x_2X+x_1y_2Y-1}{x_1y_2X+y_1 y_2 Y-1}X.
\end{align}
\end{subequations}

One can verify that the above map is a non-involutive YB map.\index{Yang-Baxter (YB) map(s)} The invariants of this map are given by 
\begin{equation}
I_1(\textbf{x},\textbf{y})=XY\qquad \text{and}\qquad I_2(\textbf{x},\textbf{y})=(\textbf{x}\cdot\pi\textbf{y})XY+X+Y.
\end{equation}

The map $Y:K^3\times K^3 \rightarrow K^3\times K^3$, given by $\{(\ref{f-g-h}),(\ref{f1-Aff})-(\ref{f3-Aff})\}$, can be restricted to a map of the Cartesian product of two two-dimensional affine varieties
\begin{subequations}\label{symleaves2}
\begin{align}
A_a&=\left\{(x_1,x_2,X)\in K^3; X-X^2x_1x_2=a\in K\right\}, \\
A_b&=\left\{(y_1,y_2,Y)\in K^3; Y-Y^2y_1y_2=b\in K\right\},
\end{align}
\end{subequations}
which are invariant varieties of the map $Y$. Thus, the YB map, $Y_{a,b}$,\index{Yang-Baxter (YB) map(s)} is a birational map $Y_{a,b}:A_a\times A_b \rightarrow A_a\times A_b$.

It is easy to uniformise the rational variety $A_a$ and express the YB map explicitly.\index{Yang-Baxter (YB) map(s)} The equations defining the varieties $A_a$ and $A_b$ are linear in $x_1, x_2$ and $y_1,y_2$, respectively. Hence, we can express
\begin{equation}
x_2=\frac{1}{x_1X}-\frac{a}{x_1X^2}, \qquad y_2=\frac{1}{y_1Y}-\frac{b}{y_1Y^2}.
\end{equation}
The resulting map is given by
\begin{equation}
x_1\mapsto u_1=\frac{h_1}{h_2}y_1, \quad X\mapsto U=h_2Y, \quad y_1\mapsto v_1=x_1, \quad Y\mapsto V=\frac{1}{h_2}X
\end{equation}
where the quantities $h_i$, $i=1,2$, are given by
\begin{equation}
h_1=\frac{ay_1Y+x_1X(a-Y)}{ay_1Y+x_1X(b-Y)}, \quad h_2=\frac{ay_1Y+x_1X(b-Y)}{by_1Y+x_1X(b-X)}.
\end{equation}

Nevertheless, in the next section, we present a more symmetric way to parametrise the varieties $A_a$, $A_b$ and the Lax matrix.\index{Lax matrix(-ces)}


\subsubsection{$\field{Z}_2$ reduction: A reducible six-dimensional YB map}\index{Yang-Baxter (YB) map(s)} 
Now, let us go back to the Darboux matrix (\ref{affDarboux})\index{Darboux matrix(-ces)} and  replace $(f p,f q_{10},f;c)\rightarrow (x_1,x_2,X;1)$, namely  
\begin{equation}\label{Affine-3d1}
  M(\textbf{x};\lambda)=\lambda^{2}\left(\begin{matrix}
    X & 0\\
    0 & 0\end{matrix}\right)+\lambda\left(\begin{matrix}
    0 & x_1\\
    x_2 & 0\end{matrix}\right)+\left(\begin{matrix}
    1 & 0\\
    0 & 1\end{matrix}\right), \quad \textbf{x}=(x_1,x_2,X).
\end{equation} 
 
From the Lax equation\index{Lax equation} we obtain the following equations
\begin{eqnarray}
&u_2v_1=x_1y_2, \quad u_2V=Xy_2, \quad Uv_1=x_1Y, \quad UV=XY& \nonumber \\
&u_1+v_1=x_1+y_1, \quad U+u_1v_2+V=X+x_2y_1+Y, \quad u_2+v_2=x_2+y_2.& \nonumber
\end{eqnarray}

Now, the first six-dimensional component of the algebraic variety corresponds to the trivial map (\ref{perMap}) and the second component corresponds to a map of the form (\ref{f-g-h}), with $f_1,f_2$ and $f_3$ now given by
\begin{subequations}
\begin{align} \label{YB-Aff-3d1}
f_1(\textbf{x},\textbf{y})&=\frac{(x_1+y_1)X-x_1Y-x_1x_2(x_1+y_1)}{X-x_1(x_2+y_2)}, \\
\label{YB-Aff-3d2}
f_2(\textbf{x},\textbf{y})&=\frac{X-x_1(x_2+y_2)}{Y-y_2(x_1+y_1)}y_2,  \\ 
\label{YB-Aff-3d3}
f_3(\textbf{x},\textbf{y})&=\frac{X-x_1(x_2+y_2)}{Y-y_2(x_1+y_1)}Y. 
\end{align}
\end{subequations}

This map has the following invariants
\begin{subequations}\label{invar-aff-3d}
\begin{align}
I_1(\textbf{x},\textbf{y})&=XY, \qquad \quad \,\,I_2(\textbf{x},\textbf{y})=\textbf{x}\cdot\pi\textbf{y}+X+Y,\\
\label{invar-aff-3d-2}
I_3(\textbf{x},\textbf{y})&=x_1+y_1, \qquad I_4(\textbf{x},\textbf{y})=x_2+y_2.
\end{align}
\end{subequations}


\subsubsection{Restriction on invariant leaves}\index{invariant leaves}
In this section, we show that the map given by (\ref{f-g-h}) and (\ref{YB-Aff-3d1})-(\ref{YB-Aff-3d3}) can be restricted to a completely integrable\index{completely integrable} four-dimensional map on invariant leaves.\index{invariant leaves} As in the previous section, the idea of this restriction is motivated by the existence of the first integral\index{first integral(s)} of the system (\ref{baecklundZ2-a})-(\ref{baecklundZ2-c}),
\begin{equation}
f-(fp)(fq_{10})=k=constant,
\end{equation}

Particularly, we have the following.

\begin{proposition} For the six-dimensional map $\left\{(\ref{f-g-h}),(\ref{YB-Aff-3d1})-(\ref{YB-Aff-3d3})\right\}$ we have the following:
\begin{enumerate}
	\item The quantities $\Phi =X-x_1x_2$ and $\Psi=Y-y_1y_2$ are its first integrals.\index{first integral(s)}
	\item It can be restricted to a four-dimensional map $Y_{a,b}:A_a\times A_b \longrightarrow A_a\times A_b$, given by
\begin{equation} \label{YB-affine}
\hspace{-.3cm}\left(\textbf{x},\textbf{y}\right)\overset{Y_{a,b}}{\longrightarrow }\left(y_1+\frac{a-b }{a-x_1y_2}x_1, \frac{a-x_1 y_2}{b-x_1 y_2}y_2, \frac{b-x_1 y_2}{a-x_1 y_2}x_1,x_2+\frac{b-a}{b-x_1 y_2}y_2\right),
\end{equation}
and $A_a$, $A_b$ are given by (\ref{symleaves}).
\end{enumerate}
\end{proposition}

\begin{proof}
\vspace{-1cm}
\begin{enumerate}
	\item Map $\left\{(\ref{f-g-h}),(\ref{YB-Aff-3d1})-(\ref{YB-Aff-3d3})\right\}$ implies $U-u_1u_2=X-x_1x_2$ and $V-v_1v_2=Y-y_1y_2$. Therefore, $\Phi$ and $\Psi$ are first integrals of the map.\index{first integral(s)}
	\item The conditions $X=x_1x_2+a$ and $Y=y_1y_2+b$ define the level sets, $A_a$ and $A_b$, of $\Phi$ and $\Psi$, respectively. Using these conditions, we can eliminate $X$ and $Y$ from map $\left\{(\ref{f-g-h}),(\ref{YB-Aff-3d1})-(\ref{YB-Aff-3d3})\right\}$. The resulting map, $Y_{a,b}:A_a\times A_b \longrightarrow A_a\times A_b$, is given by (\ref{YB-affine}).
\end{enumerate}
\end{proof}

Now, using condition $X=x_1x_2+a$, matrix (\ref{Affine-3d1}) takes the following form
\begin{equation} \label{Darboux_Affine}
  M(\textbf{x};k;\lambda)=\lambda^{2}\left(\begin{matrix}
    k+x_1x_2 & 0\\
    0 & 0\end{matrix}\right)+\lambda\left(\begin{matrix}
    0 & x_1\\
    x_2 & 0\end{matrix}\right)+\left(\begin{matrix}
    1 & 0\\
    0 & 1\end{matrix}\right).
\end{equation} 
Now, map (\ref{YB-affine}) follows from the strong Lax representation (\ref{lax_eq_NLS}).\index{Lax representation} Therefore, it is reversible parametric YB map.\index{Yang-Baxter (YB) map(s)!parametric} It can also be verified that it is not involutive.

For the integrability of map (\ref{YB-affine}) we have the following

\begin{proposition}\label{cintZ2}
Map (\ref{YB-affine}) is completely integrable.\index{completely integrable}
\end{proposition}

\begin{proof}
The invariants of map (\ref{YB-affine}) which we retrieve from the trace of $M(\textbf{y};b,\lambda)M(\textbf{x};a,\lambda)$ are
\begin{equation}
I_1(\textbf{x},\textbf{y})=(a+x_1x_2)(b+y_1y_2), \qquad I_2(\textbf{x},\textbf{y})=(x_1+y_1)(x_2+y_2)+a+b.
\end{equation}
However, the quantities $x_1+y_1$ and $x_2+y_2$ in $I_2$ are invariants themselves. The Poisson bracket\index{Poisson!bracket} in this case is given by
\begin{equation}
\left\{x_1,x_2\right\}=\left\{y_1,y_2\right\}=\left\{x_2,y_1\right\}=\left\{y_2,x_1\right\}=1,\quad \text{and all the rest} \quad \left\{x_i,y_j\right\}=0.
\end{equation}
The rank of the Poisson matrix\index{Poisson!martix} is 2, $I_1$ is one invariant and $I_2=C_1C_2+a+b$, where $C_1=x_1+y_1$ and $C_2=x_2+y_2$ are Casimir functions.\index{Casimir function(s)} The latter are preserved by (\ref{YB-affine}), namely $C_i\circ Y_{a,b}=C_i$, $i=1,2$. Therefore, map (\ref{YB-affine}) is completely integrable.\index{completely integrable}
\end{proof}

\begin{corollary}
Map (\ref{YB-affine}) can be expressed as a map of two variables on the symplectic leaf\index{symplectic leaf(-ves)}
\begin{equation}
x_1+y_1=c_1, \qquad x_2+y_2=c_2.
\end{equation}
\end{corollary}

\subsection{A deformation of the DNLS equation: Dihedral Group}
Recall that, in the case of the deformation of the DNLS equation,\index{deformation of DNLS!equation} the Lax operator\index{Lax operator(s)!for deformation of DNLS equation} is given by
\begin{equation}\label{laxDihedral}
\begin{array}[pos]{ccl}
 \mathcal{L}(p,q;\lambda)=D_x+\lambda^{2}U_2+\lambda U_1+\lambda^{-1} U_{-1}-\lambda^{-2}U_{-2},\qquad \text{where} \\ 
 U_{2}\equiv U_{-2}=\sigma_3, \qquad 
 U_1=\left(\begin{matrix} 0 & 2p\\ 2q & 0\end{matrix}\right), \qquad U_{-1}=\sigma_1 U_1 \sigma_1,
\end{array}
\end{equation}
and $\sigma_1$, $\sigma_3$ are Pauli matrices. Here, the reduction group consists of the following set of transformations acting on the Lax operator\index{Lax operator(s)} (\ref{laxDihedral}) ,
\begin{equation}
 \mathcal{L}(\lambda)\rightarrow\sigma_3\mathcal{L}(-\lambda)\sigma_3 \qquad \text{and} \qquad \mathcal{L}(\lambda)
 \rightarrow\sigma_1\mathcal{L}(\lambda^{-1})\sigma_1,
\label{reduct_group}
\end{equation}
and it is isomorphic to $\field{Z}_2\times\field{Z}_2\cong D_2$, \cite{Lombardo}.

In this case, the Darboux matrix\index{Darboux matrix(-ces)} is given by 
\begin{equation} \label{Darboux-D2}
  M=f\left(\left(\begin{matrix} 
    \lambda^2 & 0\\
    0 & \lambda^{-2}\end{matrix}\right)+\lambda\left(\begin{matrix}
    0 & p\\
    q_{10} & 0\end{matrix}\right)+
    g\left(\begin{matrix} 
    1 & 0\\
    0 & 1\end{matrix}\right)
    +\frac{1}{\lambda}\left(\begin{matrix}
    0 & q_{10}\\
    p & 0\end{matrix}\right)\right),
\end{equation}
where its entries obey the following equations
\begin{subequations}
\begin{align}
\partial_x p&=2((p_{10}q_{10}-pq)p+(p-p_{10})g+q-q_{10}), \\
\partial_x q_{10}&=2((p_{10}q_{10}-pq)q_{10}+p-p_{10}+(q-q_{10})g), \\
\partial_x g&=2((p_{10}q_{10}-pq)g+(p-p_{10})p+(q-q_{10})q_{10}), \\
\partial_x f&=-2(p_{10}q_{10}-pq)f.
\end{align}
\end{subequations}
Moreover, the above system of differential equations admits two first integrals,\index{first integral(s)} $\partial_x\Phi_i=0$, $i=1,2$, where
\begin{equation}\label{eq.f-g}
  \Phi_1:=f^2(g-pq_{10}) \qquad \text{and}\qquad
  \Phi_2:=f^2(g^2+1-p^2-q_{10}^2).
\end{equation}

In the next section we construct a six-dimensional map from (\ref{Darboux-D2}).

\subsubsection{Dihedral group: A six-dimensional YB map}\index{Yang-Baxter (YB) map(s)}
We consider the matrix $N:=fM$, where $M$ is given by (\ref{Darboux-D2}), and we change $(p,q_{10},f^2)\rightarrow (x_1,x_2,X)$. Then,
\begin{equation} \label{3dLaxD2}
  N(\textbf{x},X;\lambda)=\left(\begin{matrix}
    \lambda^{2}X+x_1x_2X+1 & \lambda x_1X+\lambda^{-1}x_2X\\
    \lambda x_2X+\lambda^{-1}x_1X & \lambda^{-2}X+x_1x_2X+1\end{matrix}\right),
\end{equation}
where we have substituted the product $f^2g$ by
\begin{equation}\label{f-g-X}
f^2g=1+x_1x_2X,
\end{equation}
using the first integral,\index{first integral(s)} $\Phi_1$, in (\ref{eq.f-g}), and having rescaled $c_1\rightarrow 1$.

The Lax equation \index{Lax equation}for the Darboux matrix (\ref{3dLaxD2}) \index{Darboux matrix(-ces)}reads
\begin{equation}
N(\textbf{u};\lambda)N(\textbf{v};\lambda)=N(\textbf{y};\lambda)N(\textbf{x};\lambda),
\end{equation}
from where we obtain an algebraic system of equations, omitted because of its length.

The first six-dimensional component of the corresponding algebraic variety corresponds to the trivial YB map\index{Yang-Baxter (YB) map(s)} 
\begin{equation} 
 \textbf{x}\mapsto \textbf{u}=\textbf{y}, \quad \textbf{y}\mapsto \textbf{v}=\textbf{x}, \quad X\mapsto U=Y, \quad Y\mapsto V=X, \nonumber
\end{equation}
and the second component corresponds to the following map
\begin{eqnarray}\label{f-g-h-D2}
&&x_1 \mapsto u_1=\frac{f(\textbf{x},\textbf{y})}{g(\textbf{x},\textbf{y})},\qquad y_1 \mapsto v_1=x_1, \nonumber \\ 
&&x_2 \mapsto u_2=y_2,\qquad \qquad   y_2 \mapsto v_2=\frac{f(\pi\textbf{y},\pi\textbf{x})}{g(\pi\textbf{y},\pi\textbf{x})}, \\ 
&&~X \mapsto U=\frac{g(\textbf{x},\textbf{y})}{h(\textbf{x},\textbf{y})},\qquad Y \mapsto V=\frac{g(\pi\textbf{y},\pi\textbf{x})}{h(\pi\textbf{y},\pi\textbf{x})}\nonumber,
\end{eqnarray}
where $f,g$ and $h$ are given by
\begin{subequations}
\begin{align}
&f(\textbf{x},\textbf{y})=x_1X+[x_2-y_2+2x_1x_2y_1+x_1^2(y_2-3x_2)]XY+\nonumber\\
&\qquad (y_2^2-1)[y_1(1+x_1^2)-x_1(1+y_1^2)]XY^2-(x_1^2-1)(y_2-x_2)X^2-\nonumber\\
&\qquad (x_1^2-1)[x_2^2(3x_1-y_1)-x_1-y_1+2y_2(y_1y_2-x_1x_2)]X^2Y-\\
&\qquad (x_1^2-1)(y_2^2-1)[y_2(y_1^2-1)+x_2(y_1^2-2x_1y_1+1)]X^2Y^2+\nonumber\\
&\qquad y_1(x_1^2-1)^2(x_2^2-1)(y_2^2-1)X^3Y^2+(x_1^2-1)^2(x_2^2-1)(y_2-x_2)X^3Y+\nonumber\\
&\qquad (y_1-x_1)Y,\nonumber\\ \nonumber \\
&g(\textbf{x},\textbf{y})=X+2y_2(y_1-x_1)XY+(y_2^2-1)(x_1-y_1)^2XY^2+\nonumber\\
&\qquad 2(x_1^2-1)(1-x_2y_2)X^2Y+2x_2(x_1^2-1)(y_2^2-1)(x_1-y_1)X^2Y^2+\\
&\qquad (x_1^2-1)^2(x_2^2-1)(y_2^2-1)X^3Y^2\nonumber\\ \nonumber\displaybreak[0]
\intertext{and}
&h(\textbf{x},\textbf{y})=1-2x_1(y_2-x_2)X-2(x_1y_1-1)(y_2^2-1)XY+\nonumber\\
&\qquad (x_1^2-1)(x_2-y_2)^2X^2-2y_1(x_2-y_2)(x_1^2-1)(y_2^2-1)X^2Y+\\
&\qquad (x_1^2-1)(y_1^2-1)(y_2^2-1)^2X^2Y^2.\nonumber
\end{align}
\end{subequations}

It can be verified that this is a YB map.\index{Yang-Baxter (YB) map(s)!parametric} From $\tr(N(\textbf{x},X;\lambda)N(\textbf{y},Y;\lambda))$ we extract the following invariants for the above map
\begin{subequations}
\begin{align}
\hspace{-1.2cm} I_1(\textbf{x},\textbf{y})&=XY,  \\
\hspace{-1.2cm} I_2(\textbf{x},\textbf{y})&=X+Y+(x_1+y_1)(x_2+y_2)XY,  \\
\hspace{-1.2cm} I_3(\textbf{x},\textbf{y})&=2x_1x_2X+2y_1y_2Y+2(\textbf{x}\cdot\textbf{y}+x_1x_2y_1y_2)XY+2. 
\end{align}
\end{subequations}

\subsection{Dihedral group: A linearised YB map}\index{Yang-Baxter (YB) map(s)!linearised}
In the cases of the NLS and DNLS equations we were able to derive six-dimensional maps and, using their invariants, reduce them to four-dimensional YB maps. The matrix refactorisation (\ref{eqLax}) for matrix (\ref{Darboux-D2}) is not unique and it is difficult to deduce a six-dimensional YB map in this case. Therefore, even though we can reuniformise equations (\ref{eq.f-g}), we can not use them to derive a four-dimensional map. Yet, we can find a linear approximation to the former.

In particular, let us replace $(f q_{10},f p)\rightarrow (\epsilon x_1,\epsilon x_2)$ in the Darboux matrix (\ref{Darboux-D2}) \index{Darboux matrix(-ces)}and linearise the corresponding map around $\epsilon=0$.

It follows from $\Phi_1=\frac{1-k^2}{4}$ and $\Phi_2=\frac{1+k^2}{2}$ that the quantities $f$ and $fg$ are given by
$$f=\frac{1+k}{2}+{\cal{O}}(\epsilon) \qquad \text{and} \qquad f g=\frac{1-k}{2}+{\cal{O}}(\epsilon),$$
and, therefore, the Lax matrix\index{Lax matrix(-ces)} by
\begin{equation*}
  M=\frac{1+k}{2}\left(\begin{matrix} 
    \lambda^2  & 0\\
    0 & \lambda^{-2} \end{matrix}\right)+\left(\begin{matrix}
    0 & \lambda x_1+\lambda^{-1}x_2\\
    \lambda x_2+\lambda^{-1}x_1 & 0\end{matrix}\right)+
    \frac{1-k}{2}\left(\begin{matrix}
    1 & 0\\
    0 & 1\end{matrix}\right)+{\cal{O}}(\epsilon).
\end{equation*}

The linear approximation to the YB map\index{Yang-Baxter (YB) map(s)!linear approximation} is given by
\begin{equation}
\left(\begin{matrix} x_1 \\ x_2 \\ y_1\\ y_2 \end{matrix}\right) \overset{U_0}{\longrightarrow}
\left(\begin{matrix} u_1 \\ u_2 \\ v_1\\ v_2 \end{matrix}\right)=
\left(\begin{array}{c c|c c}
\frac{(a-1)(a-b)}{(a+1)(a+b)} & \frac{a-b}{a+b} & \frac{2 a}{a+b} & \frac{(a+1)(b-a)}{(b+1)(a+b)} \\
0 & 0 & 0 & \frac{a+1}{b+1}\\
\hline
\frac{b+1}{a+1} & 0 & 0 & 0\\
\frac{(a-b)(b+1)}{(a+1)(a+b)} & \frac{2b}{a+b} & \frac{b-a}{a+b} & \frac{(b-1)(b-a)}{(b+1)(a+b)}\end{array}\right)
\left(\begin{matrix} x_1 \\ x_2 \\ y_1\\ y_2 \end{matrix}\right)
\label{linear_YB_D2}
\end{equation}
which is a linear parametric YB map \index{Yang-Baxter (YB) map(s)!linear} and it is not involutive.


\section{$2N\times 2N$-dimensional YB maps}\index{Yang-Baxter (YB) map(s)}
In this section, we consider the vector generalisations of the YB maps\index{Yang-Baxter (YB) map(s)!generalisations} (\ref{YB_NLS}) and (\ref{YB-affine}). We replace the variables, $x_1$ and $x_2$, in the Lax matrices\index{Lax matrix(-ces)} with $N$-vectors $\textbf{w}_1$ and $\textbf{w}_2^T$ to obtain $2N \times 2N$ YB maps. 

In what follows we use the following notation for a $n$-vector $\textbf{w}=(w_1,...,w_n)$
\begin{equation}
\textbf{w}=(\textbf{w}_1,\textbf{w}_2),\qquad \text{where} \qquad \textbf{w}_1=(w_1,...,w_N), \qquad \textbf{w}_2=(w_{N+1},...,w_{2N})
\end{equation}
and also
\begin{equation}
\langle u_i|:=\textbf{u}_i,\qquad |w_i\rangle:=\textbf{w}_i^T \qquad \text{and their dot product with}\qquad \langle u_i,w_i\rangle.
\end{equation}

\subsection{NLS equation}
Replacing the variables in (\ref{laxNLS}) with $N$-vectors, namely
\begin{equation}
M(\textbf{w};a,\lambda)=\left(     
\begin{matrix}
\lambda+a+\langle w_1,w_2\rangle  & \langle w_1| \\
|w_2\rangle      & I             
\end{matrix}\right),
\label{laxNLSnN}
\end{equation}
we obtain a unique solution of the Lax equation\index{Lax equation} given by the following $2N\times 2N$ map
\begin{subequations}
\begin{equation}
\begin{cases}
\langle u_1|=\langle y_1|+f(z;a,b)\langle x_1|,\\
\langle u_2|=\langle y_2|,
\end{cases}
\end{equation}
and
\begin{equation}
\begin{cases}
\langle v_1|=\langle x_1|,\\
\langle v_2|=\langle x_2|+f(z;b,a)\langle y_2|,
\end{cases}
\end{equation}
where $f$ is given by
\begin{equation}
f(z;b,a)=\frac{b-a}{1+z},\qquad z:=\langle x_1,y_2\rangle.
\end{equation}
\end{subequations}
The above is a non-involutive parametric $2N\times 2N$ YB map\index{Yang-Baxter (YB) map(s)} with strong Lax matrix\index{Lax matrix(-ces)!strong} given by (\ref{laxNLSnN}). As a YB map it appears in \cite{PT}, but it is originally introduced by Adler \cite{AdlerNLSVec}. Moreover, one can construct the above $2N\times 2N$ map for the $N\times N$ Darboux matrix (\ref{laxNLSnN}) \index{Darboux matrix(-ces)} by taking the limit of the solution of the refactorisation problem\index{refactorisation problem(s)} in \cite{Kouloukas2}.

Two invariants of this map are given by
\begin{subequations}
\begin{align}
I_1(\textbf{x},\textbf{y};a,b)&=\langle x_1,x_2\rangle+\langle y_1,y_2\rangle,\\ 
I_2(\textbf{x},\textbf{y};a,b)&=b\langle x_1,x_2\rangle+a\langle y_1,y_2\rangle+\langle x_1,y_2\rangle+\langle x_2,y_1\rangle+\langle x_1,x_2\rangle \langle y_1,y_2\rangle.
\end{align}
\end{subequations}
These are the invariants which are obtained from the trace of $M(\textbf{y};b,\lambda)M(\textbf{x};a,\lambda)$ and they are not enough to claim Liouville integrability.\index{Liouville integrability}

\subsection{$\field{Z}_2$ reduction}
In the case of $Z_2$ we consider, instead of (\ref{Darboux_Affine}), the following matrix
\begin{equation}
M(\textbf{w};a,\lambda)=\left(     
\begin{matrix}
\lambda^2(a+\langle w_1,w_2\rangle)  & \lambda \langle w_1|\\
\lambda |w_2\rangle      & I           
\end{matrix}\right),
\label{laxAffinenN}
\end{equation}
we obtain a unique solution for the Lax equation\index{Lax equation} given by the following $2N\times 2N$ map
\begin{subequations}
\begin{equation}
\label{vectZ2_1}
\begin{cases}
\langle u_1|=\langle y_1|+f(z;a,b)\langle x_1|, \\
\langle u_2|=g(z;a,b)\langle y_2|,
\end{cases}
\end{equation}
and
\begin{equation}
\begin{cases}
\label{vectZ2_2}
\langle v_1|=g(z;b,a)\langle x_1|,\\
\langle v_2|=\langle x_2|+f(z;b,a)\langle y_2|,
\end{cases}
\end{equation}
where $f$ and $g$ are given by
\begin{equation}\label{vectZ2_3}
f(z;a,b)=\frac{a-b}{a-z},\qquad g(z;a,b)=\frac{a-z}{b-z},\qquad z:=\langle x_1,y_2\rangle.
\end{equation}
\end{subequations}

The above map is a non-involutive parametric $2N\times 2N$ YB map\index{Yang-Baxter (YB) map(s)} with strong Lax matrix\index{Lax matrix(-ces)!strong} given by (\ref{laxAffinenN}).

It follows from the trace of $M(\textbf{y};b,\lambda)M(\textbf{x};a,\lambda)$, that two invariants for (\ref{vectZ2_1})-(\ref{vectZ2_2}) are given by 
\begin{subequations}
\begin{align}
I_1(\textbf{x},\textbf{y};a,b)&=b\langle x_1,x_2\rangle+a\langle y_1,y_2\rangle+\langle x_1,x_2\rangle \langle y_1,y_2\rangle,\\ I_2(\textbf{x},\textbf{y};a,b)&=\langle x_1+y_1,x_2+y_2\rangle.
\end{align}
\end{subequations}
In fact, both vectors of the inner product in $I_2$ are invariants. 

However, as in the case of the NLS vector generalisation, the invariants are not enough to claim Liouville integrability.\index{Liouville integrability}

\chapter{Extensions on Grassmann algebras}
\label{chap6} \setcounter{equation}{0}
\renewcommand{\theequation}{\thechapter.\arabic{equation}}

\section{Overview}
Noncommutative extensions of integrable equations\index{noncommutative extension(s)!of integrable equations} have been of great interest since late seventies; for instance, in \cite{Chaichian} the supersymmetric Liouville and sine-Gordon equations were studied, while in \cite{Kupershmidt} hierarchies of the KdV equation were associated to super Lie algebras. Other examples include the NLS, the DNLS and KP equations etc. Indicatively, we refer to \cite{Antonowicz, Dimakis2005, Dimakis2010, Olver1998}.

In this chapter, motivated by \cite{Georgi}, we are interested in the Grassmann extensions of some of the Darboux matrices\index{Grassmann extention(s)!of Darboux matrix(-ces)} and their associated YB maps presented in the previous chapters. Specifically, we shall present noncommutative extensions of the Darboux matrices in the cases of the NLS equation and the DNLS equation, together with the noncommutative extensions of the associated YB maps. In fact, we will use the Darboux matrices to construct ten-dimensional YB maps which can be restricted to eight-dimensional YB maps on invariant leaves.

The chapter is organised as follows: In the next section we present the basic facts and properties of Grassmann algebras \index{Grassmann algebra(s)} which we will need in the following sections. Section 3 deals with the noncommutative extensions of the Darboux matrices of NLS type equations and, particularly, the Grassmann extension of the Darboux matrix \index{Grassmann extention(s)!of Darboux matrix(-ces)} in the cases of the NLS equation \cite{Georgi} and the DNLS equation. Section 4 is devoted to the construction of ten-dimensional YB maps and their eight-dimensional restrictions on invariant leaves corresponding to first integrals\index{first integral(s)} of the former. These YB maps constitute noncommutative extensions of the YB maps presented in chapter 5. Finally, in section 5 we present vector generalisations of these YB maps.

\section{Elements of Grassmann algebras}
In this section, we briefly present the basic properties of Grassmann algebras\index{Grassmann algebra(s)} that we will need in the rest of this chapter. However, one could consult \cite{Berezin} for further details.

Let $G$ be a $\field{Z}_2$-graded algebra over $\field{C}$ or, in general, a field $K$ of characteristic zero. Thus, $G$ as a linear space is a direct sum $G=G_0\oplus G_1$ (mod 2), such that $G_iG_j\subseteq G_{i+j}$. Those elements of $G$ that belong either to $G_0$ or to $G_1$ are called \textit{homogeneous}, the ones from $G_0$ are called \textit{even} (Bosonic),\index{even (Bosonic) element(s)} while those in $G_1$ are called \textit{odd} (Fermionic).\index{odd (Fermionic) element(s)}

By definition, the parity $|a|$ of an even homogeneous element\index{even (Bosonic) element(s)} $a$ is $0$ and it is $1$ for odd homogeneous elements.\index{odd (Fermionic) element(s)} The parity of the product $|ab|$ of two homogeneous elements is a sum of their parities: $|ab|=|a|+|b|$. Grassmann commutativity \index{Grassmann commutativity} means that $ba=(-1)^{|a||b|}ab$ for any homogeneous elements $a$ and $b$. In particular, $a_1^2=0$, for all $a_1\in G_1$ and even elements\index{even (Bosonic) element(s)} commute with all the elements of $G$.

\begin{remark}\normalfont
In the rest of this chapter we shall be using Latin letters for (commutative) even elements\index{even (Bosonic) element(s)} of Lax operators or entries of Darboux matrices\index{Darboux matrix(-ces)}, and Greek letters when referring to the (noncommutative) odd ones (even though Greeks are not odd!). Moreover, for the sake of consistency with the rest of the thesis we shall continue using the greek letter $\lambda$ when referring to the spectral parameter, despite the fact that $\lambda$ is a commutative element.
\end{remark}

\subsection{Supertrace and superdeterminant}
Let $M$ be a square matrix of the following form
\begin{equation}\label{G-matrixForm}
M=\left(
\begin{matrix}
 P & \Pi \\
 \Lambda & L
\end{matrix}\right),
\end{equation}
where $P$ and $L$ are square matrices of even variables,\index{even variable(s)} whereas $\Pi$ and $\Lambda$ are matrices of odd variables,\index{odd variable(s)} not necessarily square.

We define the \textit{supertrace} \index{supertrace} of $M$ --and we will denote it by $\str(M)$-- to be the following quantity 
\begin{equation}
\str(M)=\tr (P)-\tr (L),
\end{equation}
where $\tr(.)$ is the usual trace of a matrix.

Moreover, we define the \textit{superdeterminant}\index{superdeterminant} of $M$ --and we will denote it by $\sdet(M)$-- to be 
\begin{subequations}
\begin{eqnarray}
\sdet(M)&=&\det(P-\Pi L^{-1}\Lambda)\det(L^{-1})=\\
        &=&\det(P^{-1})\det(L-\Lambda P^{-1}\Pi),  
\end{eqnarray}
\end{subequations}
where $\det(.)$ is the usual determinant of a matrix.

The simplest properties of the supertrace\index{supertrace} and the superdeterminant, for two matrices $A$ and $B$ of the form (\ref{G-matrixForm}), are the following
\begin{enumerate}
	\item \str(AB)=\str(BA),
	\item \sdet(AB)=\sdet(A)\sdet(B).
\end{enumerate}

\subsection{Differentiation rule for odd variables}\index{odd variable(s)!differentiation rule}
The (left) derivative of a product of odd elements,\index{odd (Fermionic) element(s)} say $\alpha_{i_1},\ldots,\alpha_{i_k}$, obeys the following rule
\begin{eqnarray}
\frac{\partial}{\partial \alpha_i}(\alpha_{i_1}\cdot\ldots\cdot\alpha_{i_k})=&&\delta_{ii_1}\alpha_{i_2}\alpha_{i_3}\cdot\ldots\cdot\alpha_{i_k}-\nonumber\\
&&\delta_{ii_2}\alpha_{i_1}\alpha_{i_3}\cdot\ldots\cdot\alpha_{i_k}+\delta_{ii_3}\alpha_{i_1}\alpha_{i_2}\cdot\ldots\cdot\alpha_{i_k}-\ldots
\end{eqnarray}
where $\delta_{ij}$ is the Kronecker operator. For example, if $\alpha$ and $\beta$ are odd variables,\index{odd variable(s)} then
\begin{equation}
\frac{\partial}{\partial\alpha}(\alpha\beta)=\beta,\quad \text{but} \quad \frac{\partial}{\partial\beta}(\alpha\beta)=-\alpha.
\end{equation}


\subsection{Properties of the Lax equation}
Let $L=L(x,\chi;a)$ be a Lax matrix where $x$ is an even variable,\index{even variable(s)} $\chi$ is an odd variable\index{odd variable(s)} and $a$ a parameter. Since the Lax equation, 
\begin{equation}\label{Laxeq-G}
L(u,\xi;a)L(v,\eta;b)=L(y,\psi;b)L(x,\chi;a)
\end{equation}
has the obvious symmetry 
\begin{equation}\label{symL-G}
(u,\,\xi,\,v,\,\eta,\,a,\,b) \longleftrightarrow  (y,\,\psi,\,x ,\,\chi,\,b,\,a)
\end{equation}
we have the following
\begin{proposition} \label{rationality-G}
If a matrix refactorisation problem (\ref{Laxeq-G})\index{refactorisation problem(s)} yields a rational map $(x,\chi,y,\psi)=Y_{a,b}(u,\xi,v,\eta)$, then this map is birational.
\end{proposition}

\begin{proof} Let $Y:(x,\chi,y,\psi)\mapsto (u,\xi,v,\eta)$ be a rational map corresponding to a refactorisation problem\index{refactorisation problem(s)} (\ref{eqLax}), i.e.
\begin{subequations}
\begin{align}
x\mapsto u&=\frac{n_1(x,\chi,y,\psi;a,b)}{d_1(x,\chi,y,\psi;a,b)}, \qquad y\mapsto v=\frac{n_2(x,\chi,y,\psi;a,b)}{d_2(x,\chi,y,\psi;a,b)},\\
\chi\mapsto \xi &=\frac{n_3(x,\chi,y,\psi;a,b)}{d_3(x,\chi,y,\psi;a,b)}, \qquad \psi\mapsto \eta=\frac{n_4(x,\chi,y,\psi;a,b)}{d_4(x,\chi,y,\psi;a,b)},
\end{align}
\end{subequations}
where $n_i$, $d_i$, $i=1,2,3,4$, are polynomial functions of their variables.

Due to the symmetry (\ref{symL-G}) of the refactorisation problem (\ref{Laxeq-G}),\index{refactorisation problem(s)} the inverse map of $Y$, $Y^{-1}:(x,\chi,y,\psi)\mapsto (u,\xi,v,\eta)$, is also rational and, in fact,
\begin{subequations}
\begin{align}
u\mapsto x&=\frac{n_1(v,\eta,u,\xi;b,a)}{d_1(v,\eta,u,\xi;b,a)}, \qquad v\mapsto y=\frac{n_2(v,\eta,u,\xi;b,a)}{d_2(v,\eta,u,\xi;b,a)},\\
\xi\mapsto \chi&=\frac{n_3(v,\eta,u,\xi;b,a)}{d_3(v,\eta,u,\xi;b,a)}, \qquad \eta\mapsto \psi=\frac{n_4(v,\eta,u,\xi;b,a)}{d_4(v,\eta,u,\xi;b,a)}.
\end{align}
\end{subequations}
Therefore, $Y$ is a birational map.
\end{proof}

\begin{remark}\normalfont
Functions $d_i(x,\chi,y,\psi;a,b)$, $i=1,2,3,4$, must depend on the odd variables\index{odd variable(s)} in a way such that their expressions are even. For example, the expression $xy+\chi\psi$ is even.
\end{remark}

\begin{proposition}\label{genInvGras}
If $L=L(x,\chi,a;\lambda)$ is a Lax matrix with corresponding YB map, $Y:(x,\chi,y,\psi)\mapsto (u,\xi,v,\eta)$, then $\str(L(y,\psi,b;\lambda)L(x,\chi,a;\lambda))$ is a generating function of invariants of the YB map.
\end{proposition}

\begin{proof}
Since, 
\begin{eqnarray} 
\str(L(u,\xi,a;\lambda)L(v,\eta,b;\lambda))&\overset{(\ref{eqLax})}{=}&\str(L(y,\psi,b;\lambda)L(x,\chi,a;\lambda))\nonumber\\
&=&\str(L(x,\chi,a;\lambda)L(y,\psi,b;\lambda)),\label{strace}
\end{eqnarray}
and function $\str(L(x,\chi,a;\lambda)L(y,\psi,b;\lambda))$ can be written as $\str(L(x,\chi,a;\lambda)L(y,\psi,b;\lambda))$ $\displaystyle=\sum_k \lambda^k I_k(x,\chi,y,\psi;a,b)$, from (\ref{strace}) follows that
\begin{equation}
I_i(u,\xi,v,\eta;a,b)=I_i(x,\chi,y,\psi;a,b),
\end{equation}
which are invariants for $Y$.
\end{proof}

\begin{remark} \normalfont 
The invariants of a YB map, $I_i(x,\chi,y,\psi;a,b)$, may not be functionally independent.
\end{remark}

\section{Extensions of Darboux transformations on Grass-\\mann algebras}\index{Grassmann algebra(s)}
In this section we consider the Grassmann extensions of the Darboux matrices\index{Grassmann extention(s)!of Darboux matrix(-ces)} corresponding to the NLS equation and the DNLS equation. In particular, we present the noncommutative extension of Darboux matrices \index{noncommutative extension(s)!of Darboux matrix(-ces)} (\ref{M-NLS}) (see \cite{Georgi}) and (\ref{DT-sl2-gen}) (see \cite{GSS}).


\subsection{Nonlinear Schr\"odinger equation}
The Grassmann extension of the Darboux matrix (\ref{M-NLS})\index{Grassmann extention(s)!of Darboux matrix(-ces)} was constructed in \cite{Georgi}. We present this Darboux matrix together with its associated Lax operator and we will use it in the next section to construct a Grassmann extension of the Adler-Yamilov map.\index{Grassmann extention(s)!of Adler-Yamilov map}  

Specifically, let us consider a more general Lax operator than (\ref{NLS-U}), namely the following noncommutative extension of the NLS operator \index{noncommutative extension(s)!of NLS operator}
\begin{subequations}\label{G-NLS-U}
\begin{align}
&\mathcal{L} := D_x + U(p,q,\psi,\phi,\zeta,\kappa;\lambda)=D_x +\lambda U^{1}+U^{0},
\intertext{where $U^{1}$ and $U^0$ are given by}
&U^1={\rm{diag}}(1,-1,0),\quad U^0=\left(\begin{array}{ccc} 0 & 2p & \theta \\ 2q & 0 & \zeta \\ \phi & \kappa & 0\end{array}\right),
\end{align}
\end{subequations}
where $p,q\in G_0$ and $\psi,\phi,\zeta,\kappa \in G_1$. Note that, if we set the odd variables\index{odd variable(s)} equal to zero, operator (\ref{G-NLS-U}) coincides with (\ref{NLS-U}).

A linear in the spectral parameter Darboux matrix for (\ref{NLS-U}) is given by the following \cite{Georgi}.

\begin{proposition}\label{DarbouxNLS-Prop}
Let $M=\lambda M_1+M_0$ be a Darboux matrix for (\ref{G-NLS-U}). Moreover, let $M$ define a Darboux transformation of rank 1. Then, up to a gauge transformation\index{gauge transformation}, $M$ is of the following form
\begin{equation}\label{DarbouxNLS}
M(p,q,\theta,\phi;c_1,c_2)=\left(
\begin{matrix}
 F+\lambda & p & \theta\\
 q_{10} & c_1 & 0\\
 \phi_{10} & 0 & c_2
\end{matrix}\right),
\end{equation}
where $c_1$ and $c_2$ can be either 1 or 0. In the case where $c_1=c_2=1$, the entries of $M(p,q,\theta,\phi;1,1)$ satisfy the following system of differential-difference equations \index{differential-difference equation(s)}
\begin{subequations}\label{DarbouxNLS-sys}
\begin{align}
F_x&=2(pq-p_{10}q_{10})+\theta \phi-\theta_{10}\phi_{10},\\
p_x&=2(Fp-p_{10})+\theta \zeta,\\
q_{10,x}&=2(q-q_{10}F)-\kappa_{10} \phi_{10},\\
\theta_x&=F\theta-\theta_{10}+p\kappa,\\
\phi_{10,x}&=\phi-\phi_{10}F-\zeta_{10}q_{10},
\end{align}
\end{subequations}
and the algebraic equations
\begin{subequations}
\begin{align}
\theta q_{10}&=(\mathcal{S}-1)\kappa,\\
\phi_{10}p&=(\mathcal{S}-1)\zeta.
\end{align}
\end{subequations}
\end{proposition}

\begin{proof}
Substitution of $M$ to equation (\ref{DMeq}) implies a second order algebraic equation in $\lambda$. Equating the coefficients of the several powers of $\lambda$ equal to zero, we obtain the following system of equations
\begin{subequations}
\begin{align}
\lambda^2: & \quad \left[U^1,M_1\right]=0,\label{G-M-syst-1}\\
\lambda^1:   & \quad M_{1,x}+\left[U^1,M_0\right]+U^{0}_{10}M_1-M_1U^0=0,\label{G-M-syst-2}\\
\lambda^0: & \quad M_{0,x}+U^0_{10}M_0-M_0U^0=0.\label{G-M-syst-3}
\end{align}
\end{subequations}
The first equation, (\ref{G-M-syst-1}), implies that $M_1$ must be diagonal, say $M_1=\diag(\alpha,\beta,\gamma)$. Then, from the diagonal part of (\ref{G-M-syst-2}) we deduce that $\alpha_x=\beta_x=\gamma_x=0$. Since $\rank(M_1)=1$, only one of $\alpha$, $\beta$ and $\gamma$ can be nonzero. Without any loss of generality, we choose $\alpha=1$ and $\beta=\gamma=0$.

Now, the off-diagonal part of (\ref{G-M-syst-2}) implies $M_{0,12}=p$, $M_{0,13}=\theta$, $M_{0,21}=q_{10}$, $M_{0,31}=\phi_{10}$ and $M_{0,32}=M_{0,23}=0$. We call the $M_{0,11}$ entry $M_{0,11}=F$.

Finally, from equation (\ref{G-M-syst-3}) we obtain $(M_{0,22})_x=(M_{0,33})_x=0$, namely $M_{0,22}=c_1$, $M_{0,22}=c_2$, together with equations (\ref{DarbouxNLS-sys}).
\end{proof}

\begin{note}\normalfont
At this point, it is worth mentioning that the superdeterminant of matrix $M(p,q,\theta,\phi;1,1)$ in (\ref{DarbouxNLS}) implies the following
\begin{equation}\label{1integralNLS}
\partial_x(F-pq_{10}-\phi_{10}\theta)=0,
\end{equation}
since $\partial_x(\sdet (M))=0$.
Moreover, (\ref{1integralNLS}) is a first integral\index{first integral(s)} for system (\ref{DarbouxNLS-sys}) for $c_1=c_2=1$.
\end{note}

\begin{remark}\normalfont
If one sets the odd variables\index{odd variable(s)} equal to zero, matrix (\ref{DarbouxNLS}) the results of Proposition \ref{DarbouxNLS-Prop} agree with those in Proposition \ref{M-NLS-Prop} presented in chapter 3.
\end{remark}


\subsection{Derivative nonlinear Schr\"odinger equation}
Let us now consider a more general than (\ref{SL2-Lax-Op}), namely the following noncommutative extension of the DNLS operator \index{noncommutative extension(s)!of DNLS operator}
\begin{subequations} \label{G-SL2-Lax-Op} 
\begin{align}
&\mathcal{L}=D_x+\lambda^{2} U^2+\lambda U^1,
\intertext{where}
U^2=&\diag(1,-1,-1) \quad \text{and} \quad U^1=\left(\begin{matrix}
 0 & 2p & 2\theta\\
 2q & 0 & 0\\
 2\phi & 0 & 0
\end{matrix}\right).
\end{align}
\end{subequations}
Operator (\ref{G-SL2-Lax-Op}) is invariant under the transformation
\begin{equation}\label{reduction-Z2-G}
s_1(\lambda): \mathcal{L}(\lambda) \rightarrow \mathcal{L}(-\lambda)=\mathfrak{s}_{3}\mathcal{L}(\lambda) \mathfrak{s}_{3}, 
\end{equation}
where $\mathfrak{s}_3=\diag(1,-1,-1)$, $\mathfrak{s}_3^2=1$.

We are seeking a Darboux matrix \index{Darboux matrix(-ces)} for (\ref{G-SL2-Lax-Op}) with square dependence in the spectral parameter, namely of the form
\begin{equation}\label{sqformM}
M=\lambda^2 M_2+\lambda M_1+M_0,
\end{equation}
where $M_i$, $i=0,1,2$, is a $3\times 3$ matrix.

\begin{lemma}
Let $M$ be a second order matrix polynomial in $\lambda$ of the form (\ref{sqformM}). Then, $M$ is invariant under the involution $s_1(\lambda)$ iff 
\begin{subequations}\label{form-M-Z2-G}
\begin{align}
M_{i,12}&=M_{i,13}=M_{i,21}=M_{i,31}=0,\quad i=0,2,\quad\text{and}\\
M_{1,11}&=M_{1,22}=M_{1,33}=M_{1,23}=M_{1,32}=0.
\end{align}
\end{subequations}
\end{lemma}
\begin{proof}
It can be readily proven from $M(-\lambda)=\mathfrak{s}_3M(\lambda)\mathfrak{s}_3$.
\end{proof}

We restrict ourselves to the case where $M_2$ in (\ref{sqformM}) has rank one.

\begin{proposition}\label{G-Darboux-DNLScase}
Let $M$ be a Darboux matrix for (\ref{G-SL2-Lax-Op}) of the form (\ref{sqformM}), with $\rank{M_2}=1$, and suppose that it is invariant under the involution (\ref{reduction-Z2-G}). Then, up to a gauge transformation\index{gauge transformation}, $M$ is given by
\small
\begin{equation}\label{DarbouxZ2-M}
M(f,p,q_{10},\theta,\phi_{10};c_1,c_2)=\lambda^2 \left(
\begin{matrix}
 f & 0 & 0\\
 0 & 0 & 0\\
 0 & 0 & 0
\end{matrix}\right)+\lambda \left(
\begin{matrix}
 0 & f p & f \theta \\
 q_{10} f & 0 &0 \\
 \phi_{10} f & 0 & 0
\end{matrix}\right)
+\left(
\begin{matrix}
 c_1 & 0 & 0 \\
 0 & 1 &0 \\
 0 & 0 & c_2
\end{matrix}\right),
\end{equation}
\normalsize
where its entries satisfy the following differential-difference equations \index{differential-difference equation(s)}
\begin{subequations}\label{dif-difEqs}
\begin{align}
f_x&=2f(pq-p_{10}q_{10}\theta\phi-\theta_{10}\phi_{10}),\label{dif-difEqs-1}\\
p_x&=2p(p_{10}q_{10}-pq+\theta_{10}\phi_{10}-\theta\phi)-2\frac{c_2p_{10}-c_1p}{f},\label{dif-difEqs-2}\\
q_{10,x}&=2q_{10}(p_{10}q_{10}-pq+\theta_{10}\phi_{10}-\theta\phi)-2\frac{c_1q_{10}-c_2q}{f},\\
\theta_x&=2\phi(p_{10}q_{10}-pq+\theta_{10}\phi_{10}-\theta\phi)+\frac{c_1\theta-c_2\theta_{10}}{f},\\
\phi_{10,x}&=2\phi_{10}(p_{10}q_{10}-pq+\theta_{10}\phi_{10}-\theta\phi)+\frac{c_2\phi-c_1\phi_{10}}{f}.\label{dif-difEqs-5}
\end{align}
\end{subequations}
\end{proposition}

\begin{proof}
First of all, for the entries of matrices $M_i$, $i=1,2,3$ we have (\ref{form-M-Z2-G}). Then, substitution of $M$ to equation (\ref{DMeq}) implies a second order algebraic equation in $\lambda$. Equating the coefficients of the several powers of $\lambda$ equal to zero, we obtain the following system of equations
\begin{subequations}\label{Dar-eq}
\begin{align}
\left[U^2,M_2\right]&=0 \label{Dar-eq1} \\
\left[U^2,M_1\right]+U_{10}^1M_2-M_2U^1&=0 \label{Dar-eq2}\\
M_{2,x}+\left[U^2,M_0\right]+U_{10}^1M_1-M_1U^1&=0 \label{Dar-eq3}\\
M_{1,x}+U_{10}^1M_0-M_0U^1&=0 \label{Dar-eq4}\\
M_{0,x}&=0.\label{Dar-eq5}
\end{align}
\end{subequations}

From (\ref{Dar-eq5}) follows that the matrix $M_0$ must be constant, whereas equation (\ref{Dar-eq1}) implies that $M_2$ is diagonal. Since $\rank{M_2=1}$ we can choose $M_2=\diag\{f,0,0\}$ without loss of generality; the cases $M_2=\diag\{0,g,0\}$ and $M_2=\diag\{0,0,h\}$ lead to gauge equivalent Darboux matrices. In this case, from equation (\ref{Dar-eq2}) we have that the entries of $M_1$ are given by
\begin{equation}
M_{1,12}=fp,\quad M_{1,13}=f\theta,\quad M_{1,21}=q_{10}f \quad\text{and}\quad M_{1,31}=\phi_{10}f.
\end{equation}

Now, from equation (\ref{Dar-eq3}) we deduce equation (\ref{dif-difEqs-1}) and that $M_0$ must be diagonal, namely of the form $M_0=\diag(c_1,1,c_2)$ (one of the parameters along its diagonal can be rescaled to 1). Therefore, matrix $M$ is of the form (\ref{DarbouxZ2-M}).

Finally, equation (\ref{Dar-eq4}) implies system (\ref{dif-difEqs-2})-(\ref{dif-difEqs-5}) (where we have made use of (\ref{dif-difEqs-1})).
\end{proof}

\begin{note}\normalfont
Entry $f$ is permitted to appear in the denominator in (\ref{dif-difEqs}), as it is an even variable.\index{even variable(s)} This is due to the superdeterminant of matrix $M$. In fact, the  constant determinant property of matrix $M(f,p,q_{10},\theta,\phi_{10};1,1)$ implies the following equation
\begin{equation}\label{G-z2-1stInt}
\partial_x(f-f^2(pq_{10}+\theta\phi_{10}))=0,
\end{equation}
which makes it quite obvious that $f$ must be even. Moreover, (\ref{G-z2-1stInt}) is a first integral\index{first integral(s)} for the system of differential-difference equations \index{differential-difference equation(s)} (\ref{dif-difEqs}).
\end{note}

\begin{remark}\normalfont The results of Proposition \ref{G-Darboux-DNLScase} agree with those of Proposition \ref{Darboux-DNLScase} if one sets the odd variables\index{odd variable(s)} equal to zero.
\end{remark}


\section{Grassmann extensions of Yang-Baxter maps}\index{Grassmann extention(s)!of Yang-Baxter maps}
Here we employ the Darboux matrices presented in the previous section to construct ten-dimensional YB maps, which can be restricted to eight-dimensional YB maps on invariant leaves. We start with the case of NLS equation.


\subsection{Nonlinear Schr\"odinger equation} 
According to (\ref{DarbouxNLS}) we define the following matrix
\begin{equation}\label{mLaxNLS}
M(\textbf{x};\lambda)=\left(
\begin{matrix}
 X+\lambda & x_1 & \chi_1\\
 x_2 & 1 &0 \\
 \chi_2 & 0 & 1
\end{matrix}\right), \qquad \textbf{x}:=(x_1,x_2,\chi_1,\chi_2,X),
\end{equation}
and substitute to the Lax equation.

The corresponding algebraic variety is a union of two ten-dimensional components. The first one is obvious from the refactorisation problem,\index{refactorisation problem(s)} and it corresponds to the permutation map
\begin{equation}
 \textbf{x}\mapsto \textbf{u}=\textbf{y}, \qquad \textbf{y}\mapsto \textbf{v}=\textbf{x}, \nonumber
\end{equation}
which is a trivial YB map. The second one can be represented as a ten-dimensional non-involutive Yang-Baxter map given by
{\allowdisplaybreaks
\begin{subequations}\label{NLS5D}
\begin{align} 
 x_1\, \mapsto \, u_1&=y_1-\frac{X-x_1x_2-\chi_1\chi_2-Y+y_1y_2+\psi_1\psi_2}{1+x_1y_2+\chi_1\psi_2} x_1,\\
 x_2\, \mapsto \, u_2&=y_2,\\
 \chi_1\, \mapsto \, \xi_1&=\psi_1-\frac{X-x_1x_2-Y+y_1y_2+\psi_1\psi_2}{1+x_1y_2} \chi_1,\\
 \chi_2\, \mapsto \, \xi_2&=\psi_2,\\
 X\, \mapsto \, U&=\frac{X-x_1x_2-\chi_1\chi_2+(x_1y_2+\chi_1\psi_2)Y+y_1y_2+\psi_1\psi_2}{1+x_1y_2+\chi_1\psi_2}, \\
 y_1\, \mapsto \, v_1&=x_1,\\
 y_2\, \mapsto \, v_2&=x_2+\frac{X-x_1x_2-\chi_1\chi_2-Y+y_1y_2+\psi_1\psi_2}{1+x_1y_2+\chi_1\psi_2} y_2,\\
 \psi_1\, \mapsto \, \eta_1&=\chi_1,\\
 \psi_2\, \mapsto \, \eta_2&=\chi_2+\frac{X-x_1x_2-\chi_1\chi_2-Y+y_1y_2}{1+x_1y_2} \psi_2,\\
 Y\, \mapsto \, V&=\frac{(x_1y_2+\chi_1\psi_2)X+x_1x_2+\chi_1\chi_2+Y-y_1y_2-\psi_1\psi_2}{1+x_1y_2+\chi_1\psi_2}.
\end{align}
\end{subequations}
}

Finally, map (\ref{NLS5D}) is birational due to Proposition \ref{rationality-G}.


\subsubsection{Restriction on invariant leaves: Extension of Adler-Yamilov map}
In this section, we derive an eight-dimensional Yang-Baxter map from map (\ref{NLS5D}). This is a Grassmann extension of the Adler-Yamilov map\index{Grassmann extention(s)!of Adler-Yamilov map} \cite{Adler-Yamilov, kouloukas, PT}. Our proof is motivated by the existence of the first integral (\ref{1integralNLS}).\index{first integral(s)}

In particular, we have the following

\begin{proposition}
\begin{enumerate}
	\item The quantities $\Phi =X-x_1x_2-\chi_1\chi_2$ and $\Psi=Y-y_1y_2-\psi_1\psi_2$ are invariants (first integrals) of the map\index{first integral(s)} (\ref{NLS5D}),
	\item The ten-dimensional map (\ref{NLS5D}) can be restricted to an eight-dimensional map $Y_{a,b}:A_a\times A_b \longrightarrow A_a\times A_b$, where $A_a$, $A_b$ are level sets of the first integrals\index{first integral(s)} $\Phi$ and $\Psi$, namely
\begin{subequations}\label{symleavesGras}
\begin{align}
&A_a=\{(x_1,x_2,\chi_1,\chi_2,X)\in K^5; X=a+x_1x_2+\chi_1\chi_2\}, \\
&A_b=\{(y_1,y_2,\psi_1,\psi_2,Y)\in K^5; Y=b+y_1y_2+\psi_1\psi_2\},
\end{align}
\end{subequations}
  \item The bosonic limit\index{bosonic limit} of map $Y_{a,b}$ is the Adler-Yamilov map.
\end{enumerate}
\end{proposition}

\begin{proof}
\vspace{-1cm}
\begin{enumerate}
	\item It can be readily verified that (\ref{NLS5D}) implies $U-u_1u_2-\xi_1\xi_2=X-x_1x_2-\chi_1\chi_2$ and $V-v_1v_2-\eta_1\eta_2=Y-y_1y_2-\psi_1\psi_2$. Thus, $\Phi$ and $\Psi$ are invariants, i.e.  first integrals\index{first integral(s)} of the map.
	\item The existence of the restriction is obvious. Using the conditions $X=x_1x_2+\chi_1\chi_2+a$ and $Y=y_1y_2+\psi_1\psi_2+b$, one can eliminate $X$ and $Y$ from (\ref{NLS5D}). The resulting map is $\textbf{x}\rightarrow \textbf{u}=\textbf{u}(\textbf{x},\textbf{y})$, $\textbf{y}\rightarrow \textbf{v}=\textbf{v}(\textbf{x},\textbf{y})$, where $\textbf{u}$ and $\textbf{v}$ are given by
\begin{subequations}\label{Adler-Yamilov-Grassmann}
\begin{align}
\textbf{u}&=\left(y_1+\frac{(b-a)(1+x_1y_2-\chi_1\psi_2)}{(1+x_1 y_2)^2}x_1,\,\,y_2,\,\,\psi_1+\frac{b-a}{1+x_1y_2}\chi_2,\,\,\psi_2\right),  \\
\textbf{v}&=\left(x_1,\,\,x_2+\frac{(a-b)(1+x_1y_2-\chi_1\psi_2)}{(1+x_1y_2)^2}y_2,\,\,\chi_1,\,\,\chi_2+\frac{a-b}{1+x_1y_2}\psi_2 \right).
\end{align}
\end{subequations}  

\item If one sets the odd variables\index{odd variable(s)} of the above map equal to zero, namely $\chi_1=\chi_2=0$ and $\psi_1=\psi_2=0$, then the map (\ref{Adler-Yamilov-Grassmann}) coincides with the Adler-Yamilov map.
\end{enumerate}
\end{proof}

Now, one can use the condition $X=x_1x_2+\chi_1\chi_2+a$ to eliminate $X$ from the Lax matrix (\ref{mLaxNLS}), i.e.
\begin{equation} \label{laxSNLS}
M(\textbf{x};a,\lambda)=\left(
\begin{matrix}
 a+x_1x_2+\chi_1\chi_2+\lambda & x_1 & \chi_1\\
 x_2 & 1 &0 \\
 \chi_2 & 0 & 1
\end{matrix}\right),
\end{equation}
which corresponds to the Darboux matrix derived in \cite{Georgi}.
Now, Adler-Yamilov map's extension follows from the strong Lax representation
\begin{equation} \label{lax_eq_NLS-G}
  M(\textbf{u};a,\lambda)M(\textbf{v};b,\lambda)=M(\textbf{y};b,\lambda)M(\textbf{x};a,\lambda).
\end{equation}
Therefore, the extension of the Adler-Yamilov's map (\ref{Adler-Yamilov-Grassmann}) is a reversible parametric YB map. Moreover, it is easy to verify that it is not involutive. Birationality of map (\ref{Adler-Yamilov-Grassmann}) is due to Prop. \ref{rationality-G}.

Now, from $\str(M(\textbf{y};b,\lambda)M(\textbf{x};a,\lambda))$ we obtain the following invariants for map (\ref{Adler-Yamilov-Grassmann})
\begin{eqnarray*}
&& T_1=a+b+x_1x_2+y_1y_2+\chi_1\chi_2+\psi_1\psi_2, \\
&& T_2=(a+x_1x_2+\chi_1\chi_2)(b+y_1y_2+\psi_1\psi_2)+x_1y_2+x_2y_1+\chi_1\psi_2-\chi_2\psi_1,
\end{eqnarray*}
However, these are linear combinations of the following integrals
\begin{subequations}
\begin{align}
I_1=&(a+x_1x_2+\chi_1\chi_2)(b+y_1y_2)+\psi_1\psi_2(a+x_1x_2)+x_1y_2+\nonumber\\
    &x_2y_1+\chi_1\psi_2-\chi_2\psi_1, \\
I_2=&x_1x_2+y_1y_2, \qquad I_3=\chi_1\chi_2+\psi_1\psi_2, \qquad I_4=\chi_1\chi_2\psi_1\psi_2.
\end{align}
\end{subequations}
These are in involution with respect to the Poisson bracket
\begin{eqnarray*}
&&\{x_1,x_2\}=\{y_1,y_2\}=1, \quad \{\chi_1,\chi_2\}=\{\psi_1,\psi_2\}=1 \quad \text{and all the rest} \\ && \{x_i,x_j\}=\{y_i,y_j\}=\{x_i,y_j\}=0.
\end{eqnarray*}
and the corresponding Poisson matrix is invariant under the YB map (\ref{Adler-Yamilov-Grassmann}). However, we cannot make any conclusions about the Liouville integrability of (\ref{Adler-Yamilov-Grassmann}), as $I_3$ and $I_4$ are not functionally independent (notice that $I_3^2=2I_4$).


\subsection{Derivative nonlinear Schr\"odinger equation}
According to matrix $M(p,q_{10},\theta,\phi_{10};1,1)$ in (\ref{DarbouxZ2-M}) we consider the following matrix
\begin{equation} \label{DarbouxZ2}
M(\textbf{x};\lambda)=\lambda^2 \left(
\begin{matrix}
 X & 0 & 0\\
 0 & 0 & 0\\
 0 & 0 & 0
\end{matrix}\right)+\lambda \left(
\begin{matrix}
 0 & x_1 &  \chi_1 \\
 x_2 & 0 &0 \\
 \chi_2 & 0 & 0
\end{matrix}\right)
+\left(
\begin{matrix}
 1 & 0 & 0 \\
 0 & 1 &0 \\
 0 & 0 & 1
\end{matrix}\right),
\end{equation}
where $\textbf{x}=(x_1,x_2,\chi_1,\chi_2,X)$ and, in particular, we have set
\begin{equation}
X:=f,\quad x_1:=fp,\quad x_2=fq_{10},\quad \chi_1:=f\theta \quad \text{and}\quad \chi_2:=\psi_{10}f.
\end{equation}

The Lax equation for the above matrix implies the following equations
\begin{subequations}
\begin{align}
&U+V+u_1v_2+\xi_1\eta_2=Y+X+y_1x_2+\psi_1\chi_2, \\
&Uv_i=x_iY, \quad i=1,3, \qquad Vu_i=y_iX, \quad i=2,4, \\
&u_i+v_i=x_i+y_i, \qquad i=1,\ldots, 4.
\end{align}
\end{subequations}
As in the previous section, the algebraic variety consists of two components. The first ten-dimensional component corresponds to the permutation map
\begin{equation}
\textbf{x} \mapsto \textbf{u}=\textbf{y}, \qquad \textbf{y} \mapsto \textbf{v}=\textbf{x},
\end{equation}
and the second corresponds to the following ten-dimensional YB map
\begin{subequations}\label{Z25D}
\begin{align}
& \textbf{x}\, \mapsto \, \textbf{u}=\left(y_1+\frac{f}{g}x_1,\,\,\frac{g}{h}y_2,\,\,\psi_1+\frac{f}{g}\chi_1,\,\,\frac{g}{h}\psi_2,\,\,\frac{g}{h}Y\right), \\
& \textbf{y}\, \mapsto \, \textbf{v}=\left(\frac{h}{g}x_1,\,\,x_2+\frac{f}{h}y_2,\,\,\frac{h}{g}\chi_1,\,\,\chi_2+\frac{f}{h}\psi_2,\,\,\frac{h}{f}X \right).
\end{align}
\end{subequations}
where $f=f(\textbf{x},\textbf{y})$, $g=g(\textbf{x},\textbf{y})$ and $h=h(\textbf{x},\textbf{y})$ are given by the following expressions
\begin{subequations}\label{f-g-h-Z2}
\begin{align}
\label{f-g-h-Z2-1}
&f(\textbf{x},\textbf{y})=X-x_1x_2-\chi_1\chi_2-Y+y_1y_2+\psi_1\psi_2,\\
\label{f-g-h-Z2-2}
&g(\textbf{x},\textbf{y})=X-x_1(x_2+y_2)-\chi_1(\chi_2+\psi_2),\\
\label{f-g-h-Z2-3}
&h(\textbf{x},\textbf{y})=Y-(x_1+y_1)y_2-(\chi_1+\psi_1)\psi_2.
\end{align}
\end{subequations}


\subsection{Restriction on invariant leaves}
In this section, we show that the map given by (\ref{Z25D})-(\ref{f-g-h-Z2}) can be restricted to a completely integrable eight-dimensional YB map on invariant leaves. As in the previous section, the idea of this restriction is motivated by the first integral\index{first integral(s)} (\ref{G-z2-1stInt}).

Particularly, we have the following
\begin{proposition}
\begin{enumerate}
	\item $\Phi =X-x_1x_2-\chi_1\chi_2$ and $\Psi=Y-y_1y_2-\psi_1\psi_2$ are invariants of the map (\ref{Z25D})-(\ref{f-g-h-Z2}),
	\item The ten-dimensional map (\ref{Z25D})-(\ref{f-g-h-Z2}) can be restricted to an eight-dimensional map $Y_{a,b}:A_a\times A_b \longrightarrow A_a\times A_b$, where $A_a$, $A_b$ are given by (\ref{symleavesGras}).
	\item The bosonic limit\index{bosonic limit} of the above eight-dimensional map is map (\ref{YB-affine}), corresponding to the DNLS equation.
\end{enumerate}
\end{proposition}

\begin{proof}
\vspace{-1cm}
\begin{enumerate}
	\item Map (\ref{Z25D})-(\ref{f-g-h-Z2}) implies $U-u_1u_2-\xi_1\xi_2=X-x_1x_2-\chi_1\chi_2$ and $V-v_1v_2-\eta_1\eta_2=Y-y_1y_2-\psi_1\psi_2$. Therefore, $\Phi$ and $\Psi$ are first integrals\index{first integral(s)} of the map.
	\item The conditions $X=x_1x_2+\chi_1\chi_2+a$ and $Y=y_1y_2+\psi_1\psi_2+b$ define the level sets, $A_a$ and $A_b$, of $\Phi$ and $\Psi$, respectively. Using these conditions, we can eliminate $X$ and $Y$ from map (\ref{Z25D})-(\ref{f-g-h-Z2}). The resulting map, $Y_{a,b}:A_a\times A_b \longrightarrow A_a\times A_b$, is given by 
\begin{subequations}\label{Z2-Grassmann} 
\begin{align}
 x_1\, \mapsto \, u_1&=y_1+\frac{(a-b)(a-x_1y_2+\chi_1\psi_2)}{(a-x_1 y_2)^2}x_1, \\
 x_2\, \mapsto \, u_2&=\frac{(a-x_1y_2-\chi_1\psi_2)(b-x_1y_2+\chi_1\psi_2)}{(b-x_1y_2)^2}y_2, \\
 \chi_1\, \mapsto \, \xi_1&=\psi_1+\frac{a-b}{a-x_1y_2}\chi_1, \\
 \chi_2\, \mapsto \, \xi_2&=\frac{a-x_1y_2}{b-x_1y_2}\psi_2, \displaybreak[0]\\
 y_1\, \mapsto \, v_1&=\frac{(b-x_1y_2-\chi_1\psi_2)(a-x_1y_2+\chi_1\psi_2)}{(a-x_1y_2)^2}x_1, \\
 y_2\, \mapsto \, v_2&=x_2+\frac{(b-a)(b-x_1y_2+\chi_1\psi_2)}{(b-x_1y_2)^2}y_2, \\
 \psi_1\, \mapsto \, \eta_1&=\frac{b-x_1y_2}{a-x_1y_2}\chi_1, \\
 \psi_2\, \mapsto \, \eta_2&=\chi_2+\frac{b-a}{b-x_1y_2}\psi_2. 
\end{align}
\end{subequations}
\item Setting the odd variables\index{odd variable(s)} of the above map equal to zero, namely $\chi_1=\chi_2=0$ and $\psi_1=\psi_2=0$, we obtain the YB map (\ref{YB-affine}).	
\end{enumerate}
This proves the Proposition.
\end{proof}

Now, using condition $X = x_1x_2+\chi_1\chi_2+a$, matrix (\ref{DarbouxZ2}) takes the following form
\begin{equation} \label{DarbouxLaxZ2}
M=\lambda^2 \left(
\begin{matrix}
 k+x_1x_2+\chi_1\chi_2 & 0 & 0\\
 0 & 0 & 0\\
 0 & 0 & 0
\end{matrix}\right)+\lambda \left(
\begin{matrix}
 0 & x_1 & \chi_1 \\
 x_2 & 0 &0 \\
 \chi_2 & 0 & 0
\end{matrix}\right)
+\left(
\begin{matrix}
 1 & 0 & 0 \\
 0 & 1 &0 \\
 0 & 0 & 1
\end{matrix}\right).
\end{equation}
Map (\ref{Z2-Grassmann}) has the following Lax representation
\begin{equation} \label{lax_eq_Z2}
  M(\textbf{u};a,\lambda)M(\textbf{v};b,\lambda)=M(\textbf{y};b,\lambda)M(\textbf{x};a,\lambda).
\end{equation}
Therefore, it is reversible parametric YB map which is birational due to Prop. \ref{rationality-G}. It can also be verified that it is not involutive.

For the integrability of map (\ref{Z2-Grassmann}) we have the following

\begin{proposition}
Map (\ref{Z2-Grassmann}) is completely integrable.
\end{proposition}

\begin{proof}
The invariants of map (\ref{Z2-Grassmann}) which we retrieve from $\str(M(\textbf{y};b,\lambda)M(\textbf{x};a,\lambda))$ are
\begin{eqnarray*}
&&K_1=(a+x_1x_2+\chi_1\chi_2)(b+y_1y_2+\psi_1\psi_2) \\
&&K_2=a+b+x_1x_2+y_1y_2+x_1y_2+x_2y_1+\chi_1\chi_2+\psi_1\psi_2+\chi_1\psi_2-\chi_2\psi_1,
\end{eqnarray*}
where the constant terms can be omitted. However, $K_1$ is sum of the following quantities 
\begin{subequations}
\begin{align}
I_1&=(a+x_1x_2)(y_1y_2+\psi_1\psi_2)+b(x_1x_2+\chi_1\chi_2)+y_1y_2\chi_1\chi_2\\
I_2&=\chi_1\chi_2\psi_1\psi_2,
\end{align}
\end{subequations}
which are invariants themselves. Moreover, $K_2$ is sum of the following invariants
\begin{equation}
I_3=(x_1+y_1)(x_2+y_2) \qquad \text{and} \qquad I_4=(\chi_1+\psi_1)(\chi_2+\psi_2).
\end{equation}
In fact, the quantities $C_i=x_i+y_i$ are invariants theirselves.

We construct a Poisson matrix, $J$, such that the following equation is satisfied
\begin{equation}\label{casEq}
\nabla C_i\cdot  J=0.
\end{equation}
In this case the Poisson bracket is given by
\begin{eqnarray*}
&\left\{x_1,x_2\right\}=\left\{x_2,y_1\right\}=\left\{y_2,x_1\right\}=\left\{y_1,y_2\right\}=1,&\\
&\left\{\chi_1,\chi_2\right\}=\left\{\psi_1,\psi_2\right\}=1 \quad \text{and} \quad \left\{\chi_1,\psi_2\right\}=\left\{\psi_1,\chi_2\right\}=-1.&
\end{eqnarray*}
Map (\ref{Z2-Grassmann}) preserves the Poisson bracket. Moreover, due to (\ref{casEq}), $C_i$'s are Casimir functions as
\begin{equation}
\left\{C_i,f\right\}=(\nabla C_i)\cdot J\cdot (\nabla f)^t=0,\quad \text{for any} \quad f=f(\textbf{x},\textbf{y}).
\end{equation}
Moreover, the invariants $I_1$ and $I_2$ are in involution with respect to this Poisson matrix, namely $\nabla I_1\cdot J \cdot (\nabla I_2)^t=0$. The rank of the Poisson matrix is 4 and $C_i$, $i=1,2,3,4$, are four Casimir functions. Therefore, the eight-dimensional map (\ref{Z2-Grassmann}) is completely integrable.
\end{proof}

\section{Vector generalisations: $4N\times4N$ Yang-Baxter maps}
In what follows we use the following notation for a $n-$vector $\textbf{w}=(w_1,...,w_n)$
\begin{eqnarray}
&& \textbf{w}=(\textbf{w}_1,\textbf{w}_2,\textbf{$\boldsymbol\omega$}_1,\textbf{$\boldsymbol\omega$}_2),\qquad \text{where} \qquad \textbf{w}_1=(w_1,...,w_N), \qquad \textbf{w}_2=(w_{N+1},...,w_{2N})\nonumber \\
&& \text{and} \qquad \textbf{$\boldsymbol\omega$}_1=(\omega_{2N+1},...,\omega_{3N}), \qquad \textbf{$\boldsymbol\omega$}_2=(\omega_{3N+1},...,\omega_{4N}), \nonumber
\end{eqnarray}
where $\textbf{w}_1$ and $\textbf{w}_2$ are even and $\textbf{$\boldsymbol\omega$}_1$ and $\textbf{$\boldsymbol\omega$}_2$ are odds.
Also,
\begin{equation}
\langle u_i|:=\textbf{u}_i,\qquad |w_i\rangle:=\textbf{w}_i^T \qquad \text{and their dot product with}\qquad \langle u_i,w_i\rangle.
\end{equation}

\subsection{Nonlinear Schr\"odinger equation}
Now, we replace the variables in map (\ref{Adler-Yamilov-Grassmann}) with $N-$vectors, namely we consider the following $4N\times 4N$ map 
\begin{subequations}\label{VectorNLS-G}
\begin{align}
\begin{cases}
\langle u_1|=\langle y_1|+f(z;a,b)\langle x_1|(1+\langle x_1,y_2\rangle-\langle \chi_1,\psi_2\rangle),\\
\langle u_2|=\langle y_2|, \\
\langle \xi_1|=\langle \psi_1|+f(z;a,b)\langle \chi_1|(1+\langle x_1,y_2\rangle-\langle \chi_1,\psi_2\rangle), \\
\langle \xi_2|=\langle \psi_2|,
\end{cases}
\intertext{and}
\begin{cases}
\langle v_1|=\langle x_1|,\\
\langle v_2|=\langle x_2|+f(z;b,a)\langle y_2|(1+\langle x_1,y_2\rangle-\langle \chi_1,\psi_2\rangle),\\
\langle \eta_1|=\langle \chi_1| \\
\langle \eta_2|=\langle \chi_2|+f(z;b,a)\langle y_4|(1+\langle x_1,y_2\rangle-\langle \chi_1,\psi_2\rangle)
\end{cases}
\end{align}
\end{subequations}
where $f$ is given by
\begin{equation}\label{f-g-h-VNLS}
f(z;b,a)=\frac{b-a}{(1+z)^2},\qquad z:=\langle x_1,y_2\rangle.
\end{equation}

Map (\ref{VectorNLS-G})-(\ref{f-g-h-VNLS}) is a reversible parametric YB map, for it has the following strong Lax-representation
\begin{equation}\label{LaxRep}
M(\textbf{u};a)M(\textbf{v};b)=M(\textbf{y};b)M(\textbf{x};a)
\end{equation}
where
\begin{equation}
M(\textbf{w};a)=\left(     
\begin{matrix}
\lambda+a+\langle w_1,w_2\rangle+\langle \omega_1,\omega_2\rangle  & \langle w_1| & \langle \omega_1| \\
|w_2\rangle        \\
|\omega_2\rangle      & \quad \qquad \qquad I_{2N-1} &
\end{matrix}\right).
\end{equation}
Moreover, map (\ref{VectorNLS-G})-(\ref{f-g-h-VNLS}) is birational and not involutive. 

The invariants of this map are given by
\begin{subequations}
\begin{align}
 I_1&=a+b+\langle x_1,x_2\rangle+\langle y_1,y_2\rangle+\langle \chi_1,\chi_2\rangle+\langle \psi_1,\psi_2\rangle, \\
 I_2&=b(\langle x_1,x_2\rangle+\langle \chi_1,\chi_2\rangle)+a(\langle y_1,y_2\rangle+\langle \psi_1,\psi_2\rangle)+(\langle y_1,y_2\rangle+\\
    & \quad \langle \psi_1,\psi_2\rangle)(\langle x_1,x_2\rangle+\langle \chi_1,\chi_2\rangle).\nonumber
\end{align}
\end{subequations}
However, the number of the invariants we obtain from the supertrace\index{supertrace} of the monodromy matrix\index{monodromy matrix} is not enough to claim integrability in the Liouville sense.

\subsection{Derivative nonlinear Schr\"odinger equation}
Now, replacing the variables in (\ref{Z2-Grassmann}) with $N$-vectors we obtain the following $4N\times 4N$-dimensional map
\begin{subequations}\label{vectorZ2-G}
\begin{align}
\begin{cases}
\langle u_1|&=\langle y_1|+h(z;a,b)\langle x_1|(a-\langle x_1,y_2\rangle+\langle \chi_1,\psi_2\rangle), \\
\langle u_2|&=g(z;a,b)\langle y_2|-h(z;b,a)\langle \chi_1,\psi_2\rangle \langle y_2|,\\
\langle \xi_1|&=\langle \psi_1|+f(z;a,b)\langle \chi_1|,\\
\langle \xi_2|&=g(z;a,b)\langle \psi_2|,
\end{cases}
\intertext{and}
\begin{cases}
\langle v_1|&=g(z;b,a)\langle x_1|-h(z;a,b)\langle x_1|,\\
\langle v_2|&=\langle x_2|+h(z;b,a)\langle y_2|(b-\langle x_1,y_2\rangle+\langle \chi_1,\psi_2\rangle),\\
\langle \eta_1|&=g(z;b,a)\langle \chi_1|,\\
\langle \eta_2|&=\langle \chi_2|+f(z;b,a)\langle \psi_2|,
\end{cases}
\end{align}
\end{subequations}
where $f$, $g$ and $h$ are given by
\begin{equation}\label{f-g-h-VZ2}
f(z;a,b)=\frac{a-b}{a-z},\quad g(z;a,b)=\frac{a-z}{b-z},\quad h(z;a,b)=\frac{a-b}{(a-z)^2}, \quad z:=\langle x_1,y_2\rangle.
\end{equation}

Map (\ref{vectorZ2-G})-(\ref{f-g-h-VZ2}) is reversible parametric YB map, as it has the strong Lax-representation (\ref{LaxRep}) where
\begin{equation} \label{DarbouxVectorLaxZ2}
M= \left(
\begin{matrix}
 \lambda^2(k+\langle x_1,x_2\rangle+\langle \chi_1,\chi_2\rangle) & \lambda \langle x_1|  & \lambda \langle \chi_1|\\
 \lambda |x_2\rangle &  & \\
 \lambda |\chi_2\rangle & \qquad \quad I_{2N} & 
\end{matrix}\right).
\end{equation}
Moreover, it is a non-involutive map and birational.

The invariants of the map that we retrieve from the supertrace\index{supertrace} of the monodromy matrix\index{monodromy matrix} are given by
\begin{eqnarray*}
&&K_1=(a+\langle x_1,x_2\rangle+\langle\chi_1,\chi_2\rangle)(b+\langle y_1,y_2\rangle+\langle\psi_1,\psi_2\rangle) \\
&&K_2=a+b+\langle x_1,x_2\rangle+\langle y_1,y_2\rangle+\langle x_1,y_2\rangle+\langle x_2,y_1\rangle+\langle\chi_1,\chi_2\rangle+\langle\psi_1,\psi_2\rangle+\\
&&\qquad~~\langle\chi_1,\psi_2\rangle-\langle\chi_2\psi_1\rangle.
\end{eqnarray*}
In fact, $K_2$ is a sum of the following invariants
\begin{equation}
I_1=\langle x_1+y_1,x_2+y_2\rangle,\qquad I_2=\langle\chi_1+\chi_2,\psi_1+\psi_2 \rangle,
\end{equation}
and the entries of the above dot products are invariant vectors. However, as in the case of the NLS vector generalisation in the previous section, the invariants are not enought to claim Liouville integrability.

\chapter{Conclusions}
\label{chap7} \setcounter{equation}{0}
\renewcommand{\theequation}{\thechapter.\arabic{equation}}

\section{Summary of results}
In this thesis we used Darboux transformations \index{Darboux transformation(s)} as the main tool to first link integrable partial differential equations to discrete integrable systems \index{integrable system(s)!discrete} and, then, to construct parametric Yang-Baxter maps.

In particular, we constructed Darboux matrices \index{Darboux matrix(-ces)} for certain Lax operators of NLS type\index{Lax operator(s)}\index{Lax operator(s)!for NLS type operators} and employed them in the derivation of integrable systems of difference equations.\index{integrable system(s)!of difference equations} The advantage of this approach is that it provides us not only with difference equations,\index{difference equation(s)} but also with their Lax pairs, symmetries, first integrals\index{first integral(s)} and conservation laws;\index{conservation law(s)} even with B\"acklund transformations\index{B\"acklund transformation(s)} for the original partial differential equations they are related to. Then, using first integrals in a systematic way we were able to reduce some of our systems to integrable equations of Toda type.\index{Toda (type) equation(s)}

In fact, we studied the cases of the NLS equation, the DNLS equation and a deformation of the DNLS equation. These equations were not randomly chosen, but their associated Lax operators possess certain symmetries, namely they are invariant under action of (reduction) groups of transformation which correspond to some classification results on automorphic Lie algebras.\index{automorphic Lie algebras}

More precisely, we derived novel systems of difference equations, \index{difference equation(s)} namely systems (\ref{nls-comp}), (\ref{sl2-res-eq}) and (\ref{DS-dDNLS}), which are actually systems with vertex and bond variables \cite{HV}. These systems have symmetries and first integrals which follow from the derivation of the Darboux matrix and they are integrable in the sense that they possess Lax pair. Additionally, they are multidimensionally consistent. \index{multidimensionally consistent} 

On the other hand, motivated by the similarity of the Bianchi-type compatibility condition with the Lax equation for Yang-Baxter maps, we used the afore-mentioned Darboux matrices to construct Yang-Baxter maps.\index{Yang-Baxter (YB) map(s)} Specifically, we constructed ten-dimensional Yang-Baxter maps as solutions of matrix refactorisation problems \index{refactorisation problem} related to Darboux matrices for all the NLS type equations we mentioned earlier.

Motivated by the fact that the potential-entries of the Darboux matrices obey systems of differential-difference equations which admit certain first integrals, we used the latter to restrict our ten-dimensional maps to four-dimensional Yang-Baxter maps on invariant leaves. Particularly, in the case of NLS equation we derived the Adler-Yamilov map, while in the case of DNLS equation we derived a novel Yang-Baxter map, namely map (\ref{YB-affine}). Moreover, we showed that these Yang-Baxter maps are completely integrable and we considered their vector generalisations.

Finally, following \cite{Georgi} where the noncommutative extension of the Darboux transformation for the NLS operator was constructed, we derived the Grassmann extension of the Darboux matrix for the DNLS operator. We employed these Darboux matrices to construct the noncommutative extensions of the Adler-Yamilov map and map (\ref{YB-affine}), namely maps (\ref{Adler-Yamilov-Grassmann}) and (\ref{Z2-Grassmann}). Moreover, we considered the vector generalisations of these maps.

\section{Future work}
The goal of a PhD thesis is not only to solve problems but also to create new ones. Therefore, we list some open problems for future work epigrammatically, and we analyse them.
\begin{enumerate}
	\item Study the integrability of the transfer maps for the Yang-Baxter maps corresponding to NLS type equations;
	\item Study the corresponding entwining Yang-Baxter maps;
	\item Examine the possibility of deriving auto-B\"acklund transformations from Yang-baxter maps for the associated partial differential equations;
	\item Examine the possibility of deriving hetero-B\"acklund transformations from entwining Yang-baxter maps;
	\item Derive auto-B\"acklund and hetero-B\"acklund transformations related to the noncommutative extensions of the NLS and the DNLS equations;
	\item Examine the relation between YB maps which have as Lax representations gauge equivalent Darboux matrices\index{Darboux matrix(-ces)!gauge equivalent}.
\end{enumerate}

Regarding \textbf{1.} one could consider the transfer maps, $T_n$, as in \cite{kouloukasEnt}, which arise out of the consideration of the initial value problem on the staircase, as in Figure \ref{staircase}. For their integrability we need to use the monodromy matrix (\ref{monMat}). However, it is obvious that for $n>2$ the expressions derived from the trace of the generating function (\ref{monMat}) will be of big length and, therefore, not quite useful for (Liouville) integrability claims. Moreover, the trace of (\ref{monMat}) does not guarantee that the derived invariants are functionally independent. Thus, the most fruitful approach would be the discovery of a universal generating function of invariants for all the maps $T_n$, $n\in \field{N}$. That may demand that we have to define the transfer maps $T_n$ in a different way.

Concerning \textbf{2.} the idea is to consider a matrix refactorisation problem of the form (\ref{eqLax}) using two different Darboux matrices. For instance, having in our disposal Darboux matrices $M_1=M_1(x;a,\lambda)$ and $M_2=M_2(x;a,\lambda)$, study the solutions of the following problem
\begin{equation}\label{EntRef}
M_1(u;a,\lambda)M_2(v;b,\lambda)=M_2(y;b,\lambda)M_1(x;a,\lambda).
\end{equation}

For \textbf{3.} recall that the Yang-Baxter maps in chapter 5 were derived as solutions of matrix refactorisation problems of particular Darboux matrices. Now, one needs to take into account that the entries of these Darboux matrices are potentials satisfying certain partial differential equations. The Yang-Baxter map does not preserve these solutions. Imposing this as a condition, namely that the Yang-Baxter map maps a solution of a particular PDE to another solution of the same PDE, we obtain some relations among these solutions. 

With regards to \textbf{4.} the idea is to impose that an entwning Yang-Baxter map preserves the the solutions of the associated partial differential equations. Then, check if the resulted relations constitute hetero-B\"acklund transformations for these PDEs.

Concerning \textbf{5.} the idea is the same with the one mentioned in \textbf{3.} and \textbf{4.} but regarding the noncommutative extensions of the Yang-Baxter maps presented in chapter 6.

Finally, regarding \textbf{6.} one should find how YB maps with gauge equivalent Darboux-Lax representations are related, and use this information to classify them.

\addcontentsline{toc}{chapter}{Appendix}

\chapter*{Appendices\markboth{Appendices}{}}

\renewcommand{\thesection}{\Alph{section}}

\renewcommand{\theequation}{\Alph{section}.\arabic{equation}}

\renewcommand{\thefigure}{\Alph{section}.\arabic{figure}}

\appendix


\section{Solution of the system of discrete equations associated to the deformation of the DNLS equation}\label{appendix1}
The solution of the system (\ref{DS-dDNLS}) consists of two branches: The trivial solution given by
\begin{equation}
p_{01}=p_{10}, \quad q_{01}=q_{01},\quad v=g,\quad g_{01}=v_{10},
\end{equation}
and a non-trivial given by
\begin{equation}
p_{01}=\frac{1}{A}\mathcal{F}_1,\quad q_{01}=\frac{1}{B}\mathcal{F}_2,\quad v=\frac{1}{B}\mathcal{F}_3,\quad g_{01}=\frac{1}{A}\mathcal{F}_4,
\end{equation}
where $A$ and $B$ are given by the following expressions
\begin{subequations}
\begin{align}
A=&(g (p (q_{11}-1)+p_{10}-v_{10})+p^2 (q_{11}-1)+p q_{10} (v_{10}-p_{10})-\nonumber\\
&(q_{10}+1) (q_{11}-1)) (g (p_{10}-p (q_{11}+1)+v_{10})+p^2 (q_{11}+1)-\\
&p q_{10} (p_{10}+v_{10})+(q_{10}-1) (q_{11}+1)),\nonumber
\intertext{and}
B=&(q_{11} (q_{11}-p_{10} (g+q_{10}))+v_{10} (g+q_{10}-q_{11})+p (p_{10}-q_{11} v_{10}+\nonumber\\
&q_{11}^2-1)+p_{10}-1) (p_{10} (-g q_{11}+q_{10} q_{11}-1)+v_{10} (g-q_{10}+q_{11})+\\
&p (p_{10}-q_{11} (q_{11}+v_{10})+1)+q_{11}^2-1),\nonumber
\end{align}
\end{subequations}
whereas $\mathcal{F}_i$, $i=1,\ldots,4$, are given by
\small
\begin{subequations}
\begin{align}
\mathcal{F}_1=&g^3 (p_{10}-p) (p_{10} q_{11}-v_{10})+g^2 (p^2 (p_{10}-q_{11} v_{10})-p (p_{10}^2 (q_{10} q_{11}+1)-p_{10} (q_{10}+\nonumber\\
&2 q_{11}) v_{10}-q_{11}^2+v_{10}^2+1)+p_{10} q_{11} (q_{10}-2 q_{11})+p_{10}+(q_{11}-q_{10}) v_{10})+\nonumber\\
&g (q_{10} (p^2 (v_{10}^2-2 p_{10} q_{11} v_{10}+p_{10}^2-q_{11}^2+1)-2 p_{10} q_{11} v_{10}+4 p_{10} p (q_{11}^2-1)+\nonumber\\
&p_{10}^2-q_{11}^2+v_{10}^2+1)+(p^2-1) (q_{11} (p p_{10}+v_{10}^2+1)-(p+p_{10}) v_{10}-q_{11}^3)+\nonumber\\
&(p-p_{10}) q_{10}^2 (p_{10} q_{11}-v_{10}))+((p^2-1) q_{11}+q_{10}) (p^2 v_{10}+p q_{10} (q_{11}^2-v_{10}^2-1)+\nonumber\\
&(q_{10}^2-1) v_{10})+p_{10} (p^3 q_{10} v_{10}-p^2 (p^4+(q_{10}^2-1) (q_{10} q_{11}-1)-q_{10} (q_{10}-\nonumber\\
&2 q_{10} q_{11}^2-q_{11})-2)-p q_{10} (q_{10}^2-2 q_{11} q_{10}+1) v_{10})+p p_{10}^2 q_{10}^2 (q_{10} q_{11}-1),\nonumber\\
\nonumber\\
\mathcal{F}_2=&g^2 (p_{10} q_{11}-v_{10}) (p (q_{11}^2-1)+p_{10}-q_{11} v_{10})-g (p_{10}^2 (q_{10}-q_{11}) v_{10}+\nonumber\\
&p_{10} (q_{10} q_{11}-1) v_{10}^2+p (p_{10}^2 (q_{11}^2+2 q_{10} q_{11}-1)-2 p_{10} q_{10} (q_{11}^2+1) v_{10}+\nonumber\\
&(1-q_{11}^2+2 q_{10} q_{11}) v_{10}^2+(q_{11}^2-1)^2)+p_{10}^3 (1-q_{10} q_{11})+p_{10} (q_{10} q_{11}+1) (q_{11}^2-1)+\nonumber\\
&(q_{11}-q_{10}) v_{10}^3-(q_{10}+q_{11}) (q_{11}^2-1) v_{10})+(1-p^2) (p_{10} q_{11}-v_{10}) (p (q_{11}^2-1)+\nonumber\\
&p_{10}-q_{11} v_{10})+q_{10} (2 p^2 (v_{10}-p_{10} q_{11})^2+4 p_{10} q_{11} v_{10}+p (p_{10}-q_{11} v_{10}) (p_{10}^2+\nonumber\\
&q_{11}^2-v_{10}^2-1)-p_{10}^2 (q_{11}^2+1)-(q_{11}^2+1) v_{10}^2+(q_{11}^2-1)^2)+\nonumber\\
&q_{10}^2 (p_{10} q_{11}-v_{10}) (1-v_{10} (q_{11}-p v_{10})-p p_{10}),\nonumber\\
\nonumber\\
\mathcal{F}_3=&g^2 (p_{10} q_{11}-v_{10}) (p q_{11} v_{10}+(1-p) p_{10}^2-v_{10}^2)+g (p^2 (p_{10}^2 (q_{11}^2+1)-\nonumber\\
&4 p_{10} q_{11} v_{10}+q_{11}^2 (v_{10}^2+2)-q_{11}^4+v_{10}^2-1)+p (p_{10}^2 v_{10} (q_{11} (1-p_{10} q_{10})-\nonumber\\
&p_{10}+q_{10})+p_{10} (q_{10} q_{11} (q_{11}^2+v_{10}^2-1)-q_{11}^2+v_{10}^2+1)-(q_{10}+q_{11}) v_{10}^3+\nonumber\\
&(q_{10}-q_{11}) (1-q_{11}^2) v_{10})+2 (p_{10} q_{11}-v_{10}) (p_{10} (q_{10}-q_{11})-q_{10} q_{11} v_{10}+v_{10}))-\nonumber\\
&(p (p_{10}-q_{11} v_{10})+q_{11}^2-1) (p^2 (p_{10} q_{11}-v_{10})-p q_{10} (p_{10}^2+q_{11}^2-v_{10}^2-1)+\nonumber\\
&(q_{10}^2-1) (p_{10} q_{11}-v_{10})),\nonumber\\
\nonumber\\
\mathcal{F}_4=&g^3 (v_{10}-p q_{11}) (p_{10} q_{11}-v_{10})-g^2 (v_{10} (q_{11} ((p^2+1) q_{11}+q_{10})+p p_{10} (q_{10} q_{11}+2)-2)-\nonumber\\
&p (q_{10}+q_{11}) v_{10}^2-q_{11} (p q_{11}^2+p_{10} (q_{10} q_{11}-1)+p (p_{10} (p+p_{10})-1)))+\nonumber\\
&g (p^3 q_{11} (p_{10} q_{11}-v_{10})+p^2 (p_{10} (2 q_{10}+q_{11}) v_{10}+(q_{10} q_{11}+1) (1-p_{10}^2-q_{11}^2)-\nonumber\\
&q_{10} q_{11} v_{10}^2)+p (q_{11}-q_{10} (q_{11} (q_{10}-4 q_{11})+4)) v_{10}-p_{10} (q_{11}+q_{10} (q_{10} q_{11}-2)) v_{10}+\nonumber\\
&p_{10} p (q_{10}^2-1) q_{11}^2-(q_{10} q_{11}-1) (p_{10}^2+q_{11}^2-1)+q_{10} (q_{10}-q_{11}) v_{10}^2)+\nonumber\\
&v_{10} (q_{11} q_{10} (p (p_{10}(1-p^2)+p_{10} q_{10}^2+p)+q_{10}^2-1)+q_{11}^2 ((1-p^2)^2-(p^2+1) q_{10}^2)+\nonumber\\
&2 p (p-p_{10}) q_{10}^2)+(1-p^2-q_{10} q_{11}) (p_{10} p^2 q_{11}-p q_{10} (p_{10}^2+q_{11}^2-1)+\nonumber\\
&p_{10} (q_{10}^2-1) q_{11})+p q_{10}^2 (q_{11}-q_{10}) v_{10}^2.\nonumber
\end{align}
\end{subequations}

\include{index}


\addcontentsline{toc}{chapter}{Bibliography}
\bibliography{refs}
\cleardoublepage
\addcontentsline{toc}{chapter}{Index}
\printindex
\end{document}